%% file: main.tex
\documentclass[11pt]{article}

\usepackage[margin=1in]{geometry}
\usepackage{amsmath, amssymb, amsthm}
\usepackage{newtxtext, newtxmath}
\usepackage{bm}
\usepackage{setspace}
\usepackage{algorithm}
\usepackage{algpseudocode}
\usepackage{tikz}
\usepackage{authblk}

\input{preamble}

\bibpunct[, ]{(}{)}{,}{a}{}{,}

\DeclareMathOperator*{\argmax}{arg\,max}
\DeclareMathOperator*{\argmin}{arg\,min}

\theoremstyle{plain}
\newtheorem{theorem}{Theorem}
\newtheorem{lemma}{Lemma}
\newtheorem{proposition}{Proposition}
\newtheorem{corollary}{Corollary}
\theoremstyle{definition}
\newtheorem{assumption}{Assumption}
\newtheorem{definition}{Definition}
\newtheorem{example}{Example}
\theoremstyle{remark}

\title{Deployment of AI-Assisted Interventions:\\
Capacity Constraints and Noisy Compliance}

\author[1]{Carri W. Chan}
\author[2]{Yi Han}
\author[1]{Hannah Li}
\author[3]{Benjamin L. Ranard}
\affil[1]{Decision, Risk, and Operations, Columbia Business School}
\affil[2]{Department of Statistics, Columbia University}
\affil[3]{Division of Pulmonary, Allergy, and Critical Care Medicine, Department of Medicine, Vagelos College of Physicians and Surgeons, and Department of Biomedical Informatics, Columbia University}

\date{}

\begin{document}

\maketitle

\begin{abstract}
\input{section/abstract}
\end{abstract}


\input{section/intro_v2}
\input{section/related_work}
\input{section/model}
\input{section/thresholding}
\input{section/algo_selection}
\input{section/case_study}
\input{section/extension}
\input{section/discussion}

\bibliographystyle{informs2014}
\bibliography{bib}

\newpage
\appendix
\input{section/appendix_theory}
\input{section/appendix_bilevel}
\input{section/appendix_result}
\input{section/appendix_case_study}

\end{document}

%% file: preamble.tex
\usepackage{amsmath, amssymb}
\usepackage{subcaption}
\usepackage{bm}
\usepackage{graphicx}
\usepackage{booktabs}
\usepackage{natbib}
\usepackage{hyperref}
\usepackage{cleveref}
\usepackage{enumerate}
\usepackage{diagbox}
\crefname{assumption}{assumption}{assumptions}
\Crefname{assumption}{Assumption}{Assumptions}
\hypersetup{colorlinks=true,linkcolor=blue,citecolor=blue,urlcolor=blue}
\usepackage[dvipsnames]{xcolor}

\newcommand{\obj}{W}
\newcommand{\objfluid}{\tilde{W}}
\newcommand{\nservedfluid}{\tilde{N}}
\newcommand{\operationalmetricfull}{Operational AUC}
\newcommand{\operationalmetric}{OpAUC}

%% file: section/abstract.tex
AI tools increasingly drive targeted interventions in various service settings, including healthcare, education, and public services. Algorithms score individuals, trigger outreach to those above a threshold (e.g., high-risk or high-value), and encourage them to request service; then providers deliver service to those who request. Much of the work in this area has focused on improving predictive accuracy, implicitly assuming that better predictions lead to better outcomes. We show that predictions are only one component of a larger service system: when service capacity is limited and behavioral responses to outreach are probabilistic, system efficacy depends on operational forces that predictive accuracy does not capture. In such settings, the optimal score threshold must balance two effects: ensuring all capacity is filled (\textit{utilization}) and, when capacity is constrained, preventing low-value requests from crowding out high-value ones (\textit{cannibalization}). We characterize the optimal threshold and prove that thresholds based solely on predictive accuracy are generally suboptimal. Further, algorithm selection metrics such as AUC can be misaligned with operational performance: they weight all thresholds equally, while optimal deployment uses a subset of thresholds that depends on both capacity and compliance behavior. We introduce a new metric, Operational AUC (OpAUC), and show that it identifies the efficacy-optimal algorithm. Finally, we conduct a case study on sepsis early warning data that illustrates the magnitude of the improvement available from better threshold selection and shows that a predictor with lower AUC can achieve higher system efficacy under optimal deployment.

%% file: section/intro_v2.tex
\section{Introduction}\label{sec:intro}

Service delivery systems are increasingly integrating AI tools in order to target the delivery of scarce resources and interventions.
In healthcare, hospitals use risk models that monitor patient features in order to identify the most \textit{at-risk} individuals. Then, the hospital either nudges the flagged individuals to seek follow-up screening \citep{park2024machine} or nudges their providers to respond quickly \citep{escobar2020automated, adams2022prospective}. In education, algorithms are used to target advisor outreach to the most at-risk students \citep{schechtman2025discretion}
or students at risk of failing to renew financial aid \citep{athey2025machine}.
In government operations, risk-based algorithms prioritize cases in public services such as child-welfare hotline screening \citep{chouldechova2018case} and the allocation of scarce housing assistance \citep{kube2022just, rahmattalabi2022learning,showkat2023right}.
These settings share the structure we refer to as \emph{AI-assisted interventions}: the algorithm makes a prediction, humans are alerted and encouraged to act on it, and capacity constraints, along with behavioral responses to outreach, dictate who eventually gets served.

Existing approaches have focused on identifying the individuals who would benefit the most from the service through predictive modeling, assuming that better predictions correspond to better system performance.
As such, the field has paid substantial attention to predictive performance, including accuracy and calibration \citep{rajkomar2018scalable, tomavsev2019clinically}, while the deployed algorithm and flagging threshold (i.e., the quantile cutoff for intervention) are chosen without regard to the downstream service system \citep{elkan2001foundations, fawcett2006introduction, tomavsev2019clinically}.

From a system-design perspective, however, the objective of AI-assisted intervention is to maximize system efficacy, which we define as the number of true positives served (for binary outcomes) or the sum of individual-level values across those served (for continuous outcomes).
In the healthcare screening example above, system efficacy is the number of patients who truly need follow-up care and receive it; a false positive is a flagged patient who does not need care, yet may nonetheless request service and occupy a screening slot.
These two objectives need not coincide: accurate predictions alone need not translate into better outcomes \citep{perdomo2025difficult}, because service delivery is shaped by operational forces that predictive accuracy does not capture.

Two such forces influence who ultimately receives service: \emph{noisy compliance} and \emph{service capacity}.
In many settings, compliance is inherently noisy: in the flagged group, patients may ignore reminders, advisors may skip follow-ups, and applicants may fail to submit applications even when nudged \citep{athey2025machine}; in the unflagged group, patients may request follow-up care, students may seek advising \citep{schechtman2025discretion}, and individuals may apply for public services independently \citep{zhou2023optimal}. Moreover, service capacity is often limited: when demand exceeds capacity, some requests go unserved or are served with delay \citep{mann2024potential, mann2024negative, chang2025rapid, park2024machine, boutilier2025operationaldosageimplicationscapacity, schechtman2025discretion}.

We study the optimal design of AI-assisted interventions under noisy compliance and capacity constraints at the deployment stage. In such systems, the path from algorithmic flagging to service delivery forms an operational funnel (\Cref{fig:funnel}). First, algorithmic alerts nudge individuals to request service\footnote{A similar dynamic arises when algorithmic alerts nudge \textit{providers} to provide service to flagged individuals. The provider can have probabilistic outreach to both flagged and unflagged individuals, and may have a limited capacity to serve.}. Compliance with these nudges is imperfect: a nudge raises an individual's likelihood of requesting service, and some unflagged individuals request service independently. As a result, both unflagged and flagged individuals can submit requests---flagged individuals at a higher rate due to the nudge---thus introducing competition for limited capacity.
In practice, many systems allocate slots among requesters at random when demand exceeds capacity: slots may be filled on a first-come, first-served basis, score-based prioritization at the allocation stage may be infeasible, or fairness concerns may preclude it \citep{kube2022just, moon2026paradox}. We therefore assume random allocation as our baseline.
This funnel gives rise to two competing forces that any deployment must account for: \textit{underutilization} and \textit{cannibalization}. The optimal flagging threshold must balance both, and existing methods focusing on predictive accuracy alone fail to do so.

\begin{figure}[H]
    \centering
    \includegraphics[width=\linewidth]{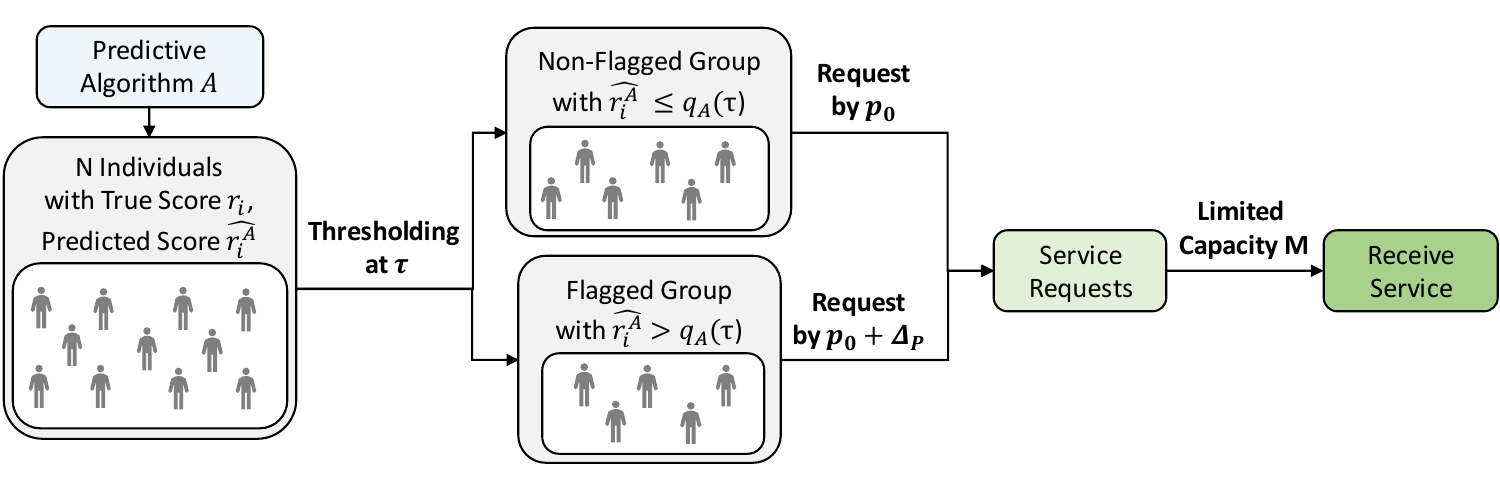}
    \caption{AI-assisted intervention with noisy compliance and capacity constraints.}
    \label{fig:funnel}
\end{figure}

To build intuition for how the optimal threshold depends on the operational forces, consider a concrete example. Suppose 100 individuals each request service with baseline probability $0.1$, and flagging raises an individual's request probability to $0.6$. Assume true scores are uniformly distributed on $[0,1]$ and the algorithm perfectly ranks individuals. Suppose a prediction-based threshold flags the top 40\% based on some exogenous assessment of the risks of false positives and negatives, generating 30 expected requests. If the system can serve 40 requests, 10 slots go unfilled. The threshold is too conservative and capacity is wasted. This leads to \textit{underutilization}, which occurs when too few individuals are flagged to generate enough requests and results in some capacity that sits idle. The optimal policy flags the top 60\%, generating exactly 40 requests to fill capacity.
Now suppose the system can serve only 10 requests. Flagging the top 40\% generates 30 requests, exceeding capacity. This leads to \textit{cannibalization}, which occurs when too many individuals are flagged and requests exceed capacity.
If service is allocated randomly among requests, lower-value requests can crowd out higher-value flagged ones. The optimal policy flags fewer individuals (top 30\%), reducing total requests while concentrating capacity on higher-value flagged individuals. Together, the two cases show how the optimal threshold responds to these forces: underutilization expands outreach to the top 60\% when capacity is abundant, and cannibalization narrows it to the top 30\% when capacity is scarce. The fixed prediction-based threshold (top 40\%) is thus suboptimal in both cases.

We focus on two central questions for the deployment of AI-assisted interventions in service delivery.
\begin{enumerate}
    \item Who should be targeted for outreach, i.e., how do we set the classification threshold as a function of the operational regime?
    \item When multiple algorithms are available, which one maximizes system efficacy under optimal deployment?
\end{enumerate}

To answer these questions, we develop an analytical model (\S\ref{sec:model}) that captures algorithmically-guided nudges, probabilistic behavior of requesting (or providing) service both with and without a nudge, and limited service capacity. These components are common across the healthcare, education, and public-service interventions referenced above.

The main contributions are outlined below.

\textbf{Characterization of the optimal threshold (\S\ref{sec:thresholding}).} For a fixed prediction model, we characterize in closed form the optimal flagging threshold that accounts for both (1) the \textit{underutilization effect}, which captures the need to fully utilize capacity by expanding outreach when capacity is abundant, and (2) the \emph{cannibalization effect} from lower-value requests when capacity is not sufficient. Despite the complexity of balancing these two effects, we show that the optimal threshold admits a strikingly simple structure characterized by the minimum of two interpretable thresholds, one that accounts for cannibalization and one that accounts for underutilization. Which threshold binds depends on the service capacity and the compliance level.
We further show that this characterization remains effective when we relax the random-allocation assumption and allow providers to prioritize high-scoring requesters to different degrees (\S\ref{sec:robustness_prioritization}).

\begin{figure}
    \centering
    \includegraphics[width=\linewidth]{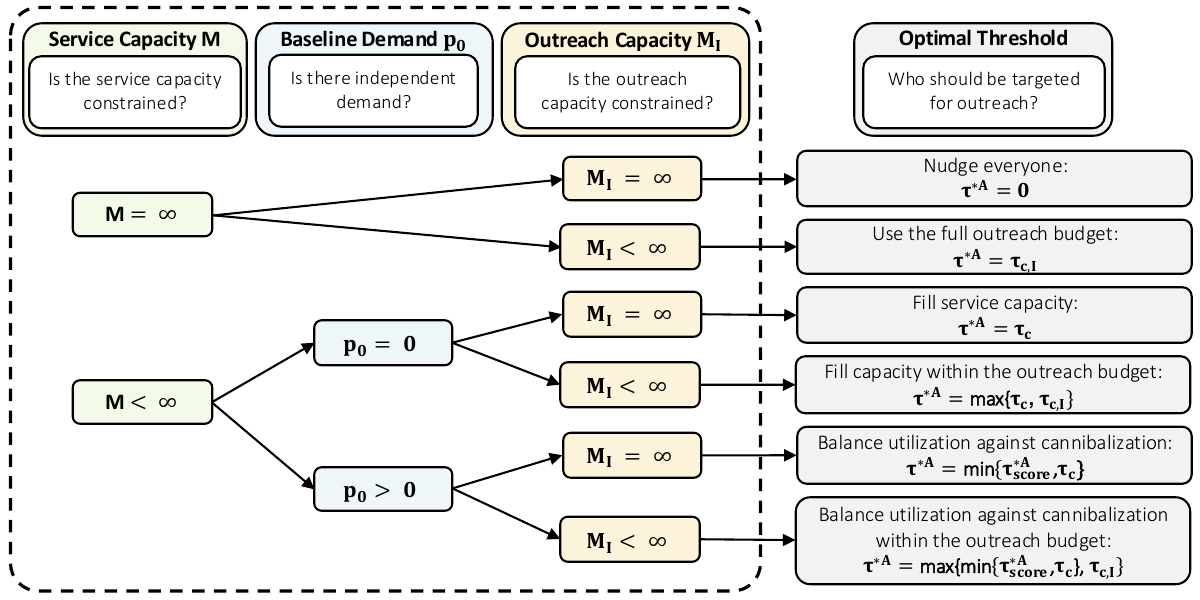}
    \caption{Decision tree of optimal thresholding policies across operational regimes. For simplicity, we write $\tau^{*A}_{\mathrm{score}}$ for $\tau^{*A}_{\mathrm{score}}(p_0, \Delta_P)$, $\tau_c$ for $\tau_c(M/N \mid p_0, \Delta_P)$, $\tau_{c,I}$ for $\tau_{c,I}(M_I/N)$ and $\tau^{*A}$ for $\tau^{*A}(M/N \mid p_0, \Delta_P)$. Each leaf is read in the parameter context defined by its branch.}
    \label{fig:decision_tree}
\end{figure}

\textbf{Suboptimality of natural thresholding policies (\S\ref{sec:thresholding}).} We analyze two natural benchmark thresholding policies: prediction-based thresholds, which are selected based on predictive performance metrics such as sensitivity and specificity, and capacity-matching thresholds, which flag just enough individuals to match expected requests to available capacity. We show that both are suboptimal except in specific regimes where they happen to coincide with the optimal threshold (Theorems \ref{thm:prediction_based_threshold_suboptimal} and \ref{thm:capacity_threshold_suboptimal}). Prediction-based thresholds ignore operational constraints entirely, while capacity-matching thresholds ignore cannibalization.

\textbf{Efficacy-optimal algorithm selection (\S\ref{sec:algorithm_selection}).} We show that standard evaluation metrics are not aligned with operational performance. AUC weights the predictive accuracy of an algorithm equally over all potential thresholds, while we show in \S\ref{sec:thresholding} that optimal deployment uses only a subset of the thresholds.
As a result, when selecting between two algorithms, selection based on AUC can choose the wrong algorithm if the ROC curves cross. We propose the \operationalmetricfull{} (\operationalmetric{}), which measures how well the algorithm performs specifically over the thresholds it will be deployed at, instead of all possible thresholds.
We show that \operationalmetric{} correctly identifies the efficacy-optimal algorithm.

\textbf{Empirical simulations (\S\ref{sec:case_study}).} We conduct a simulation case study calibrated to data from a sepsis early warning system across 8 adult acute care hospitals from the NewYork-Presbyterian health system.
In the simulation study, commonly used  thresholds lose up to 40\% of the system efficacy attained by our optimal threshold.
We also show that higher AUC does not correspond to better downstream outcomes: a model with AUC of 0.806 can outperform a model with AUC of 0.826 by up to 22.8\% in system efficacy when both are deployed at their optimal thresholds.

\textbf{Extensions: heterogeneity in baseline request probabilities and treatment effects (\S\ref{sec:extensions}).} We extend our framework along two dimensions that capture richer features of real-world deployments: (i) heterogeneous baseline request probabilities across subpopulations, allowing independent request rates to correlate with true scores; and (ii) heterogeneous treatment effects, in which responsiveness varies with the predicted score. In each case, the tradeoff between cannibalization and capacity underutilization carries through: the optimal threshold remains the minimum of a score-optimal threshold and a capacity-matching threshold, with the underlying parameters appropriately adjusted. Importantly, even when heterogeneity favors the planner, e.g., when higher-risk individuals are also more responsive to intervention, naive thresholds remain suboptimal.

\textbf{Guidance for practical implementation.} Our results have practical benefits: practitioners can use existing risk models and achieve substantial improvements simply by adjusting thresholds based on capacity and behavioral parameters (\Cref{fig:decision_tree}), without needing to retrain models. Similarly, when choosing among candidate algorithms, evaluating with \operationalmetric{} rather than AUC requires only additional knowledge of the operating capacity distribution and behavioral parameters---quantities typically estimable from operational data.

\smallskip

Taken together, our framework complements the traditional focus on algorithm accuracy with a systematic approach to algorithm implementation.
While machine learning advances have led to more accurate predictive models, our results show that the operational implementation can significantly facilitate, or hinder, the success of these systems.

%% file: section/related_work.tex
\section{Related Work}
\paragraph{Optimal targeting}
Our work is related to the literature on finding optimal targeting rules under constraints \citep{bhattacharya2012inferring, kitagawa2018should, sun2021treatment, xu2022estimating, imai2023experimental, killian2023robust, sverdrup2025qini, baek2025policy, wilder2025learning, sun2026empirical}.
This literature typically studies settings with heterogeneous treatment effects (and, in some work, heterogeneous treatment costs), where the goal is to maximize welfare subject to a budget constraint.
A general result is that the optimal policy is a thresholding rule on the predicted treatment effect, with the threshold set by the budget. Our optimal policy is also a thresholding rule, but our setting incorporates noisy compliance, which introduces a cannibalization effect that the standard framework does not capture.
The most closely related paper is \citet{hu2025optimal}, which also studies thresholds for algorithm-guided interventions under capacity constraints: taking a single predictive model as given, they characterize how its threshold should adapt over time as capacity fluctuates in a queueing system, and propose a procedure to learn the optimal dynamic policy. We instead study deployment at a single point in time, provide a closed-form characterization of the optimal threshold, and address the complementary question of selecting among candidate models via an operational metric.
\paragraph{Algorithm-assisted interventions}
Predictive algorithms increasingly guide interventions at scale: cancer screening outreach \citep{park2024machine}, sepsis intervention \citep{sendak2020real, boussina2024impact}, clinical deterioration response \citep{escobar2020automated, bertsimas2025early}, inpatient discharge planning \citep{na2025patient}, academic advising \citep{schechtman2025discretion}, dropout prevention \citep{perdomo2025difficult}, and child welfare screening \citep{chouldechova2018case}; see \citet{liu2026bridgingpredictioninterventionproblems, raji2025evaluatingpredictionbasedinterventionshuman} for broader perspectives.
Two features recur across these settings and motivate our model. First, the algorithm acts through an alert, message, or recommendation whose effect on behavior is probabilistic: flagged patients may decline to schedule diagnostic screening \citep{park2024machine} and advised students may fail to respond to outreach \citep{Oreopoulos2019-bl}. Second, service capacity is limited: screening programs and advising offices operate under constrained diagnostic and advising capacity \citep{mann2024potential, park2024machine, schechtman2025discretion}, and deterioration monitoring and discharge workflows contend with scarce inpatient beds and rapid-response teams \citep{sendak2020real, bertsimas2025early, na2025patient}.
\paragraph{Decision-focused learning} \label{sec:decision_focused_learning}
Our algorithm selection results connect to the decision-focused learning literature \citep{bertsimas2018optimization,ban2019big, bertsimas2019dynamic, lopez2019cost, wang2026prescriptive, bennett2020efficient,elmachtoub2022smart, mandi2022decision, mandi2024decision}, which aligns predictive models with downstream decision objectives; for example, \citet{elmachtoub2022smart} train predictors with decision-aware loss functions and show that accuracy-optimized models can underperform in decision-making tasks. This literature intervenes at the training stage; we instead take prediction models as given and improve downstream efficacy at the \emph{deployment} stage, through thresholding and algorithm selection. The closest in spirit is \citet{bertsimas2026algorithmicinsurance}, which shows that classification thresholds tuned for accuracy can substantially increase tail risk relative to thresholds tuned for the downstream objective.
\paragraph{Compliance in Causal Inference}
The causal inference literature has studied the impact of imperfect compliance, where individuals offered a treatment do not always take it up \citep{imbens1994late, angrist1996identification}, and note the impact this can have on diluting the treatment effect.  \citet{athey2025targeted} estimate per-individual compliance rates from randomized experiments to derive targeting rules. In constrast to this literature, we instead study optimal targeting rules under given compliance parameters, accounting for downstream capacity constraints.

%% file: section/model.tex
\section{Model}
\label{sec:model}
We introduce three key elements of the model: 1) individual request behavior with and without nudging, 2) capacity-constrained service, and 3) the algorithm-assisted outreach mechanism. We then define the metric of interest, the total system efficacy of all served requests, and introduce naive (but commonly used) thresholding policies.

\subsection{Individual Behavior With and Without Nudging}
\label{sec:individual_behavior}
Consider $N$ users. Let $s_i \in \{0,1\}$ denote the nudging indicator for user $i$, where $s_i = 1$ if user $i$ is nudged by the outreach intervention and $s_i = 0$ otherwise.
Let $p_0 \geq 0$ be the baseline probability for independent requests and $\Delta_P > 0$ be the lift in request probability due to nudging, where $p_0 + \Delta_P \le 1$. Let $d_i \in \{0,1\}$ denote the request indicator for user $i$, where $d_i=1$ if user $i$ requests service. The request probability depends on whether the user is nudged:
\begin{align*}
  \text{Without nudging: } & \mathbb{P}\{d_i=1 \mid s_i = 0\} = p_0, \\
  \text{With nudging: } & \mathbb{P}\{d_i=1 \mid s_i= 1\} = p_0 + \Delta_P.
\end{align*}
Combining the two groups, the request indicator for user $i$ is a Bernoulli random variable with:
\begin{equation}\label{eq:request_probability}
d_i \sim \mathrm{Bernoulli}(p_0 + \Delta_P\, s_i).
\end{equation}

\subsubsection{Examples}

Nudging operates through two mechanisms depending on who receives the algorithmic signal.

\paragraph{Nudging requesters.}
In cancer screening systems (e.g. \cite{park2024machine}), patients identified as high-risk are flagged for outreach encouraging them to schedule preventive screenings. Here, $d_i$ represents the patient's decision to request screening, $p_0$ reflects the likelihood a patient follows standard screening guidelines, and $\Delta_P$ captures the increase in screening uptake due to outreach.

\paragraph{Nudging providers.}
In early warning systems for clinical deterioration (e.g. \cite{escobar2020automated, bertsimas2025early}), algorithmic alerts flag at-risk individuals, prompting providers to initiate outreach. Here, $d_i$ is the provider's decision to deliver the intervention to individual $i$: $p_0$ is the baseline rate of intervention under standard processes, and $\Delta_P$ is the increase due to alerts.

\subsection{Capacity Constrained Service}
Suppose each individual makes requests independently according to \Cref{eq:request_probability}. The total number of requests is $\sum_{i=1}^N d_i$.
The system has limited capacity $M$. The number of served requests is
\begin{equation*}
\min\left\{\sum_{i=1}^N d_i , M\right\}.
\end{equation*}
When the total number of requests exceed capacity ($\sum_{i=1}^N d_i > M$), not all requesters can be served.

Unless otherwise specified, we assume service slots are allocated uniformly at random among requesters, which is a common practice in many real-world systems.
For example, in systems where individuals must request service, slots are allocated through online portals or on a first-come, first-served basis where score-based prioritization is infeasible. In other systems, e.g., healthcare, education, and/or other public service systems, ethical norms require that all requests should be treated equally. We later relax this assumption in \Cref{sec:robustness_prioritization}, where we study a mixture mechanism that allows the provider to partially or fully prioritize high-scoring requesters, and show that our main characterization remains effective across a wide range of allocation policies.

\subsection{Algorithm-Assisted Outreach}\label{sec:algorithm_assisted_outreach}
Each individual $i$ has a latent true score $r_i$, with $\{r_i\}_{i=1}^N$ drawn i.i.d.\ from a continuous distribution with mean $\mathbb{E}[r]$.
If the outcome of interest is continuous (e.g., health improvement, academic performance, job performance), the true score $r_i$ represents the efficacy from serving user $i$.
If the outcome of interest is binary (e.g., disease occurrence, student dropout, candidate acceptance), the true score $r_i$ represents the probability of the positive outcome. Without loss of generality, we assume $r_i \in [0,1]$ for all $i$.

An algorithm $A$ produces a predicted score $\hat r_i^A$ for each individual. We make the following assumptions on the predictive quality of algorithms.
\begin{assumption}[Unbiasedness] \label{assump:unbiasedness}
For any algorithm $A$, $\mathbb{E}[\hat r^A] = \mathbb{E}[r]$.
\end{assumption}

\begin{assumption}[Monotone Calibration]\label{assump:monotone_calibration}
For any algorithm $A$, higher predicted scores correspond to higher expected true scores. Formally, for any $q_1 > q_2$,
\(
\mathbb{E}\big[r \mid \hat r^A = q_1\big] > \mathbb{E}\big[r \mid \hat r^A = q_2\big].
\)
\end{assumption}
We now specify the nudging indicator $s_i$ from \Cref{sec:individual_behavior}: the designer chooses a quantile threshold $\tau \in [0,1]$ and flags the top $(1-\tau)$ fraction of individuals, i.e., those whose predicted scores are at or above the $\tau$-quantile.
Unless otherwise specified, we assume throughout that the intervention itself is uncapacitated, so that any threshold $\tau \in [0,1]$ is feasible. We relax this assumption in \S\ref{sec:implementation-guidance}, where a finite intervention budget $M_I < \infty$ further restricts the feasible threshold range. Formally, given algorithm $A$ and threshold $\tau$, individual $i$ is flagged if
\[
 s_i^A(\tau)=\mathbf 1\{\hat r_i^A \ge \hat q_A(\tau)\},
\]
where $\hat q_A(\tau)$ is the empirical $\tau$-quantile of the predicted scores $\{\hat r_i^A\}_{i=1}^N$.

In addition, we make the following standard assumptions for tractability.
\begin{assumption}[Existence and Regularity of Joint Density]\label{assump:joint_density}
For any algorithm $A$, the pair $(r, \hat r^A)$ admits a joint density
$f_{r,\hat r^A}(r,s)$ and the first moment is finite:
\(
\int |r|\, f_{r,\hat r^A}(r,s)\,dr < \infty,
\, \forall s \in [0,1].
\)
\end{assumption}

\begin{assumption}[Non-degenerate Predicted-Score Distribution]\label{assump:marginal_positive}
The marginal density
\(
f_{\hat r^A}(s) = \int f_{r,\hat r^A}(r,s)\,dr
\)
is continuous and strictly positive on its support, which is an interval $(a,b)$.
\end{assumption}

\begin{assumption}[Smoothness of Conditional Expectations]\label{assump:conditional_smoothness}
The map
\(
s \mapsto f_{r,\hat r^A}(r,s)
\)
is continuous and piecewise continuously differentiable in $s$, with an integrable dominating function.
Consequently, the conditional expectations
\(
\mathbb{E}[\,r \mid \hat r^A = s\,]
\, \text{and} \,
\mathbb{E}[\,r \mid \hat r^A > s\,]
\)
are continuous in $s$, and are continuously differentiable on the support of $f_{\hat r^A}$.
\end{assumption}

\subsection{Metric of Interest}
The performance metric of interest is \emph{system efficacy}, which we define to be the total expected value of served requests. In healthcare, this corresponds to the number of true positive cases identified for early treatment; in education, the number of at-risk students who receive advising.

Let $z_i \in \{0,1\}$ indicate whether individual $i$ is ultimately served. The system efficacy is the expected total value of served requests:
\begin{equation}\label{eq:objective_sum_form}
  \obj\left(\tau \mid A, M, N, p_0, \Delta_P\right) = \mathbb{E}\left[\sum_{i=1}^N z_i r_i\right].
\end{equation}
Under random allocation, each requester is served with probability $\min\left\{\frac{M}{\sum_{j=1}^N d_j},\, 1\right\}$. Substituting this into \eqref{eq:objective_sum_form} yields the equivalent form
\begin{equation}\label{eq:objective}
  \begin{aligned}
    & \obj\left(\tau \mid A, M, N, p_0, \Delta_P\right)  = \mathbb{E}\Biggl[\min\left\{\frac{M}{\sum_{i=1}^N d_i}, 1\right\}
    \cdot \sum_{i=1}^N d_i r_i\Biggr].
  \end{aligned}
\end{equation}

We consider two goals:
\begin{enumerate}[(i)]
  \item For a fixed algorithm $A$, characterize the optimal threshold:
  \begin{equation}\label{eq:thresholding_problem}
    \begin{aligned}
      & \tau^{*A}\left(M,N \mid p_0, \Delta_P\right)  = \arg\max_{\tau \in [0,1]} \obj\left(\tau \mid A, M, N, p_0, \Delta_P\right).
    \end{aligned}
  \end{equation}
  \item Given a set of algorithms $\mathcal{A}$, select the best algorithm under optimal deployment:
  \begin{equation}\label{eq:algo_selection_problem}
    \begin{aligned}
      & A^{*}\left(M,N \mid p_0, \Delta_P\right)  = \arg\max_{A \in \mathcal{A}} \obj\Bigl(\tau^{*A}\left(M,N \mid p_0, \Delta_P\right)\mid A, M, N, p_0, \Delta_P\Bigr).
    \end{aligned}
  \end{equation}
\end{enumerate}

\subsection{Natural Thresholding Policies}
\label{sec:naive_threshold}
We study two natural thresholding policies, prediction-based thresholds and capacity-matching thresholds.

\subsubsection{Prediction-Based Thresholds}
\label{sec:prediction_based_threshold}
In many settings, practitioners choose an operating point based on predictive performance alone \citep{fawcett2006introduction}. These \emph{prediction-based thresholds} depend only on predictive metrics such as True Positive Rate (TPR; sensitivity) or Positive Predictive Value (PPV), and remain fixed regardless of changes in operational parameters $(p_0, \Delta_P, M/N)$. We illustrate several common examples of prediction-based thresholds:
\begin{example}[Target sensitivity]
A threshold is selected to achieve a minimum sensitivity $\mathrm{TPR}_{\min}$ \citep{fawcett2006introduction}:
\begin{equation}\label{eq:target_sensitivity}
  \tau_A = \sup\{\tau \in [0,1] : \mathrm{TPR}_A(q_A(\tau)) \geq \mathrm{TPR}_{\min}\}.
\end{equation}
\end{example}

\begin{example}[Maximize specificity subject to sensitivity]
The threshold maximizes specificity while maintaining a minimum sensitivity level \citep{fawcett2006introduction}:
\begin{equation}\label{eq:optimize_specificity}
  \begin{aligned}
    & \tau_A = \argmax_{\tau} \mathrm{TNR}_A(q_A(\tau)) \\
    & \quad \text{subject to} \quad \mathrm{TPR}_A(q_A(\tau)) \geq \mathrm{TPR}_{\min}.
  \end{aligned}
\end{equation}
\end{example}

\begin{example}[Cost-sensitive classification]
A threshold is chosen to balance false negatives and false positives according to their relative costs $c_{\mathrm{FN}}$ and $c_{\mathrm{FP}}$ \citep{elkan2001foundations}:
\begin{equation}\label{eq:predictive_metric_based_threshold}
  \tau_{A}=\argmin_{\tau} c_{\mathrm{FN}} \mathrm{FNR}(\tau) + c_{\mathrm{FP}} \mathrm{FPR}(\tau)
\end{equation}
where $\mathrm{FNR}(\tau)$ and $\mathrm{FPR}(\tau)$ are the false negative rate and false positive rate at threshold $\tau$.
\end{example}

Formally, we define prediction-based thresholds considered in this paper.
\begin{definition}[Prediction-Based Threshold]\label{def:prediction-basedthresholds}
A thresholding rule $\tau_A$ is called \emph{prediction-based} if it depends only on: (i) the algorithm's predictive performance metrics (e.g., TPR, FPR, PPV), and (ii) exogenous parameters $c$ (e.g., cost ratios, target sensitivity levels) that are independent of the operational parameters $(p_0, \Delta_P, M,N)$.
Formally, $\tau_A$ can be written as
\[
\tau_A \in \arg\max_{\tau\in[0,1]} f(M_A(\tau),c),
\]
where $M_A(\tau)$ represents some prediction-based performance metrics at threshold $\tau$, and $f(\cdot, c)$ is any objective function that depends only on these metrics and the exogenous parameters $c$.
\end{definition}

\subsubsection{Capacity-Matching Thresholds}\label{sec:capacity_matching_threshold}
A natural alternative to
prediction-based thresholds, which ignore operational constraints entirely, is to set the threshold so that the expected number of requests is equal to the available capacity \citep{lakkaraju2015machine}. We call these thresholds \emph{capacity-matching thresholds}.
\begin{definition}[Capacity-matching threshold]
Given system capacity $M$, total population size $N$, baseline request probability $p_0$, and targeted request probability lift $\Delta_P$, the capacity-matching threshold $\tau_c$ is defined as a function of the \emph{capacity ratio} $\frac{M}{N}$:
\begin{equation}\label{eq:capacity_matching_threshold}
  \begin{aligned}
    & \tau_c\left(\tfrac{M}{N}\mid p_0, \Delta_P\right) = \min\!\left(1, \max\!\left(0,\, 1 - \tfrac{1}{\Delta_P}\left(\tfrac{M}{N} - p_0\right)\right)\right).
  \end{aligned}
\end{equation}
\end{definition}
The threshold depends on $M$ and $N$ only through their ratio $\frac{M}{N}$, which captures the relative abundance of capacity. This yields three operational regimes:
\begin{enumerate}[(i)]
    \item When $\frac{M}{N} \leq p_0$, $\tau_c= 1$ (no outreach needed). Baseline request alone fills available capacity.
    \item When $p_0 < \frac{M}{N} < p_0 + \Delta_P$, $\tau_c \in (0,1)$ decreases linearly with the capacity ratio to generate sufficient additional requests.
    \item When $\frac{M}{N} \geq p_0 + \Delta_P$, $\tau_c= 0$ (maximum outreach). Capacity is sufficient to serve all requests even when all individuals are flagged for outreach.
\end{enumerate}
As the capacity ratio $\frac{M}{N}$ increases, more aggressive targeting (lower $\tau_c$) is required to fill available capacity. Similarly, as the baseline request probability, $p_0$, decreases, the threshold must be lowered to compensate for fewer independent requests and prevent capacity underutilization.
Although capacity-matching thresholds ensure full utilization, we will show in \Cref{sec:thresholding} that they can still be suboptimal.

\subsection{Fluid Model}
\label{sec:fluid_limit}
Due to the complexity of our stochastic model, we consider a fluid approximation to obtain analytical insights. Our fluid model makes two approximations:
(i) we replace the stochastic number of service requests with its deterministic (potentially fractional) expectation $N[p_0 + \Delta_P(1-\tau)]$, and (ii) we replace the empirical quantile $\hat q_A(\tau)$ with the population quantile $q_A(\tau) = F_A^{-1}(\tau)$ of the predicted-score distribution $\hat r^A$. We use the above fluid model to approximate the stochastic model for large systems.

We further define the fluid objective $\objfluid\left(\tau \mid A, M, N, p_0, \Delta_P\right)$, which substitutes the approximations above into \Cref{eq:objective}, and decomposes into the product of two deterministic terms:
\begin{equation}\label{eq:objective_fluid_limit}
  \begin{aligned}
    & \objfluid\left(\tau \mid A, M, N, p_0, \Delta_P\right) = \tilde{N}(\tau \mid M, N, p_0, \Delta_P) \cdot \tilde{R}\left(\tau \mid A, p_0, \Delta_P\right),
  \end{aligned}
\end{equation}
where
\begin{equation}\label{eq:fluid_served_requests}
  \begin{aligned}
    & \tilde{N}(\tau \mid M, N, p_0, \Delta_P) = \min\left\{N[p_0 + \Delta_P(1-\tau)],\, M\right\}
  \end{aligned}
\end{equation}
is the number of served requests, and
\begin{equation}\label{eq:fluid_efficacy}
  \begin{aligned}
    & \tilde{R}\left(\tau \mid A, p_0, \Delta_P\right) = \frac{p_0\,\mathbb{E}[r] + \Delta_P\,(1-\tau)\,\mathbb{E}[r\mid \hat{r}^A \ge q_A(\tau)]}{p_0 + \Delta_P\,(1-\tau)}
  \end{aligned}
\end{equation}
is the average efficacy of a served request. We denote all fluid quantities with a tilde.
System efficacy thus depends on two factors: capacity underutilization (how much available capacity is unused) and efficacy per service slot (how effectively capacity is allocated to high-value individuals).

The following proposition shows that the fluid objective asymptotically upper bounds the stochastic objective and becomes exact as $M,N \rightarrow \infty$ along a fixed capacity ratio (proof in Appendix~\ref{proof:prop:fluid_limit_approximation}). The bound holds  asymptotically, rather than point-wise, because the fluid model replaces sample averages and the empirical score quantile with their population counterparts, whose finite-sample errors have no definite ordinal relationship.
\begin{proposition}[Fluid approximation properties]
\label{prop:fluid_limit}
Fix a threshold $\tau \in [0,1]$ and consider any sequence of systems with $N \to \infty$ and capacity $M = M_N$ such that $M_N/N \to \rho$ for some $\rho \in (0,\infty)$. Assume $\Delta_P > 0$ and $p_0 + \Delta_P(1-\tau) > 0$.\footnote{The excluded case $p_0 + \Delta_P(1-\tau)=0$ (only $\tau=1$ with $p_0=0$) is degenerate: no requests occur and $\obj = \objfluid = 0$.} Then
\begin{equation}\label{eq:fluid-upper}
\limsup_{N \to \infty}\tfrac{1}{N}\Bigl(
      \obj\left(\tau \mid A, M, N, p_0, \Delta_P\right)  - \objfluid\left(\tau \mid A, M, N, p_0, \Delta_P\right)\Bigr) \le 0,
\end{equation}
\begin{equation}\label{eq:fluid-limit}
\lim_{N \to \infty} \tfrac{1}{N}
      \bigl|\, \objfluid\left(\tau \mid A, M, N, p_0, \Delta_P\right) - \obj\left(\tau \mid A, M, N, p_0, \Delta_P\right) \bigr| = 0.
\end{equation}
\end{proposition}
\Cref{fig:fluid_approximation} shows that the fluid objective upper bounds the stochastic objective (computed via exhaustive search) with an approximation error that vanishes as $M$ and $N$ grow, motivating our use of the fluid model in subsequent analysis.

\begin{figure}
    \centering
    \includegraphics[width=0.55\linewidth]{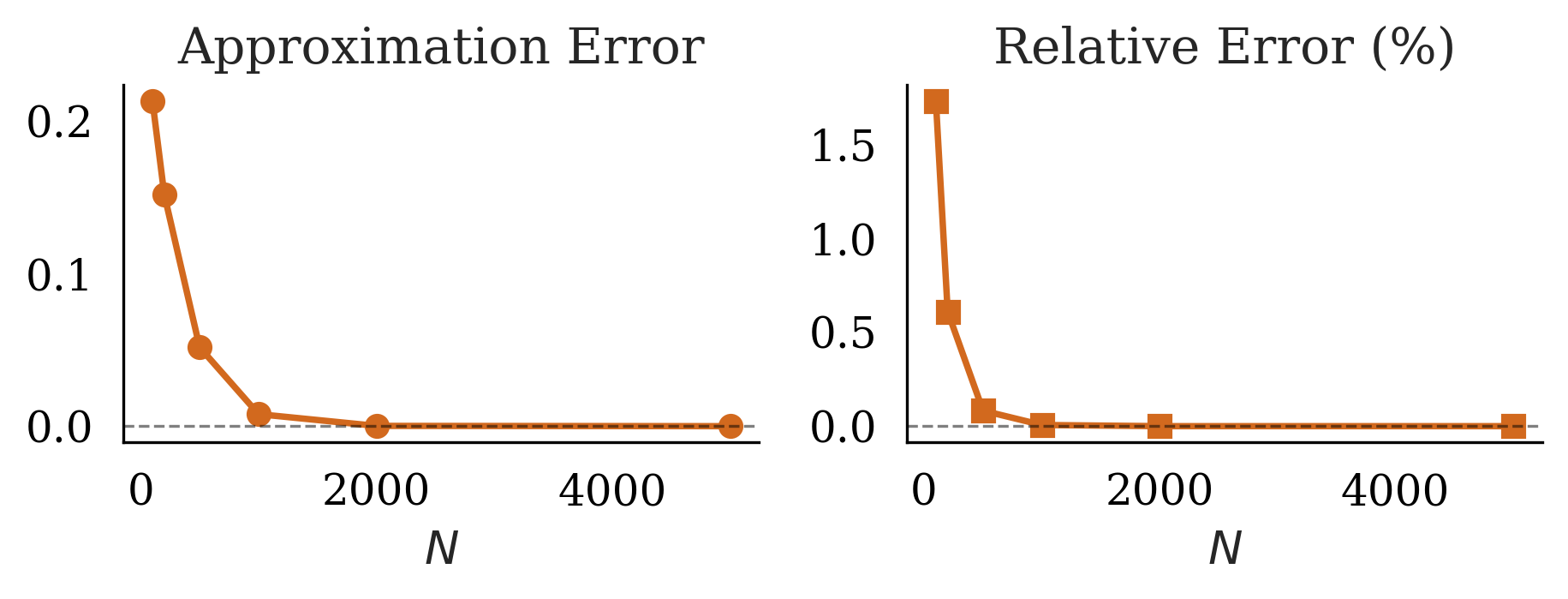}
    \caption{Visualizations of fluid approximation accuracy. The fluid approximation consistently upper bounds the true objective (left panel), and the error converges to 0 as $M$ and $N$ scale up (right panel). Parameters: $M/N=0.2$ , $p_0=0.1$, $\Delta_P=0.5$ and $\hat r = r \sim 0.7 \mathrm{Beta}(2,10)+0.3 \mathrm{Beta}(8,2)$.}
    \label{fig:fluid_approximation}
\end{figure}

%% file: section/thresholding.tex
\section{Optimal Thresholds and Suboptimality of Naive Thresholds}\label{sec:thresholding}
Using the fluid objective as an approximation,
we characterize the optimal threshold $\tau^{*A}(M,N| p_0, \Delta_P)$ to maximize the system efficacy. We first present a motivating example illustrating two competing effects to balance in capacity-constrained settings: the \emph{cannibalization effect}, which drives the suboptimality of \emph{who gets served} when capacity constraints bind, and the \emph{capacity underutilization effect}, which drives the suboptimality of \emph{how many get served} by penalizing overly conservative targeting. Interestingly, we show that the optimal threshold admits a simple two-point characterization---the minimum of two thresholds, each capturing one of these effects.

\subsection{Motivating Example}
\label{sec:cannibalization_effect}
Suppose there are 100 individuals, each with baseline request probability $p_0 = 0.1$ without outreach.
Targeted outreach increases an individual’s request probability to $p_0+\Delta_P = 0.6$. The service provider has capacity to serve 20 requests. For simplicity, assume that true scores are uniformly distributed on $[0,1]$ and that the algorithm perfectly ranks individuals by their true scores. We consider four thresholding scenarios.

\begin{figure}
    \centering
    \includegraphics[width=0.9\linewidth]{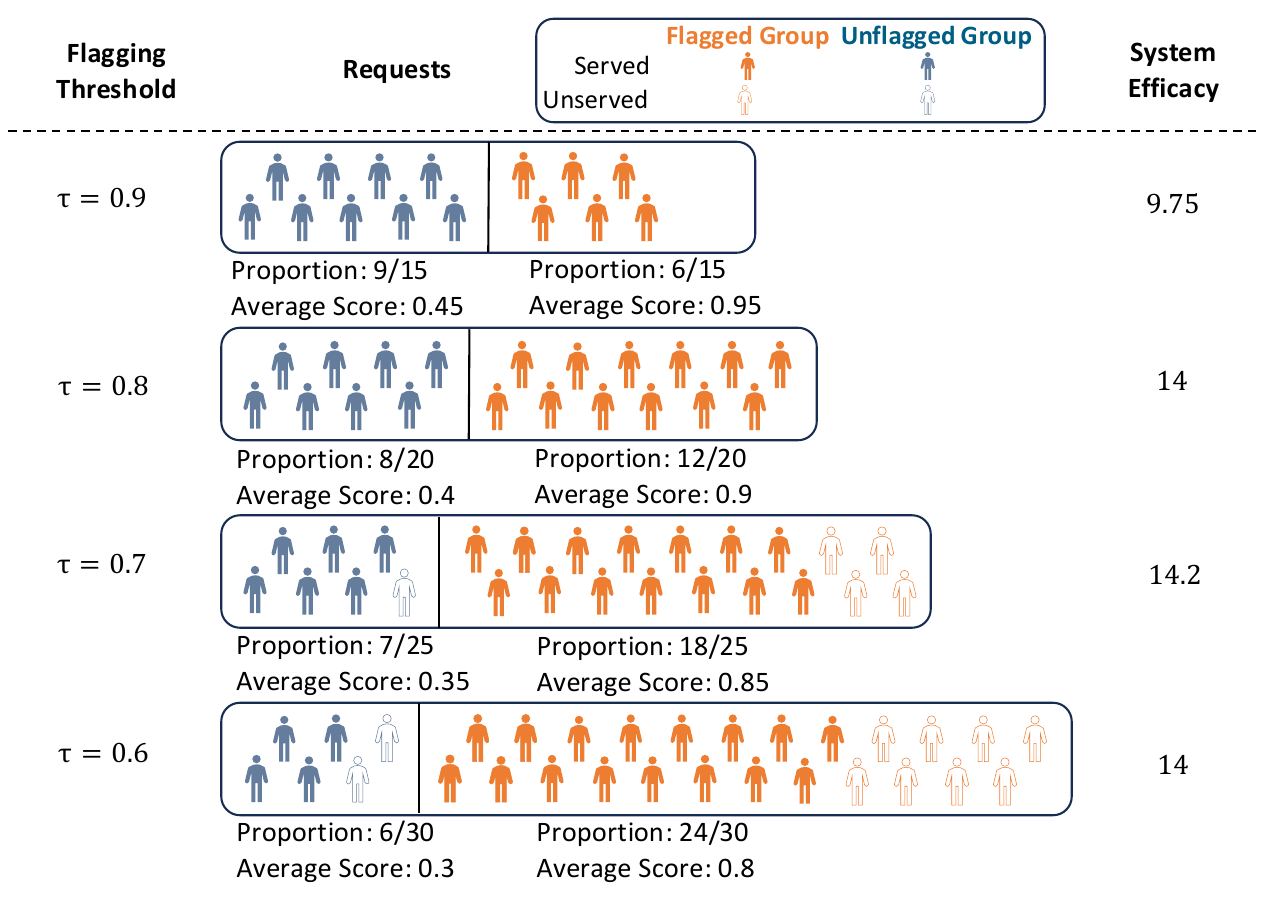}
    \caption{Requests decomposition under varying outreach thresholds. Each row corresponds to a setting from the motivating example: Setting 1 ($\tau = 0.9$) shows underutilization; Setting 2 ($\tau = 0.8$) matches capacity exactly; Setting 3 ($\tau = 0.7$) over-expands outreach, shifting capacity toward targeted requests; Setting 4 ($\tau = 0.6$) further expands outreach but increases cannibalization from lower-value individuals.}
    \label{fig:cannibalization_effect}
\end{figure}

\textbf{Setting 1: Flagging top 10\% ($\tau = 0.9$).} The expected number of independent requests is $100 \times 0.9 \times 0.1 = 9$ and targeted requests is $100 \times 0.1 \times 0.6 = 6$, yielding a total of 15 requests. In this case, the threshold is overly conservative and leads to underutilization. Such underutilization is costly, as it leaves service slots unused that could have been filled by nudging additional individuals to request service (see the first row of \Cref{fig:cannibalization_effect}).

\textbf{Setting 2: Flagging top 20\% ($\tau = 0.8$).} The expected number of independent requests is $100 \times 0.8 \times 0.1 = 8$ and targeted requests is $100 \times 0.2 \times 0.6 = 12$, yielding a total of 20 requests, exactly matching capacity. Among the served requests, $8/20$ are independent requests with average true score 0.4, and $12/20$ are targeted requests with substantially higher average true score 0.9 (see the second row of \Cref{fig:cannibalization_effect}).

\textbf{Setting 3: Flagging top 30\% ($\tau = 0.7$).} The expected number of independent requests is $100 \times 0.7 \times 0.1 = 7$ and targeted requests is $100 \times 0.3 \times 0.6 = 18$, yielding a total of 25 requests. Capacity is now scarce. Under random allocation, independent requests account for $7/25$ of demand, compared to $8/20$ at $\tau = 0.8$, allowing a larger share of service capacity to be allocated to targeted requests. Although the average true score among targeted requests decreases to 0.85, more targeted requests are served overall (see the third row of \Cref{fig:cannibalization_effect}).

\textbf{Setting 4: Flagging top 40\% ($\tau = 0.6$).} The expected number of independent requests is $100 \times 0.6 \times 0.1 = 6$ and targeted requests is $100 \times 0.4 \times 0.6 = 24$, yielding a total of 30 requests. The fraction of independent requests further declines to $6/30$, and an even larger portion of capacity is allocated to targeted requests. However, because increasingly lower-risk individuals are also nudged, the average true score of targeted requests falls to 0.8, and the system efficacy no longer continues to increase (see the fourth row of \Cref{fig:cannibalization_effect}).

Setting 2, 3 and 4 reveal a non-monotonic relationship between the threshold and system efficacy when capacity constraints bind: lowering the threshold initially improves system efficacy by shifting capacity toward higher-value individuals from the flagged group, but eventually reduces value by intensifying cannibalization from low-value requests from both non-flagged and flagged groups. This relationship is driven by
\emph{cannibalization}, where lower-value requests crowd out higher-value ones. Cannibalization takes two forms. Setting a high threshold preserves the efficacy of the flagged group but risks cannibalization from baseline requests, which capture a larger share of capacity. Conversely, setting a low threshold reduces cannibalization from baseline requests but introduces cannibalization from lower-value individuals within the flagged group.

\subsection{Optimal Threshold}
\label{sec:two_point_optimality}

We will show that the optimal threshold is characterized by two thresholds, each addressing a different operational effect: the \emph{capacity-matching threshold} $\tau_c(\tfrac{M}{N}\mid p_0,\Delta_P)$ introduced in \Cref{sec:capacity_matching_threshold}, which ensures the full utilization of capacity, and the \emph{score-optimal threshold}, which we introduce below.

The score-optimal threshold is motivated by the non-monotonic relationship between threshold and system efficacy in the motivating example: when all capacity is filled, the threshold must optimally balance the two forms of cannibalization, which amounts to maximizing efficacy per service slot.

\begin{definition}[Score-Optimal Threshold]
\label{def:score_optimal_threshold}
Given the behavioral parameters $p_0$ and $\Delta_P$, the \emph{score-optimal threshold} $\tau^{*A}_{\mathrm{score}} (p_0, \Delta_P)$ maximizes the efficacy per service slot:
\begin{equation}\label{eq:score_optimal_threshold}
  \tau^{*A}_{\mathrm{score}} (p_0, \Delta_P) \in \arg\max_{\tau \in [0,1]} \tilde{R}\left(\tau | A, p_0, \Delta_P\right).
\end{equation}
\end{definition}
To ensure the score-optimal threshold is well-defined, we impose regularity conditions on the joint distribution of true and predicted scores. Under \Cref{assump:joint_density,assump:marginal_positive,assump:conditional_smoothness,assump:monotone_calibration}, the score-optimal threshold is unique (proof in Appendix \ref{proof:lem:score_optimal_threshold_uniqueness}).

We now have two candidate thresholds, each addressing a different operational concern: the \emph{score-optimal threshold} $\tau^{*A}_{\mathrm{score}}(p_0, \Delta_P)$ maximizes efficacy per served request; the \emph{capacity-matching threshold} $\tau_c(\frac{M}{N}\mid p_0,\Delta_P)$ ensures the full utilization of capacity. Neither of these thresholds alone is sufficient to characterize the optimal threshold as the optimal threshold must account for both concerns at the same time.
Remarkably, we show that the optimal threshold has a surprisingly simple two-point structure: it is the minimum of the score-optimal threshold and the capacity-matching threshold.
Formally, we have the following result (proof in Appendix \ref{proof:thm:two_point_optimality}).

\medskip

\begin{theorem}[Two-Point Optimality]
\label{thm:two_point_optimality}
Under \Cref{assump:joint_density,assump:marginal_positive,assump:conditional_smoothness,assump:monotone_calibration}, for any capacity ratio $\tfrac{M}{N} > 0$, the threshold that maximizes the fluid objective in \Cref{eq:objective_fluid_limit} satisfies
\[
  \begin{aligned}
    & \tau^{*A}\left(\tfrac{M}{N}\mid p_0,\, \Delta_P\right)  = \min\left\{\tau^{*A}_{\mathrm{score}}(p_0,\, \Delta_P),\;\; \tau_c\left(\tfrac{M}{N}\mid p_0,\, \Delta_P\right)\right\}.
  \end{aligned}
\]
\end{theorem}
\noindent The two-point structure implies that the optimal policy always flags enough individuals to fully utilize capacity, but in some scenarios it is beneficial to flag \textit{more} individuals to shift the served slots towards the high-value requests in the flagged group.
Which of the two thresholds binds depends on system parameters, described below.

The two-point structure in \Cref{thm:two_point_optimality} implies that the system operates in one of two regimes, depending on the ratio $M/N$ of capacity to users (see \Cref{fig:thresholding_policies}). In the capacity constrained regime $\bigl(\tfrac{M}{N} < p_0 + \Delta_P\bigl(1 - \tau^{*A}_{\mathrm{score}}\bigr)\bigr)$,
slots fill easily and cannibalization is the primary concern: the optimal threshold reduces to $\tau^* = \tau^{*A}_{\mathrm{score}}$, and the key question is \emph{which requesters} are served. Conversely, in the capacity abundant regime $\bigl(\tfrac{M}{N} > p_0 + \Delta_P\bigl(1 - \tau^{*A}_{\mathrm{score}}\bigr)\bigr)$,
the optimal threshold reduces to $\tau^*=\tau_c$, shifting the focus to \emph{how many requesters} are served to ensure full utilization; at the boundary between the regimes, the two thresholds coincide.

\begin{figure}
    \centering
    \includegraphics[width=0.7\linewidth]{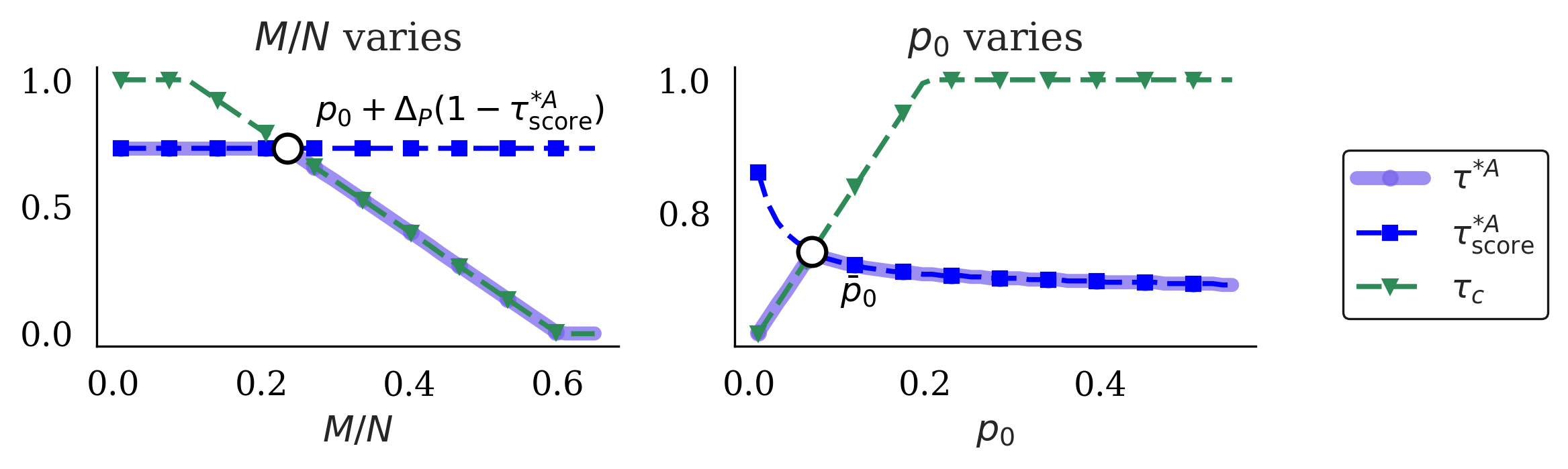}
    \caption{Optimal thresholding policy as the capacity ratio $\tfrac{M}{N}$ or the baseline request $p_0$ varies. Parameters: $N=1000, M=200, \Delta_P = 0.5, p_0 = 0.1$ and $\hat r = r \sim 0.7 \mathrm{Beta}(2,10)+0.3 \mathrm{Beta}(8,2).$ }
    \label{fig:thresholding_policies}
\end{figure}

The optimal threshold also depends on the baseline request probability $p_0$, which determines the severity of cannibalization through the volume of baseline requests (see the right panel of \Cref{fig:thresholding_policies}). The two regimes are separated by the critical baseline request probability
\begin{equation}\label{def:bar_p0}
\bar{p}_0 = \inf\Bigl\{ p_0 \in [0,1-\Delta_P] : \tau^{*A}_{\mathrm{score}}(p_0,\Delta_P) \leq \tau_c\!\left(\tfrac{M}{N}\mid p_0,\Delta_P\right) \Bigr\}.
\end{equation}
When $p_0 < \bar{p}_0$, cannibalization is weak because few individuals request service without outreach: the score-optimal threshold is high and does not fill capacity, so underutilization dominates and the optimal threshold expands outreach to ensure full utilization. Conversely, when $p_0 \geq \bar{p}_0$, cannibalization is severe: the score-optimal threshold decreases with $p_0$ to avoid crowding out high-value targeted requests, and flagging at $\tau^{*A}_{\mathrm{score}}$ already fills capacity, so the key question becomes \emph{which requesters} are served---lowering the threshold further would only reduce average efficacy by nudging lower-value individuals.

Importantly, in settings where $p_0 = 0$ (i.e. outreach is the only source of service requests), there are no baseline requests to crowd out, so cannibalization from the independent requests cannot occur.
In this case, $\tau^{*A}_{\mathrm{score}}(0, \Delta_P) = 1$ and the optimal threshold reduces to the capacity-matching threshold $\tau_c\!\left(\frac{M}{N}\mid 0, \Delta_P\right)$: the optimal strategy is simply to fill capacity with the highest-scoring individuals.

Since \Cref{thm:two_point_optimality} is derived in the fluid model, a natural question is how the two-point threshold performs in a system with finite $M$ and $N$. The following corollary shows that the two-point optimal threshold given by \Cref{thm:two_point_optimality} is asymptotically optimal as $M$ and $N$ grow with $M/N$ fixed (proof in Appendix \ref{proof:cor:finite_two_point_optimality}).
\begin{corollary}
\label{cor:finite_two_point_optimality}
Under \Cref{assump:joint_density,assump:marginal_positive,assump:conditional_smoothness,assump:monotone_calibration}, the two-point optimal threshold $\tau^{*A}\left(\frac{M}{N}\mid  p_0, \Delta_P\right)$ in the fluid model is asymptotically optimal:
\[
  \begin{aligned}
    & \lim_{N \to \infty} \tfrac{1}{N}\Bigl| \max_{\tau \in [0,1]}
      \obj\left(\tau \mid A, M, N, p_0, \Delta_P\right) - \obj\bigl(\tau^{*A}\!\left(\tfrac{M}{N}\mid p_0, \Delta_P\right)
      \mid A, M, N, p_0, \Delta_P\bigr) \Bigr| = 0.
  \end{aligned}
\]
\end{corollary}
Corollary~\ref{cor:finite_two_point_optimality} assumes the behavioral parameters $(p_0,\Delta_P)$ are known. When the parameters need to be estimated from data, the threshold computed from the estimates remains close to the true optimal and can be shown to converge to the optimal threshold as the sample size increases (Appendix~\ref{app:robustness-estimated-params}).

\subsection{Suboptimality of Prediction-based Thresholds}
\label{sec:subopt_prediction}
By \Cref{def:prediction-basedthresholds}, prediction-based thresholds depend only on predictive metrics and exogenous costs, and therefore remain \emph{fixed} as the operational parameters $(M/N, p_0, \Delta_P)$ vary.
Using the optimal threshold characterization in \Cref{thm:two_point_optimality}, we can show that prediction-based thresholds, which account for neither the capacity ratio $M/N$ nor the compliance levels $(p_0, \Delta_p)$, are suboptimal in general. (Proof in Appendix~\ref{sec:proof_suboptimality_prediction_based_thresholds}.)

\medskip

\begin{theorem}[Suboptimality of prediction-based thresholds]
\label{thm:prediction_based_threshold_suboptimal}
For an algorithm $A$, consider a threshold $\tau_{A} \in (0,1)$\footnote{The boundary threshold $\tau_A = 0$ is optimal precisely on the abundant-capacity region $M/N \ge p_0 + \Delta_P$, and is strictly suboptimal for all $M/N < p_0 + \Delta_P$.} that is fixed with respect to the operational parameters $(M/N, p_0, \Delta_P)$, in particular any prediction-based threshold as defined in \Cref{def:prediction-basedthresholds}.
\begin{enumerate}[(i)]
  \item \emph{When capacity ratio $M/N$ varies}: if the prediction-based threshold does not coincide with the score-optimal threshold, i.e., $\tau_A \neq \tau^{*A}_{\mathrm{score}}(p_0,\Delta_P)$, the gap in the objective function equals zero for at most one value of $M/N$, corresponding to cases where the optimal threshold coincides with the prediction-based threshold. If the prediction-based threshold coincides with the score-optimal threshold, i.e., $\tau_A = \tau^{*A}_{\mathrm{score}}(p_0,\Delta_P)$, the gap in the objective function equals zero if and only if $\tfrac{M}{N} \leq p_0 + \Delta_P (1-\tau^{*A}_{\mathrm{score}})$.
  \item \emph{When baseline request probability $p_0$ varies}: the gap in the objective function equals zero for at most two isolated points of $p_0$ in $[0, 1-\Delta_P]$, corresponding to cases where the optimal threshold coincides with the prediction-based threshold.
\end{enumerate}
\end{theorem}
Further, we characterize the suboptimality gap as a function of the system capacity ratio $M/N$ and baseline request rate $p_0$ in
Appendix~\ref{sec:proof_suboptimality_prediction_based_thresholds}.

Thus, as \Cref{fig:policies_comparison} shows, without any knowledge of the system parameters, it is unlikely that a provider would choose a prediction-based threshold that
coincides with the score-optimal threshold. The top row shows system performance as $M/N$ varies with $p_0 = 0.1$ and the bottom row shows system performance as $p_0$ varies with capacity fixed $M = 200$. In both cases, the fixed thresholds displayed never coincide with the optimal threshold and the relative gap in system efficacy can be as large as 50\%.

\begin{figure}
    \centering
    \includegraphics[width=\linewidth]{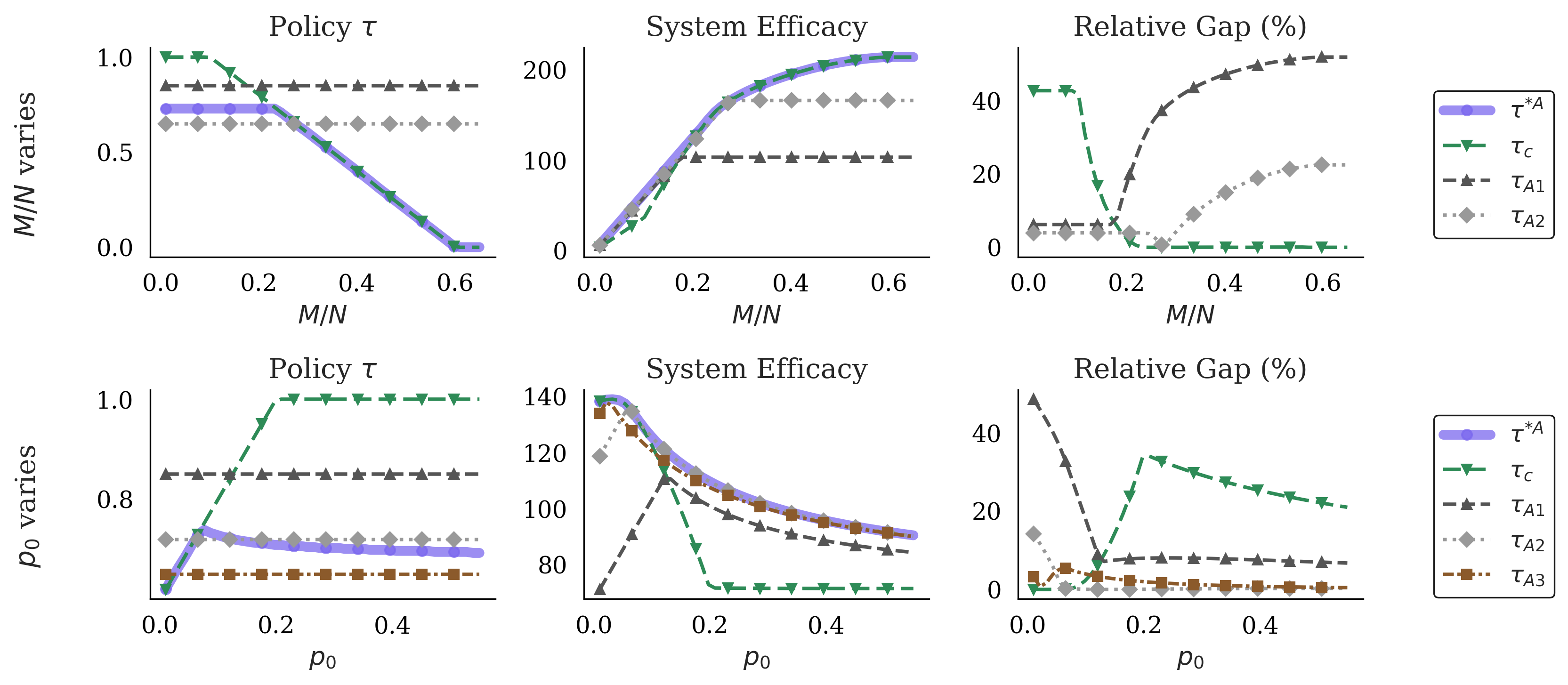}
    \caption{Thresholding behavior and performance under varying capacity and baseline request probability. The top row varies the capacity ratio $M/N$ with the baseline request probability fixed at $p_0 = 0.1$, while the bottom row varies the baseline request probability $p_0$ with capacity fixed at $M = 200$. Each column reports (from left to right) the threshold policy $\tau$, the resulting objective value $\objfluid(\tau)$, and the relative gap $(\objfluid(\tau^*) - \objfluid(\tau))/\objfluid(\tau^*)$. Curves compare the optimal threshold $\tau^*$, the capacity-matching threshold $\tau_c$, and several prediction-based thresholds. Parameters: $N = 1000$, $\Delta_P = 0.5$, and $\hat r = r \sim 0.7\,\mathrm{Beta}(2,10) + 0.3\,\mathrm{Beta}(8,2)$.}
    \label{fig:policies_comparison}
\end{figure}

\subsection{Suboptimality of Capacity-Matching Thresholds}
\label{sec:subopt_capacity}
A more operational-aware threshold is the capacity-matching threshold $\tau_c\!\left(\tfrac{M}{N}\mid p_0,\Delta_P\right)$, which ensures that all capacity is used in expectation.
However, this threshold ignores the cannibalization effect and thus can be suboptimal when cannibalization is strong (proof in Appendix \ref{sec:proof_suboptimality_capacity_matching}; full suboptimality gap in \Cref{lem:gap_capacity_matching_MN,lem:gap_capacity_matching_p0}).

\begin{theorem}[Suboptimality of the capacity-matching threshold]
\label{thm:capacity_threshold_suboptimal}
Consider the capacity-matching threshold $\tau_c\!\left(\tfrac{M}{N}\mid p_0,\Delta_P\right)$. For $p_0 \in [0, 1-\Delta_P]$ and $\tfrac{M}{N} \in [0, p_0+\Delta_P]$, the gap in the objective function equals zero if and only if the optimal threshold coincides with $\tau_c$. This occurs:
\begin{enumerate}[(i)]
  \item \emph{When capacity ratio $M/N$ varies}, $\tfrac{M}{N} \geq p_0 + \Delta_P (1-\tau^{*A}_{\mathrm{score}})$.
  \item \emph{When baseline request probability $p_0$ varies}, $p_0 \leq \bar{p}_0$ where $\bar{p}_0$ is defined in \Cref{def:bar_p0} satisfying
$\bar{p}_0 = \inf\left\{ p_0 \in [0,1-\Delta_P] : \tau^{*A}_{\mathrm{score}}(p_0,\Delta_P) \leq \tau_c\!\left(\tfrac{M}{N}\mid p_0,\Delta_P\right) \right\}$.
\end{enumerate}
The maximum relative gap across operating conditions occurs when $\frac{M}{N} \leq p_0$ (if $M/N$ varies) or when $p_0 = \frac{M}{N}$ (if $p_0$ varies). In both cases, the maximum relative gap equals
\[
\frac{\tilde{R}(\tau^{*A}_{\mathrm{score}}(p_0,\Delta_P)| A, p_0, \Delta_P)-\mathbb{E}[r]}{\tilde{R}(\tau^{*A}_{\mathrm{score}}(p_0,\Delta_P)| A, p_0, \Delta_P)}.
\]
\end{theorem}

The capacity-matching threshold is therefore optimal only when underutilization is the sole concern, that is, when capacity is abundant (large $\tfrac{M}{N}$) or baseline requests are rare (small $p_0$) as shown in \Cref{fig:policies_comparison}. Outside these regimes, cannibalization dominates and $\tau_c$ is suboptimal because it ignores this effect: the gap is zero for $M/N \geq 0.2$ in the top row and $p_0 \leq 0.1$ in the bottom row, but reaches 40\% at $M/N \leq 0.1$ and 35\% as $p_0$ approaches 0.2.
The maximum relative gap admits a natural interpretation: it equals the fractional improvement in average requester efficacy that the score-optimal threshold achieves over independent arrivals.

\subsection{Effectiveness of the Two-Point Threshold Under Prioritization}\label{sec:robustness_prioritization}
In the baseline model, service slots are allocated uniformly at random when demand exceeds capacity, which is the mechanism that induces cannibalization. In practice, providers may instead prioritize high-scoring requesters.  An example of where this may happen is if multiple requesters (e.g. students or patients) request service simultaneously, then priority will be given to the top-ranked individuals. However, requesters may still be allotted service slots in a first-come-first-serve manner, which will naturally induce some randomization in allocation. This raises two questions. First, if low-value requesters can be screened out at the allocation stage, does targeting at the outreach stage still matter? Second, does the two-point threshold of \Cref{thm:two_point_optimality}, derived under random allocation, remain effective when allocation is not random?

We model provider behavior through a mixture mechanism parameterized by $\beta_1 \in [0,1]$, where $\beta_1$ is the share of slots allocated by algorithmic prioritization and $\beta_0 = 1 - \beta_1$ the share allocated at random. This formulation nests three cases of interest: random allocation ($\beta_1 = 0$), full prioritization ($\beta_1 = 1$), and partial trust or visibility ($0 < \beta_1 < 1$), which serves as a tractable first step toward analyzing more complex prioritization rules.

The mechanism operates in two steps. Let $z_i^{(1)}$ and $z_i^{(0)}$ denote binary indicators for whether requester~$i$ is served under the prioritization and random components, respectively; the final service indicator is $z_i = z_i^{(1)} \lor z_i^{(0)}$. In the first step, $\lfloor\beta_1 M\rfloor$ slots are reserved and assigned under \emph{score-based prioritization} to the $\lfloor\beta_1 M\rfloor$ requesters with the highest predicted scores:
\begin{equation}
  \begin{aligned}
    & \mathbb{P}(z_i^{(1)} = 1 \mid d_i = 1) = \mathbf{1}\!\left\{
        \hat{r}_i^A \text{ is among the top }\lfloor\beta_1 M\rfloor
        \text{ of } \{\hat{r}_j^A : d_j = 1\}
      \right\}
  \end{aligned}
\end{equation}
In the second step, the remaining $\lfloor\beta_0 M\rfloor$ slots are distributed uniformly at random among all requesters not served in the first step:
\begin{equation}
  \begin{aligned}
    & \mathbb{P}(z_i^{(0)} = 1 \mid d_i = 1,\, z_i^{(1)} = 0)  = \min\!\left\{
        \frac{\lfloor\beta_0 M\rfloor}{\sum_j d_j(1 - z_j^{(1)})},\; 1
      \right\}.
  \end{aligned}
\end{equation}

Closed-form analysis under the mixture mechanism is intractable, so we evaluate thresholding policies by simulation across prioritization levels $\beta_1 \in [0,1]$ under algorithmic noise $\hat{r}^A = r + \mathcal{N}(0, 0.1^2)$,\footnote{Predicted scores are clipped to $[0,1]$.} computing the optimal threshold by grid search (\Cref{fig:prioritization_threshold_shift}).
\begin{figure}[h]
  \centering
  \includegraphics[width=0.8\linewidth]{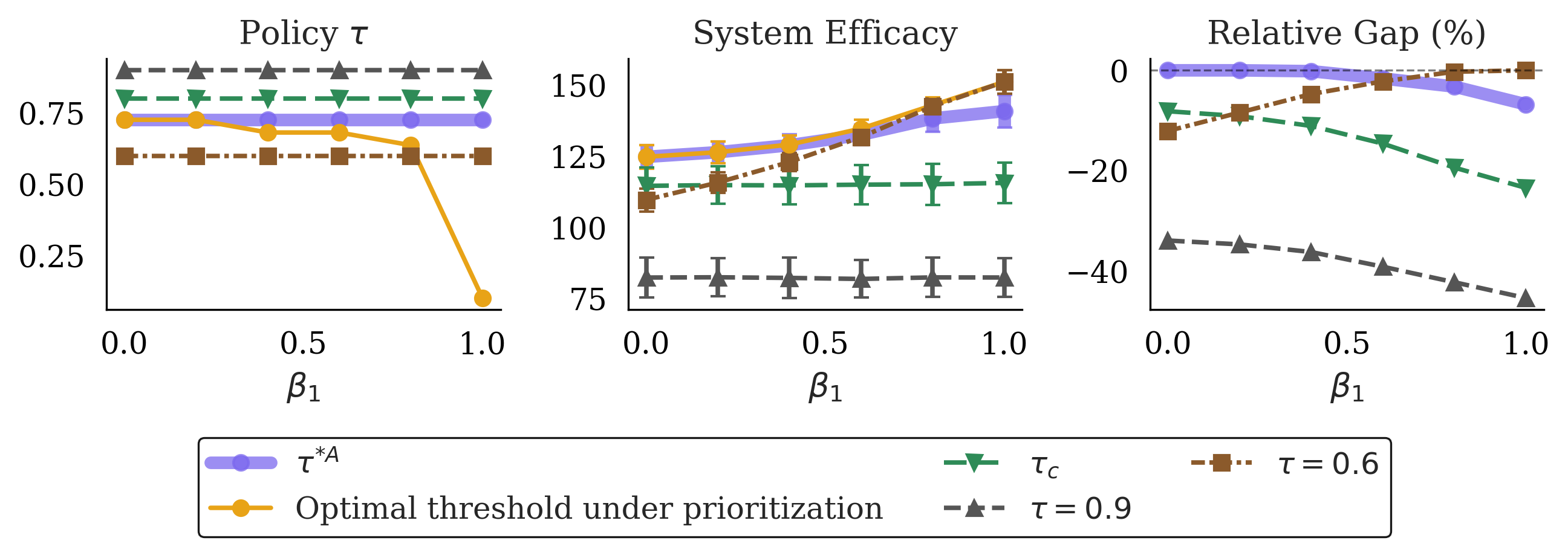}
  \caption{Performance of thresholding policies under score-based prioritization with weight $\beta_1$. Each column reports (from left to right) the threshold policies, system efficacy, and relative gap with respect to the optimal threshold under prioritization, computed via exhaustive grid search over $\tau \in [0,1]$. The two-point threshold derived under random allocation (\Cref{thm:two_point_optimality}) closely tracks the prioritization optimum across a wide range of $\beta_1$, while fixed thresholds fail at one or both extremes. Parameters: $N = 1000$, $M = 200$, $p_0 = 0.1$, $\Delta_P = 0.5$, $r \sim 0.7\,\mathrm{Beta}(2,10) + 0.3\,\mathrm{Beta}(8,2)$.}
  \label{fig:prioritization_threshold_shift}
\end{figure}
The optimal threshold depends on the prioritization weight $\beta_1$, which determines the severity of cannibalization. When $\beta_1$ is large, prioritization screens out low-value requesters, so cannibalization is mitigated and a lower threshold becomes optimal; under full prioritization the optimal threshold can reach $\tau = 0$ (flag everyone), since underutilization is the only remaining concern. When $\beta_1$ is small or intermediate, the random component still allows low-value requesters to cannibalize high-value ones, and an aggressive threshold such as $\tau = 0.6$ wastes capacity on low-value requesters. This answers the first question: targeting at the outreach stage remains essential except at full prioritization ($\beta_1 = 1$), and no fixed threshold, including flag-everyone, is optimal across the full range of $\beta_1$.

In practice, system designers rarely know $\beta_1$: providers may partially trust the algorithm, lack visibility into the scores, or face policy and fairness constraints. The two-point threshold, derived under random allocation ($\beta_1 = 0$), provides a principled hedge against this uncertainty, closely tracking the optimal threshold across a wide range of $\beta_1$ and avoiding the failures of fixed thresholds at either extreme, while the capacity-matching threshold remains suboptimal across the range in our simulations (\Cref{fig:prioritization_threshold_shift}).

\subsection{Implementation Guidance Across Operational Regimes}
\label{sec:implementation-guidance}

The two-point characterization also translates into a practical recipe: once a practitioner identifies the operational regime, the optimal threshold follows directly. \Cref{fig:decision_tree} organizes this guidance as a decision tree, which the practitioner can navigate by answering three questions about the system: Are there independent requests ($p_0 > 0$)? Is service capacity binding ($M < \infty$)? Is the intervention itself capacity-constrained ($M_I < \infty$)?

We discuss three cases of practical interest below. The first two are boundary regimes in which the two-point rule collapses to a single heuristic: no independent request ($p_0 = 0$) and abundant capacity ($M/N \ge p_0 + \Delta_P$). The third, a capacity-constrained intervention ($M_I < \infty$), instead imposes an additional constraint that floors the rule at the outreach budget.

\paragraph{No independent request ($p_0 = 0$).}
When no individual requests service independently, there are no baseline requests to crowd out, and the score-optimal threshold satisfies $\tau^{*A}_{\text{score}}(0,\Delta_P)=1$. The two-point rule reduces to the capacity-matching threshold alone,
\( \tau^{*A}\!\left(\tfrac{M}{N}\,\middle|\,0,\Delta_P\right) = \tau_c\!\left(\tfrac{M}{N}\,\middle|\,0,\Delta_P\right), \)
that is, the planner flags just enough individuals to fill capacity in expectation. In this regime, the naive capacity-matching heuristic is optimal.

\paragraph{Abundant capacity ($M/N \ge p_0 + \Delta_P$).}
When capacity is large enough to serve every possible request, the capacity-matching threshold satisfies $\tau_c(M/N \mid p_0,\Delta_P)=0$, and underutilization is the only binding concern. The optimal policy flags everyone,
\( \tau^{*A}\!\left(\tfrac{M}{N}\,\middle|\,p_0,\Delta_P\right) = 0. \)

\paragraph{Capacity-constrained intervention ($M_I < \infty$).}
Often the intervention itself is limited: a screening program can place only a fixed number of reminder calls per day, and an advising office can reach only a limited number of students per cycle. When at most $M_I$ individuals can be flagged, the feasible thresholds are bounded below,
\( \tau \in [\tau_{c,I},\,1], \, \text{where } \tau_{c,I} = 1 - \tfrac{M_I}{N}, \)
and the two-point threshold is lower bounded at this budget:
\[
\begin{aligned}
\tau^{*A}\!\left(\tfrac{M}{N}\,\middle|\,p_0,\Delta_P,M_I\right) = \max\!\left\{\tau^{*A}\!\left(\tfrac{M}{N}\,\middle|\,p_0,\Delta_P\right),\ \tau_{c,I}\right\},
\end{aligned}
\]
where $\tau^{*A}(M/N \mid p_0,\Delta_P) = \min\{\tau^{*A}_{\text{score}}(p_0,\Delta_P),\, \tau_c(M/N \mid p_0,\Delta_P)\}$ is the unconstrained two-point optimum. The planner follows the two-point optimality but never flags more individuals than the intervention can serve.

\smallskip

Taken together, the results in this section characterize the optimal thresholding policy and establish its practical relevance. \Cref{thm:two_point_optimality} shows that the optimal threshold balances capacity underutilization against cannibalization and takes a simple two-point form, the minimum of the score-optimal and capacity-matching thresholds. \Cref{thm:prediction_based_threshold_suboptimal,thm:capacity_threshold_suboptimal} further establish that neither prediction-based nor capacity-matching thresholds alone achieve optimality across operating conditions: the former ignores operational constraints, while the latter ignores cannibalization when $p_0 > 0$. The two-point characterization unifies both effects by automatically selecting the binding constraint. The robustness analysis in \Cref{sec:robustness_prioritization} strengthens the prescriptive value of this characterization: the two-point threshold remains effective even when providers partially prioritize high-scoring requesters. Finally, we extend the characterization into a practical recipe spanning all operational regimes: depending on whether service capacity, baseline demand, and intervention capacity bind, the optimal threshold is a single rule that a practitioner can read off directly (\Cref{fig:decision_tree}).

%% file: section/algo_selection.tex
\section{Algorithm Selection}\label{sec:algorithm_selection}
We now turn to our second goal: given multiple predictive models that can flag individuals for outreach, which should a provider deploy to maximize system efficacy?
Typically, this selection is driven by metrics of prediction accuracy such as the Mean Squared Error (MSE) for regression problems predicting continuous variables and the Area Under the ROC Curve (AUC) for binary classification problems. Analogous to our findings in Section \ref{sec:thresholding} that prediction-based thresholds are suboptimal, we show in \Cref{sec:misalignment_auc} that algorithm selection based solely on these prediction accuracy metrics is also suboptimal.

As an illustration, we focus on settings in which the outcome of interest is binary
and the true score $r_i$ represents the probability of the positive outcome.
Formally,
\begin{assumption}[Binary Outcome]\label{assump:binary_outcome}
For each subject $i$, the outcome $Y_i \in \{0,1\}$ is a Bernoulli random variable with $\mathbb{P}(Y_i=1 \mid r_i) = r_i$. The true score $r_i$ represents the probability of the positive outcome.
\end{assumption}

\subsection{Misalignment of AUC with System Efficacy}\label{sec:misalignment_auc}

AUC is a widely used metric for evaluating and comparing predictive algorithms. It measures the probability that a randomly chosen positive instance (with $Y=1$) receives a higher score than a randomly chosen negative instance (with $Y=0$). Under \Cref{assump:binary_outcome}, the AUC of algorithm $A$ can be written as
\begin{equation}\label{eq:auc}
  \begin{aligned}
    & \mathrm{AUC}_{A} = \frac{1}{1 - \mathbb{E}(r)}\left[
        \int_0^1 \mathrm{TPR}_{A}(q_A(\tau))\, d\tau
        - \frac{1}{2}\mathbb{E}(r)
      \right].
  \end{aligned}
\end{equation}
The appeal of AUC is intuitive: in many applications, providers face varying capacity across different operating contexts yet must select a single algorithm that performs well under all conditions. For example, in an early warning system within a hospital network, intervention capacity may differ across clinical units. Because the operating threshold may not be known at the time of algorithm selection, AUC's aggregation across all thresholds seems like a natural way to hedge against this uncertainty. We show, however, that this intuition is incorrect: AUC-based selection can diverge from the algorithm that maximizes system efficacy.

Under \Cref{assump:binary_outcome}, the system efficacy is connected to the TPR of algorithm $A$ via the expected true scores among the flagged group:
\begin{equation}\label{eq:expected_true_score_tpr}
  \begin{aligned}
    & \mathbb{E}[r\mid \hat{r}^A \ge {q}_A(\tau)] = \frac{\mathrm{TPR}_{A}\left(q_A\left(\tau\right)\right) \mathbb{E}(r)}{1-\tau},
      \quad \forall \tau \in [0,1].
  \end{aligned}
\end{equation}
AUC is correct in using TPR as a measure of classification quality.
However, AUC diverges in which thresholds it considers: AUC weights every operating point equally, ignoring how capacity constraints and behavioral responses determine which thresholds the system will actually deploy.

\begin{figure}[t]
  \centering
  \includegraphics[width=0.4\textwidth]{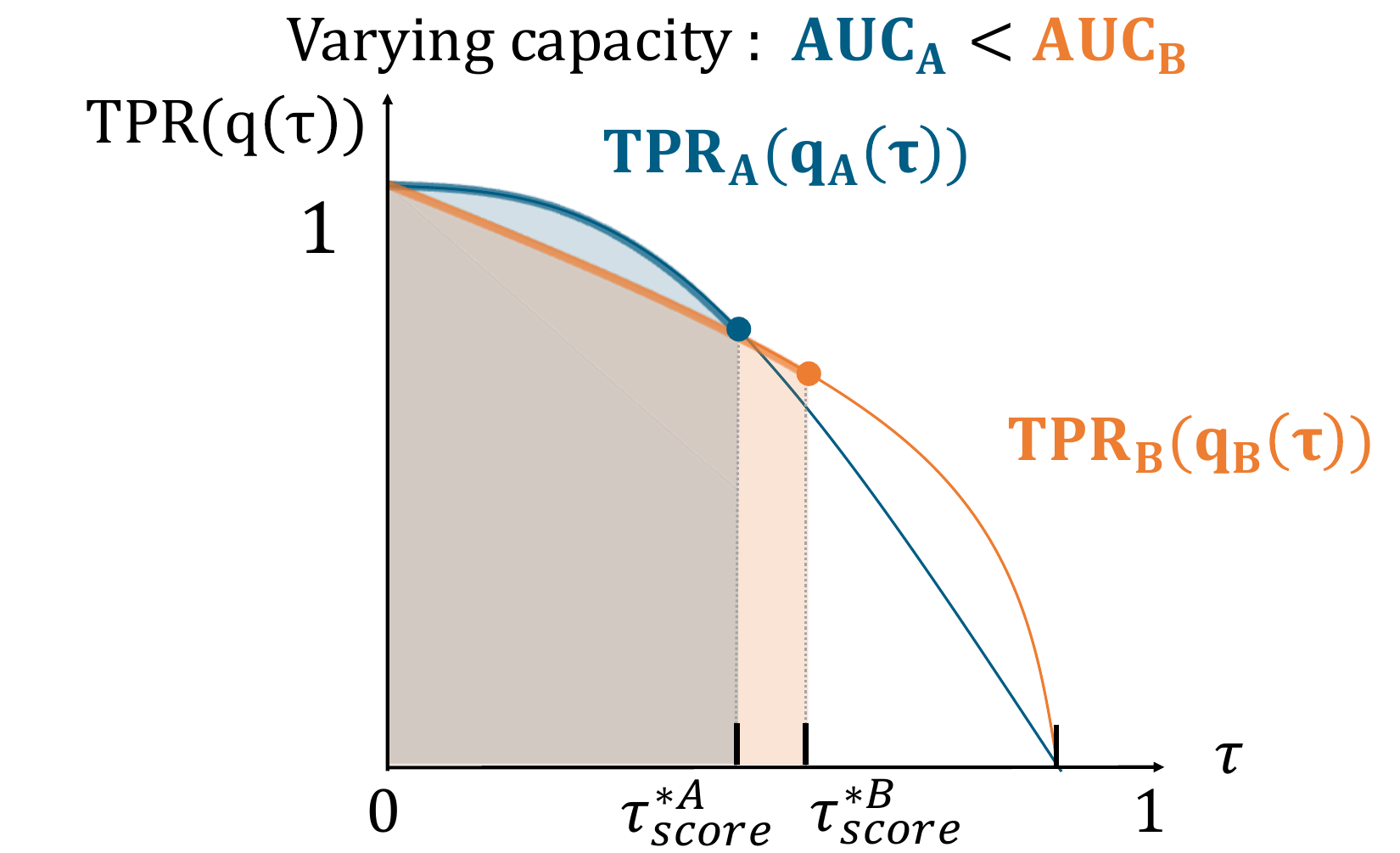}
  \caption{Misalignment between AUC and system efficacy. Algorithm $B$ achieves higher overall AUC, but algorithm $A$ dominates over the operationally relevant threshold range $[0, \tau^{*A}_{\mathrm{score}}]$ (shaded).}
  \label{fig:misalignment_auc}
\end{figure}
To build intuition for this suboptimality, recall that the optimal threshold never exceeds its score-optimal threshold $\tau^{*A}_{\mathrm{score}}(p_0, \Delta_P)$, regardless of the capacity ratio $M/N$. That is, there is a lower bound on the proportion of individuals flagged, even if $M/N$ diminishes to 0. System efficacy therefore depends on TPR only for $\tau \leq \tau^{*A}_{\mathrm{score}}$; performance at higher thresholds is irrelevant because the system never operates there.
The exception is when $p_0 = 0$ and capacity ratios span $[0,\Delta_P]$: $\tau^{*A}_{\mathrm{score}} = 1$ and $\tau^{*A} = \tau_c = 1 - \frac{M}{N\Delta_P}$, so every threshold in $[0,1]$ is deployed.
For example, in \Cref{fig:misalignment_auc}, although algorithm $B$ achieves higher overall AUC, algorithm $A$ achieves higher TPR over the operationally relevant range (shaded region). AUC fails to reflect this, instead crediting algorithm $B$ for performance at thresholds that will never be deployed. Consequently, selecting the highest-AUC algorithm may not maximize system efficacy.

To formalize this observation, we define the \emph{efficacy-optimal algorithm}, which maximizes system efficacy across the distribution of capacity ratios the provider expects to operate at.

\begin{definition}[Efficacy-optimal algorithm]
\label{def:outcome_optimal_varying_capacity}
Given a set of candidate algorithms $\mathcal{A}$ and behavioral parameters $p_0$ and $\Delta_P$, suppose the provider operates under capacity ratios $\rho = M/N$ distributed according to $\mu$.
Under \Cref{assump:joint_density,assump:marginal_positive,assump:conditional_smoothness,assump:monotone_calibration,assump:binary_outcome}, the efficacy-optimal algorithm $A^{*}(\mu \mid p_0,\Delta_P)$ is given by
\[
  \begin{aligned}
    & A^{*}(\mu \mid p_0,\Delta_P) = \arg\max_{A \in \mathcal{A}}\int \objfluid\!\Bigl(
        \tau^{*A}(\rho \mid p_0,\Delta_P) \mid A,\rho,p_0,\Delta_P
      \Bigr)\, d\mu(\rho),
  \end{aligned}
\]
where $\tau^{*A}(\rho \mid p_0,\Delta_P)$ denotes the optimal threshold given $M/N=\rho$ for algorithm $A$.
Under \Cref{assump:binary_outcome}, this simplifies to
\begin{equation*}
  \begin{aligned}
    & A^{*}\left(\mu \mid p_0, \Delta_P\right) = \arg\max_{A \in \mathcal{A}} \int \rho
      \frac{p_0 + \Delta_P\, \mathrm{TPR}_{A}\!\left(q_A\!\left(\tau^{*A}\!\left(\rho\mid p_0, \Delta_P\right)\right)\right)}
           {p_0 + \Delta_P\,\left(1-\tau^{*A}\!\left(\rho\mid p_0, \Delta_P\right)\right)} d\mu(\rho). \\
  \end{aligned}
\end{equation*}
\end{definition}
The efficacy-optimal algorithm depends on how each algorithm performs \emph{at its optimal threshold} across the capacity conditions the system will be deployed. This creates the potential for misalignment with AUC and motivates the development of an \emph{operational} performance metric that aligns algorithm selection with the capacity constraints and resulting thresholds deployed in practice.

\subsection{Alternative Metric for Algorithm Selection (\operationalmetricfull)}\label{sec:opi}
The preceding section suggests a natural alternative to AUC: rather than integrating performance uniformly over all thresholds, we can integrate over the distribution of capacity ratios and evaluate the algorithm at the corresponding optimal threshold.
We call this the \operationalmetricfull{} (\operationalmetric).

\begin{definition}[Operational Metric: \operationalmetric]
Given an algorithm $A$, behavioral parameters $p_0$ and $\Delta_P$, and a distribution $\mu$ over capacity ratios $\rho=M/N$, we define \emph{\operationalmetricfull} to be
\begin{align}\label{eq:opi}
  & \operationalmetric\left(A \mid \mu, p_0, \Delta_P\right) \nonumber = \int \rho \cdot
    \frac{p_0+\Delta_P\,\mathrm{TPR}_{A}\!\left(q_A\!\left(\tau^{*A}(\rho)\right)\right)}
         {p_0+\Delta_P\,\left(1-\tau^{*A}(\rho)\right)}
    \; d\mu(\rho), \nonumber
\end{align}
where $\tau^{*A}(\rho) \equiv \tau^{*A}(\rho\mid p_0, \Delta_P)$, $q_A(\tau)$ is the population $\tau$-quantile of scores from $A$, and $\tau^{*A}\!\left(\cdot\mid p_0, \Delta_P\right)$ is the optimal threshold characterized in \Cref{thm:two_point_optimality}.
\end{definition}
By construction, \operationalmetric{} always identifies the efficacy-optimal algorithm defined in \Cref{def:outcome_optimal_varying_capacity}.

To see how \operationalmetric{} resolves the misalignment illustrated in \Cref{fig:misalignment_auc}, suppose capacity varies uniformly over $[\bar{m}^L, \bar{m}^U]$ and is sufficiently abundant that the system always operates at the capacity-matching threshold, so that $\tau^{*A}(\rho) = \tau_c(\rho \mid p_0, \Delta_P)$. Then \operationalmetricfull{} simplifies to
\[
  \begin{aligned}
    & \operationalmetric\left(A \mid \mu, p_0, \Delta_P\right) = \frac{\Delta_P}{\bar{m}^U - \bar{m}^L}
      \int_{\tau_c(\bar{m}^U)}^{\tau_c(\bar{m}^L)}
      \Bigl( p_0 + \Delta_P\, \mathrm{TPR}_{A}\!\left(q_A(\tau)\right) \Bigr) \, d\tau,
  \end{aligned}
\]
where $\tau_c(\cdot) \equiv \tau_c(\cdot \mid p_0, \Delta_P)$.
If algorithm $A$ achieves higher TPR than $B$ over this operating range, as in \Cref{fig:misalignment_auc}, then $\operationalmetric(A \mid \mu, p_0, \Delta_P) > \operationalmetric(B \mid \mu, p_0, \Delta_P)$: \operationalmetric{} weights only the deployment-relevant thresholds and correctly identifies the efficacy-optimal algorithm despite $A$'s lower AUC.

We note, however, that \operationalmetric{} requires additional inputs: the behavioral parameters $p_0$ and $\Delta_P$, the capacity distribution $\mu$, and the resulting optimal threshold $\tau^{*A}(\cdot)$. Estimating these quantities may require historical operational data. This creates a practical tradeoff: AUC is simpler to compute but ignores operational context, while \operationalmetric{} aligns with system efficacy but requires additional knowledge of the deployment environment. When reliable estimates of behavioral and capacity parameters are available, \operationalmetric{} provides a more accurate criterion for algorithm selection.

%% file: section/case_study.tex
\section{Case Study: Healthcare Sepsis Early Warning System}\label{sec:case_study}
We apply our framework to a sepsis early warning program using data from NewYork-Presbyterian (NYP).
Recall that \Cref{sec:thresholding} establishes that prediction-based thresholds are suboptimal in the presence of cannibalization and capacity constraints. However, early warning systems in practice often use prediction-based thresholds that ignore capacity constraints in the clinical workflow. Moreover, AUC remains the primary metric used to evaluate predictive algorithms, despite known limitations in the early-warning context \citep{romero2015c}. In this case study, we first quantify the performance gap between prediction-based thresholds and the optimal thresholds characterized in \Cref{thm:two_point_optimality} using historical data. We then demonstrate that an algorithm with higher AUC can yield lower system efficacy when deployed at optimal thresholds across varying capacity conditions.

\subsection{Dataset}

Our dataset consists of de-identified electronic health record (EHR) data from 85{,}160 patient encounters across 8 adult acute care hospitals in the New York Presbyterian (NYP) hospital system, comprising over 1.8 million time-stamped clinical observations on 70{,}067 unique patients. All encounters were retrospectively assigned periodic risk scores according to the
Epic Sepsis Model \citep{wong2021external}.
The hospital system was considering whether to utilize this proprietary score to initiate alerts for an early warning system (EWS) deployed at NYP, where current practice uses rule-based SIRS criteria to trigger alerts \citep{chang2025rapid}.
For each observation, the dataset includes a rich set of features drawn from both dynamic and static sources: vital signs, laboratory values, medication administration records, patient demographics, admission characteristics, and a comprehensive set of 29 comorbidity indicators. We define sepsis according to the Centers for Disease Control and Prevention (CDC) Adult Sepsis Event (ASE) criteria, which combine the timing of organ dysfunction with clinical suspicion of infection (blood culture or antibiotic order).  Full details on data preparation, target definition, and the complete covariate list are provided in Appendix \ref{app:case_study}.

\subsection{Experimental Setup}
NYP operates an alert workflow in which the EWS notifies a physician assistant, who can check on only a limited number of patients per shift. We evaluate the candidate version of this workflow that the hospital system was considering, in which alerts are initiated by an EHR-based predictive risk score.
When a patient's risk score exceeds a threshold, the patient is added to an electronic watch list monitored by a physician assistant, who can only monitor a limited number of patients per shift.

This workflow maps directly to our model. The population consists of $N$ patients monitored during a shift, each receiving a risk score $\hat{r}$; the flagging threshold $\tau$ determines which patients trigger an alert. The capacity-constrained resource is the physician assistant, who can monitor at most $M$ patients per shift, shared between new alerts and ongoing-treatment monitoring on a first-come, first-served basis.\footnote{In practice, the physician assistant may prioritize patients with ongoing treatment over new alerts. Under such prioritization, cannibalization effects would be stronger, making our model's estimates an upper bound on system efficacy.} Given the importance of timely treatment, we consider a successful outcome as whether a patient who needs sepsis treatment begins treatment within the current shift (8 hours). Patients who need treatment will always eventually receive it, but there may be treatment delays when capacity is constrained and the physician assistant is unable to check on all patients (flagged or independently identified) within the given shift. However, given the clear evidence that delays in sepsis treatment are associated with worse outcomes (e.g. \citet{seymour2017time,liu2017timing}), we focus on `success' being defined as timely treatment -- i.e., treatment beginning within shift.

The funnel proceeds in three stages: (i) the algorithm flags high-risk patients based on the threshold; (ii) alerts nudge the physician assistant to check on flagged patients, subject to capacity; and (iii) patients who are checked receive timely treatment. Importantly, the nudge is directed at the provider: an alert increases the probability that the physician assistant checks on a patient. The parameter $p_0$ captures this baseline probability that the navigator physician assistant is prompted (e.g. by another clinician) to check on a patient requiring treatment independently of any algorithm-generated alert. The parameter $\Delta_P$ represents the incremental lift from being flagged, so flagged patients are checked in with probability $p_0 + \Delta_P$.

\paragraph{Predictive Algorithms}
We use the dataset to predict sepsis onset within the next 8 hours to allow sufficient time for clinical response while avoiding excessive noise from longer-term predictions. We compare two sepsis prediction algorithms: the Epic Early Detection of Sepsis V2 model \citep{wong2021external}, a widely available sepsis early warning score for anyone using the Epic Systems Corporation, Verona, Wisconsin Electronic Health Record (EHR), and an XGBoost model that we train on EHR data (details in Appendix~\ref{app:model}).\footnote{The XGBoost model is not intended as a clinical recommendation. We use it as an example to show that AUC-based selection can diverge from efficacy-optimal selection.}
The Epic model achieves an AUC of 0.826, while the XGBoost model achieves an AUC of 0.806 (\Cref{fig:sepsis_auc}).

\begin{figure}
    \centering
    \begin{subfigure}{0.32\linewidth}
        \centering
        \includegraphics[width=\linewidth]{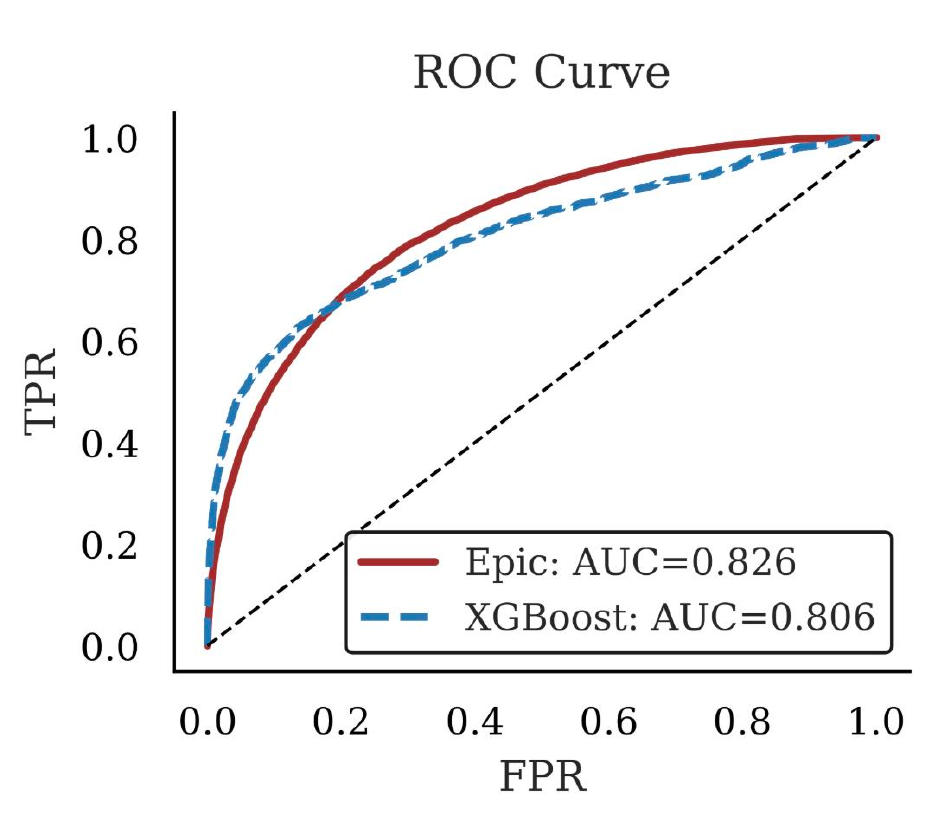}
        \caption{ROC Curve}
        \label{fig:sepsis_auc_roc}
    \end{subfigure}
    \begin{subfigure}{0.32\linewidth}
        \centering
        \includegraphics[width=\linewidth]{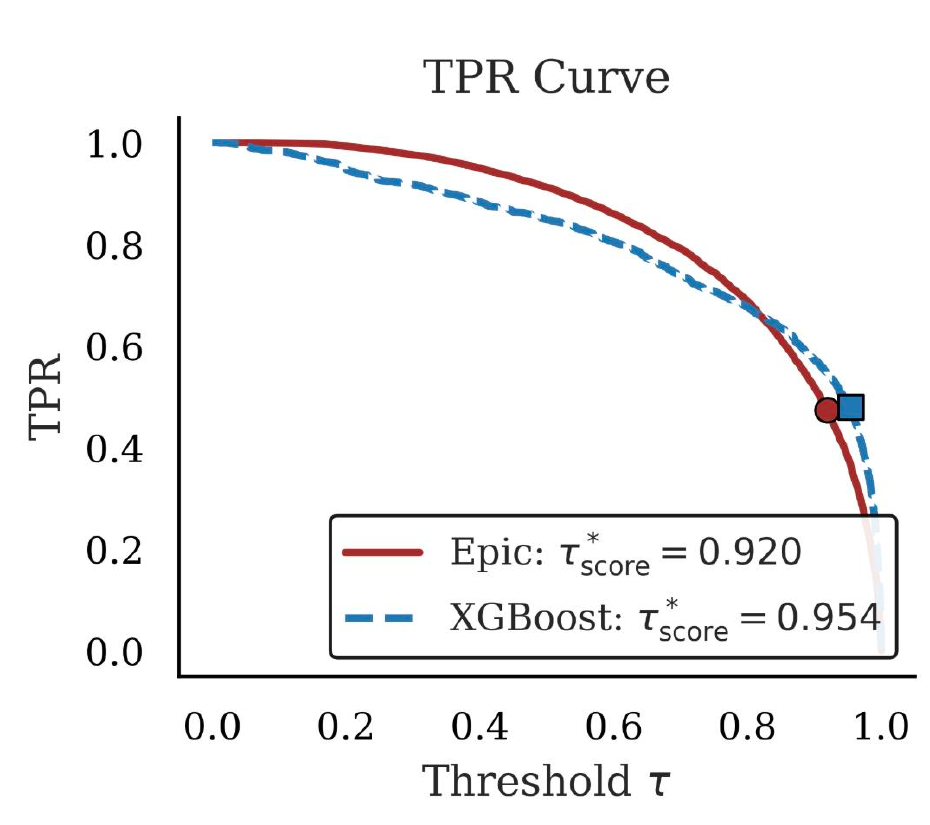}
        \caption{TPR Curve}
        \label{fig:sepsis_auc_tpr}
    \end{subfigure}
    \caption{Comparison of Epic and XGBoost sepsis prediction models. Panel~(\subref{fig:sepsis_auc_roc}) shows ROC curves, with Epic achieving higher AUC ($0.826$ vs.\ $0.806$). Panel~(\subref{fig:sepsis_auc_tpr}) shows TPR as a function of threshold $\tau$, illustrating that the curves cross. The score-optimal threshold is $\tau^*_{\mathrm{score}} = 0.954$ for XGBoost and $\tau^*_{\mathrm{score}} = 0.920$ for Epic.}
    \label{fig:sepsis_auc}
\end{figure}

\paragraph{Estimation of Operational and Behavioral Quantities}
We estimate operational parameters through consultation with our clinical collaborator at NYP. A physician assistant sepsis navigator typically covers $N = 200$ patients during a shift. The baseline probability that a patient needs to be checked for missing treatment by the physician assistant is $p_0 = 0.1$. Flagging a patient increases their check-in probability to $p_0 + \Delta_P = 0.6$. Since the predictive algorithms predict sepsis onset within 8 hours, service is considered received if and only if the physician assistant checks in the patient within 8 hours of receiving an alert. Given the behavioral parameters, we calculate the score optimal thresholds for each algorithm: $\tau^*_{\mathrm{score}} = 0.954$ for XGBoost and $\tau^*_{\mathrm{score}} = 0.920$ for Epic.

\subsection{Simulation}
We simulate the performance of different thresholding policies and predictive algorithms as follows. For each trial, we sample $N = 200$ patients from the held-out test set to form a cohort. For each patient, we record their predicted risk score from each algorithm and whether they develop sepsis within the next 8 hours. We then apply different thresholding policies to determine which patients are flagged by the EWS. The physician assistant can check  at most $M$ patients per shift; we vary $M \in [10, 100]$ to examine performance across different capacity levels. We average over 8000 iterations to estimate the expected number of true positive sepsis cases identified under each policy and algorithm. The results are shown in \Cref{fig:sepsis}.

\paragraph{Threshold Selection.}
We first consider threshold selection for a fixed algorithm $A$. \Cref{fig:sepsis_epic,fig:sepsis_xgb} compare the efficacy-optimal threshold characterized in \Cref{thm:two_point_optimality} with the capacity-matching threshold and two prediction-based thresholds ($\tau = 0.6$ and $\tau = 0.8$) as $M$ varies.
The TPR and FPR for these prediction-based thresholds are 0.66 and 0.19 at $\tau = 0.8$, and 0.78 and 0.39 at $\tau = 0.6$, and their selection would reflect a hospital that targets a minimum sensitivity while limiting alert fatigue.
For both algorithms, prediction-based thresholds incur relative gaps of up to 40\%.
Because efficacy scales with the number of slots when capacity binds, a fixed threshold would need roughly 67\% more nurse capacity to match what the optimal threshold achieves with the slots already available.
The gap is non-monotonic in $M$: it closes when prediction-based thresholds coincide with the capacity-matching threshold, which is optimal in this regime.
As capacity increases further, the gaps increase due to underutilization: prediction-based thresholds become too conservative, leaving capacity unfilled. The capacity-matching threshold performs well when capacity is high, but becomes suboptimal in low-capacity regimes where cannibalization effects dominate.

\begin{figure}[t]
    \centering
    \begin{subfigure}{0.32\linewidth}
        \centering
        \includegraphics[width=\linewidth]{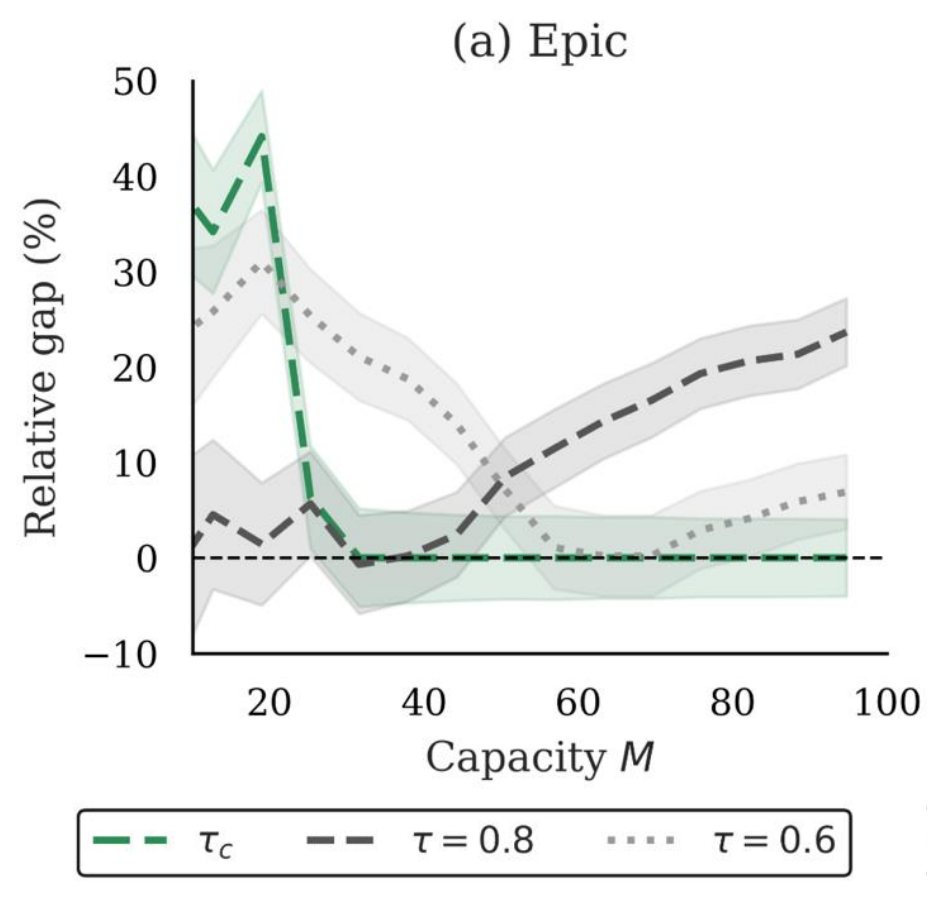}
        \phantomcaption
        \label{fig:sepsis_epic}
    \end{subfigure}
    \hfill
    \begin{subfigure}{0.32\linewidth}
        \centering
        \includegraphics[width=\linewidth]{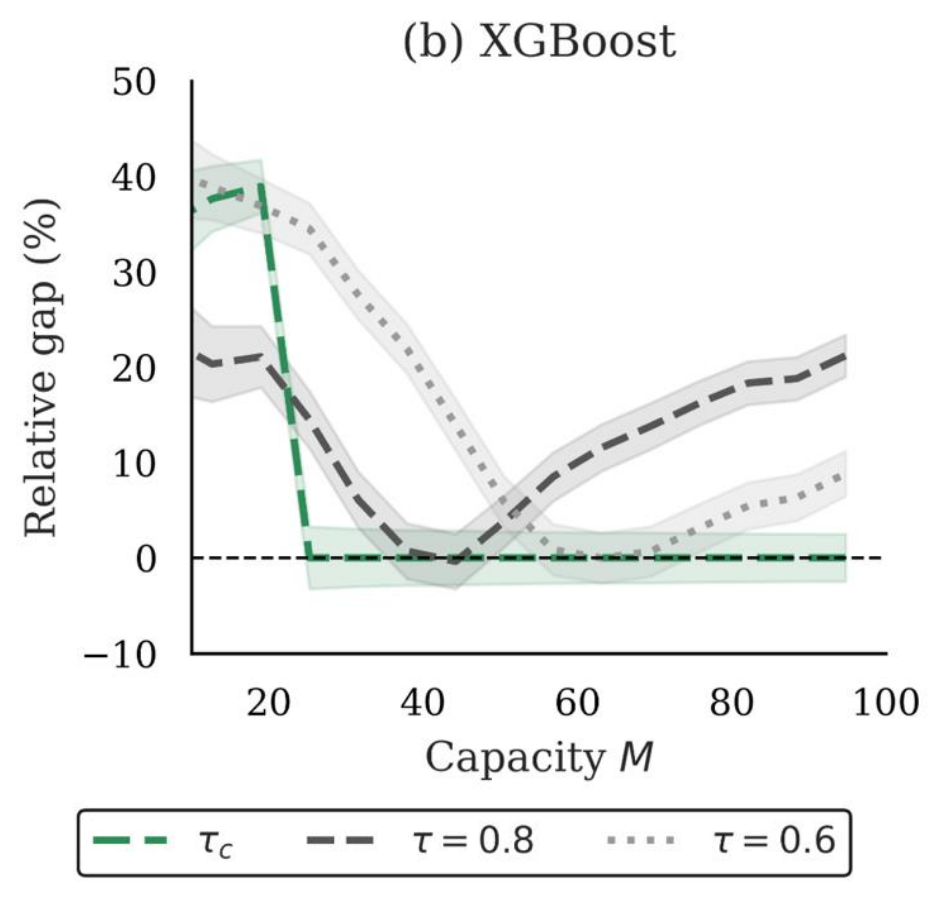}
        \phantomcaption
        \label{fig:sepsis_xgb}
    \end{subfigure}
    \hfill
    \begin{subfigure}{0.32\linewidth}
        \centering
        \includegraphics[width=\linewidth]{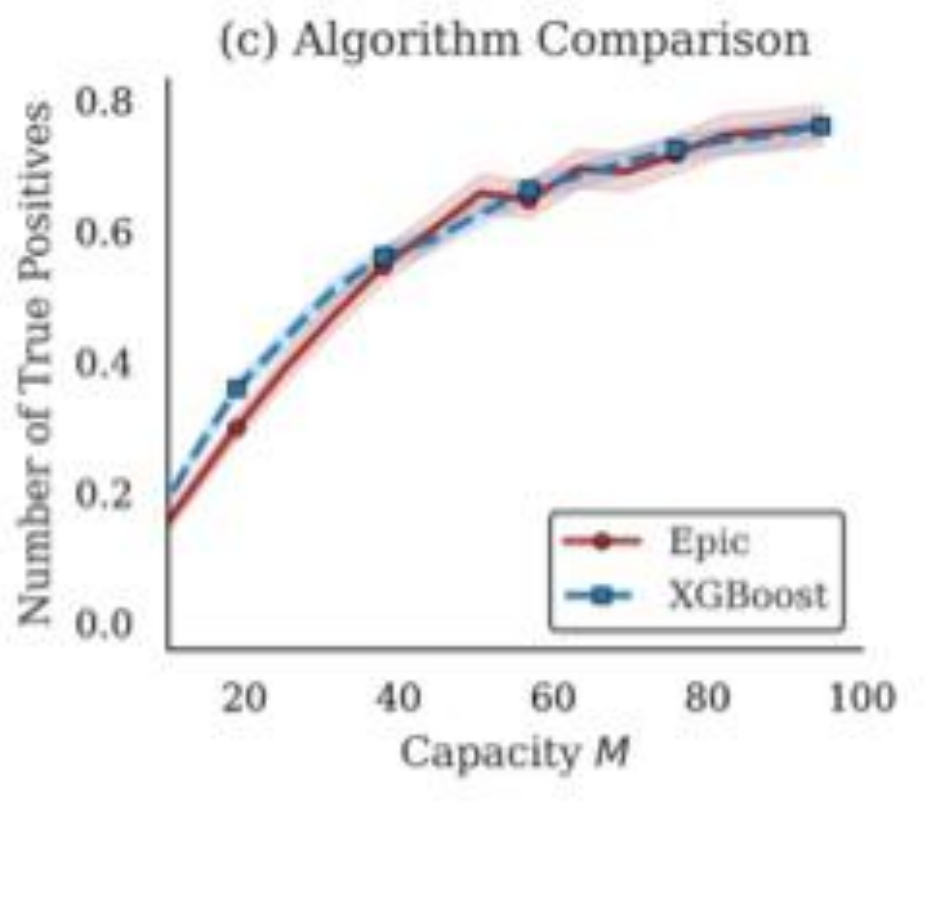}
        \phantomcaption
        \label{fig:sepsis_compare}
    \end{subfigure}
    \caption{Threshold and algorithm comparison under varying capacity. Panels~(\subref{fig:sepsis_epic}) and~(\subref{fig:sepsis_xgb}) report relative gaps in system efficacy under different thresholding policies, compared to the efficacy-optimal threshold, for Epic and XGBoost respectively; suboptimal thresholds can lead to efficacy losses of up to $40\%$. Panel~(\subref{fig:sepsis_compare}) compares algorithms under optimal thresholding: despite having lower overall AUC, XGBoost achieves higher system efficacy than Epic when $M \in [10,30]$, with differences exceeding one standard error. Parameters: $N=200$, $p_0=0.1$, $\Delta_P=0.5$. Shaded regions indicate $\pm$ one standard error across trials.}
    \label{fig:sepsis}
\end{figure}

\paragraph{Divergence between AUC and \operationalmetric{}.} We next consider algorithm selection, with each algorithm deployed at its respective efficacy-optimal threshold. Panel \subref{fig:sepsis_compare} of \Cref{fig:sepsis} reports realized system efficacy, measured by the expected number of true positive sepsis cases identified, as a function of capacity $M$.
Despite its lower AUC ($0.806$ vs.\ $0.826$), the XGBoost model outperforms the Epic model when capacity is limited ($M \in [10, 30]$, corresponding to $\tfrac{M}{N} \in [0.05, 0.15]$). Over this range, XGBoost achieves an average improvement of 13.1\% and a maximum improvement of 22.8\% in system efficacy. \operationalmetric{} captures this reversal: under an operating capacity distributed as $\mu \sim \mathrm{Uniform}(0.05, 0.15)$, \operationalmetric{} is higher for XGBoost than for Epic ($0.39$ vs.\ $0.34$), correctly identifying the efficacy-optimal algorithm. The mechanism behind this is shown in \Cref{fig:sepsis_auc_tpr}. Although Epic dominates XGBoost in overall AUC (\Cref{fig:sepsis_auc_roc}), XGBoost attains higher TPR than Epic between the TPR crossing point and the respective score-optimal thresholds (\Cref{fig:sepsis_auc_tpr}), which is precisely the operating range induced by low and mid capacity. AUC, which weights all thresholds equally, averages away this difference while \operationalmetric{} weights only the deployment-relevant thresholds and recovers the correct algorithm choice.

\paragraph{Effectiveness of the Two-Point Threshold Under Prioritization.}

The baseline simulation assumes that the sepsis navigator physician assistant allocates capacity uniformly at random among requesters. In practice, physician assistants may partially or fully prioritize patients. We examine whether the two-point optimal threshold remains effective under such prioritization, following the mixture allocation mechanism introduced in \Cref{sec:robustness_prioritization}. \Cref{fig:sepsis_prioritization} reports the relative gap in system efficacy under prioritization for $\beta_1 \in \{0.5, 1.0\}$, where the optimal threshold under prioritization is computed via grid search over $\tau \in [0,1]$. Consistent with the analysis in \Cref{sec:robustness_prioritization}, the two-point threshold derived under random allocation tracks the prioritization optimum across both algorithms and a wide range of capacity levels. The only exception is at very low capacity, where the prioritization stage filters out low-value requesters on its own and a lower threshold becomes optimal. In contrast, fixed prediction-based thresholds remain substantially suboptimal under prioritization, with losses of up to 50\% at low capacity ratios regardless of $\beta_1$.

\begin{figure}[h]
\centering
\includegraphics[width=\linewidth]{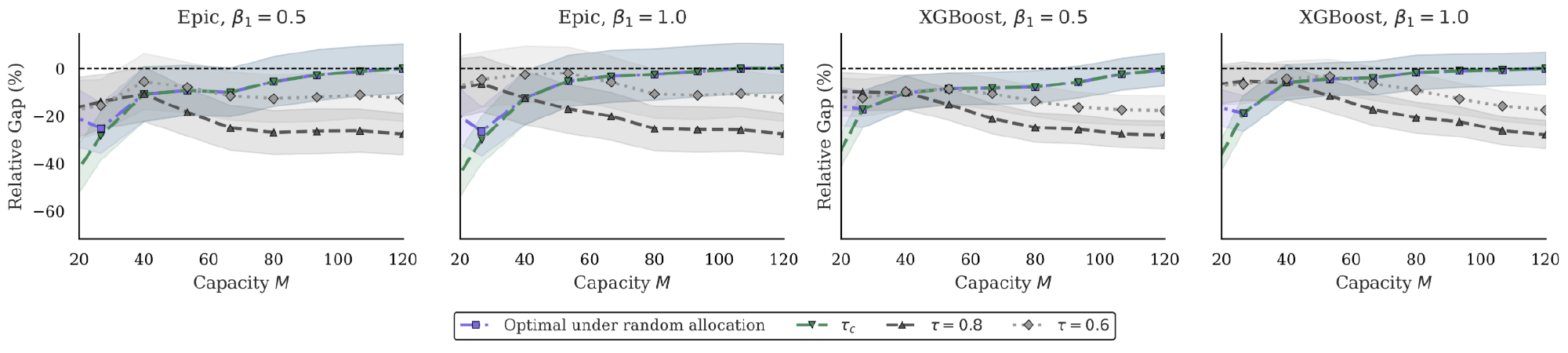}
\caption{
Robustness of thresholding policies to score-based prioritization in the NYP sepsis case study. The left panel corresponds to $\beta_1 = 0.5$ and the right panel corresponds to$\beta_1 = 1.0$. The optimal threshold under random allocation remains near-optimal across all capacity levels and both algorithms, where the optimal threshold under prioritization is computed via exhaustive grid search over $\tau \in [0,1]$. Fixed prediction-based thresholds ($\tau = 0.8$, $\tau = 0.6$) incur losses up to 50\% at low capacity ratios, irrespective of $\beta_1$. Parameters: $N = 200$, $p_0 = 0.1$, $\Delta_P = 0.5$. Shaded regions indicate $\pm$ one standard error across trials.
}
\label{fig:sepsis_prioritization}
\end{figure}

These results have two practical implications. First, threshold selection should account for capacity constraints and behavioral responses. Prediction-based thresholds that ignore operational realities can incur substantial efficacy losses, up to 40\% in our simulations. Second, algorithm selection should be guided by performance over the capacity-induced operating range. In capacity-constrained settings, an algorithm with lower AUC may lead to higher system efficacy. Deployment-aware metrics such as \operationalmetric{} provide a principled basis for navigating this tradeoff.

%% file: section/extension.tex
\section{Extensions}\label{sec:extensions}
Our baseline model makes several simplifying assumptions to yield a clean presentation of the optimal policy. However, the results and core insights extend to richer settings beyond the baseline. In this section, we discuss two
extensions of our framework to capture heterogeneous baseline request probability and heterogeneous treatment effects.

\subsection{Heterogeneous Baseline Request Probability}
Our baseline model assumes a homogeneous baseline request probability $p_0$ across the population. In many applications this probability varies across subpopulations, and its relationship with true risk shapes the \emph{composition} of the requester pool and hence the optimal targeting policy. To model this setting, we partition the population into $K$ subpopulations and let $\pi^k$, $\mu^k$, and $p_0^k$ denote the proportion, average score, and baseline request probability of subpopulation $k$. The following corollary characterizes the optimal threshold in this setting (the proof follows directly from \Cref{thm:two_point_optimality}).

\begin{corollary}
\label{cor:heterogeneous_p0}
Under \Cref{assump:joint_density,assump:marginal_positive,assump:conditional_smoothness,assump:monotone_calibration}, the optimal flagging threshold under heterogeneous baseline request probabilities $\{p_0^k\}_{k=1}^K$ is given by
\begin{equation*}
\begin{aligned}
\tau^{*A} = \min \Bigl\{
& \tau^{*A}_{\mathrm{score}}\bigl(\{p_0^k\}_{k=1}^K, \Delta_P\bigr), \tau_c\bigl(\tfrac{M}{N} \mid \{p_0^k\}_{k=1}^K, \Delta_P\bigr) \Bigr\},
\end{aligned}
\end{equation*}
where $\tau^{*A}_{\mathrm{score}}(\{p_0^k\}_{k=1}^K, \Delta_P)$ denotes the score-optimal threshold
\[
  \begin{aligned}
     \tau^{*A}_{\mathrm{score}}(\{p_0^k\}_{k=1}^K, \Delta_P)
    = & \argmax_{\tau} \tilde{R}\left(\tau \mid A, \{p_0^k\}_{k=1}^K, \Delta_P\right) \\
    = & \argmax_{\tau} \frac{\sum_{k=1}^K \pi^k p_0^k \mu^k}{\sum_{k=1}^K \pi^k p_0^k + \Delta_P\,(1-\tau)}  +\frac{\Delta_P\,(1-\tau)\,\mathbb{E}[r\mid \hat{r}^A \ge q_A(\tau)]}
      {\sum_{k=1}^K \pi^k p_0^k + \Delta_P\,(1-\tau)},
  \end{aligned}
\]
and $\tau_c(\tfrac{M}{N}\mid \{p_0^k\}_{k=1}^K,\Delta_P)$ denotes the capacity-matching threshold
\[
  \tau_c\!\left(\tfrac{M}{N}\,\Big|\,\{p_0^k\}_{k=1}^K,\Delta_P\right)
  = \min\biggl(1, \max\Bigl(0,\, 1 - \frac{1}{\Delta_P}\Bigl(\tfrac{M}{N} - \sum_{k=1}^K \pi^k p_0^k\Bigr)\Bigr)\biggr).
\]
\end{corollary}

Importantly, the heterogeneity in $p_0^k$ influences the optimal policy entirely
through the cannibalization effect. Specifically, the subpopulation structure
enters the per-slot efficacy $\tilde{R}$ in \Cref{cor:heterogeneous_p0} through
the term $\sum_{k=1}^K \pi^k p_0^k \mu^k$, the average value of the independent
requester pool, which increases with the correlation between $p_0^k$ and
$\mu^k$. When the correlation is positive, the independent requester pool is
high-value, and the planner can afford to be selective in targeting without
being cannibalized by low-value independent requests. When the correlation is
negative, the independent pool is low-value while non-requesting subpopulations
are relatively high-value, so the planner optimally lowers the threshold to
reach these high-value individuals before capacity is largely absorbed by less
productive independent requests.

\Cref{fig:het_p0_threshold_comparison} illustrates these effects for a population of $N=1000$ individuals partitioned into $K=3$ equal-sized subpopulations ($\pi^k=1/3$), with $M=175$ and $\Delta_P=0.5$. We fix the population averages $\bar{p}_0=0.1$ and $\bar{\mu}=0.35$, let $p_0^k$ range over $[0,0.2]$ and $\mu^k$ over $[0.175,0.525]$, and vary the correlation $\mathrm{corr}(p_0,\mu)\in \{-0.9,0,0.9\}$ to capture settings in which high baseline request groups are low-value, value-neutral, or high-value. As $\mathrm{corr}(p_0,\mu)$ increases, cannibalization weakens and the optimal threshold $\tau^{*A}$ rises. The gap between $\tau^{*A}$ and the capacity-matching threshold $\tau_c$ correspondingly narrows, since the suboptimality of $\tau_c$ is driven precisely by its neglect of cannibalization, while the prediction-based thresholds move closer to or further from the optimum depending on whether they sit above or below $\tau^{*A}$.

\begin{figure}[h]
    \centering
    \includegraphics[width=0.8\linewidth]{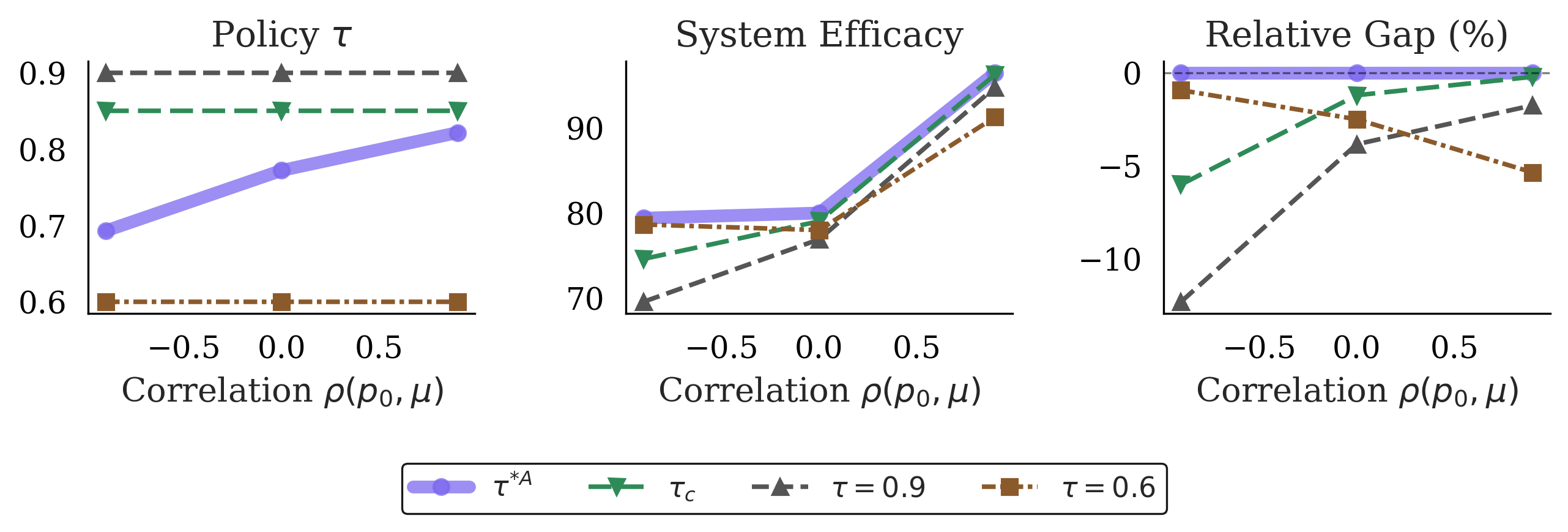}
    \caption{
    Performance of threshold strategies as a function of the correlation $\rho(p_0, \mu)$ between baseline request probability and mean risk score across $K=3$ subpopulations. Parameters: $N = 1000$, $M = 175$, $\bar{p}_0 = 0.1$, $\Delta_P = 0.5$, $r \sim 0.7\,\mathrm{Beta}(2,10) + 0.3\,\mathrm{Beta}(8,2)$, $\hat{r} = r + \mathcal{N}(0, 0.1^2)$.}
    \label{fig:het_p0_threshold_comparison}
\end{figure}

\subsection{Heterogeneous Treatment Effects}
\label{subsec:heterogeneous_te}
A second extension allows the lift $\Delta_P$ from being flagged to vary across individuals, capturing heterogeneous treatment effects (HTE). In many applications, the effectiveness of an alert varies with individual characteristics: in healthcare, patients' responsiveness to outreach may depend on their medical conditions; in education, students' responsiveness to advising may depend on their academic standing. When such characteristics also shape the underlying risk, the predicted score can be informative about treatment response.\footnote{The predicted score may not capture all factors shaping treatment response, and some heterogeneity may be orthogonal to it. We focus on heterogeneity that varies with the predicted score, which permits a tractable extension of our framework.} We extend our framework to capture such treatment effects that vary with the predicted score.

One might expect that when responsiveness increases with predicted risk, cannibalization is mitigated.
We show that this intuition is only partially correct:
favorable alignment does raise overall system efficacy, but even in this case, the most favorable one for operationally naive thresholds, the relative gap from using such a threshold remains substantial. When responsiveness instead declines with predicted risk, cannibalization only intensifies, so operationally aware thresholding matters even more outside this favorable case.

We model such heterogeneous treatment effects by adding cutoffs on the predicted score that determine the lift from being flagged. Let $0 = \lambda_0 < \lambda_1 \leq \cdots \leq \lambda_K < \lambda_{K+1} = 1$ partition the predicted-score distribution into $K+1$ regions, and suppose an individual with predicted score $\hat{r}^A \in [q_A(\lambda_k), q_A(\lambda_{k+1}))$ receives lift $\Delta_P^{(k)}$ when flagged.
To aid tractability, we focus on the monotone regime $\Delta_P^{(0)} \leq \Delta_P^{(1)} \leq \cdots \leq \Delta_P^{(K)}$, in which higher-predicted-risk individuals are more responsive to intervention. Under monotonicity, the two-point structure of \Cref{thm:two_point_optimality} extends region by region: we show in Appendix~\ref{app:hte} that there exists an optimal \emph{active region} $k_0$ such that the optimal threshold is $\tau^{*A}_{\mathrm{HTE}} = \min\{\tilde{\tau}_{c,k_0}, \tilde{\tau}^*_{\mathrm{score},k_0}\}$, where the two thresholds are within-region analogues of the capacity-matching and score-optimal thresholds (\Cref{thm:hte_optimality}). When the optimal threshold lies in region $k_0$, all higher-responsiveness regions are fully flagged and contribute a fixed volume of requests: $\tilde{\tau}_{c,k_0}$ fills the capacity remaining after these committed requests, $\tilde{\tau}^*_{\mathrm{score},k_0}$ maximizes per-request efficacy given them, and the minimum reflects the same trade-off between underutilization and cannibalization. The optimal threshold can thus be found by solving at most $K+1$ two-point problems, searching from the top region downward until the within-region optimum is feasible.

We examine the role of heterogeneity by tracing how the optimal threshold and the performance of naive thresholds vary with the level of heterogeneity. Consider a population of $N = 1000$ individuals partitioned into $K + 1 = 3$ equal-mass regions of the predicted-score distribution, with $(\lambda_0, \lambda_1, \lambda_2, \lambda_3) = (0, 1/3,2/3,1)$. The baseline request probability is $p_0 = 0.1$, and we fix the average per-region lift at $\bar{\Delta}_P = 0.4$. To vary the magnitude of heterogeneity we parametrize the region-specific lifts as
\begin{equation}
  \begin{aligned}
    & \Delta_P^{(k)} - \bar{\Delta}_P= \left(\frac{2k}{K} -1\right)\eta,
      \qquad k = 0, 1, \ldots, K.
  \end{aligned}
  \label{eq:hte_linear_param}
\end{equation}
The case $\eta = 0$ recovers the homogeneous baseline, and larger $\eta$ indicates higher responsiveness in higher-predicted-risk regions.

\Cref{fig:hte_threshold_comparison} reports the system efficacy of the thresholding strategies as the HTE magnitude $\eta$ and capacity $M$ vary. The main takeaway is that heterogeneity aligned with predicted risk helps every strategy but does not remove the need for operationally aware thresholding: a larger $\eta$ raises the system efficacy of all thresholds, yet it does not always close the relative gap between the naive thresholds and the optimum. The optimal threshold $\tau^{*A}_{\mathrm{HTE}}$ rises with $\eta$ in both capacity regimes, through different mechanisms. When cannibalization binds ($M = 150$, top row), responsiveness concentrated in the highest-predicted-risk regions makes targeted requests high-value, so cannibalization is mitigated and a more selective threshold suffices. When underutilization binds ($M = 300$, bottom row), higher responsiveness generates more targeted requests per flag, so capacity fills at a higher threshold.
\begin{figure}[h]
    \centering
    \includegraphics[width=\linewidth]{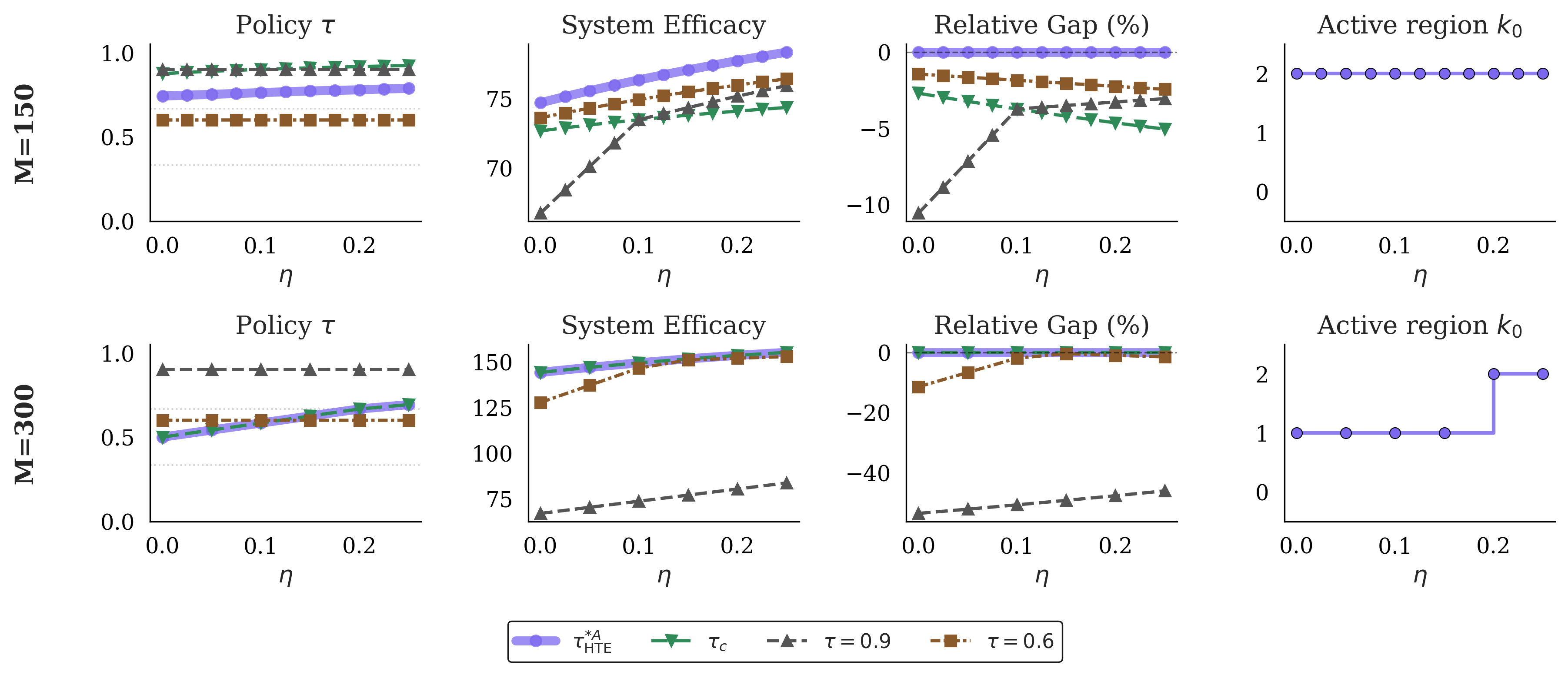}
    \caption{Performance of threshold strategies as a function of the HTE magnitude $\eta$. Parameters: $N = 1000$, $p_0 = 0.1$, $K + 1 = 3$ equal-mass regions with $(\lambda_0, \lambda_1, \lambda_2, \lambda_3) = (0, 1/3,2/3,1)$, $r \sim 0.7\,\mathrm{Beta}(2,10) + 0.3\,\mathrm{Beta}(8,2)$, $\hat{r} = r + \mathcal{N}(0, 0.1^2)$.}
    \label{fig:hte_threshold_comparison}
\end{figure}

The framework further extends to \emph{bilevel interventions}, where outreach is delivered at two intensities (e.g., a phone call versus a text reminder) with separate lifts and capacities; the optimal pair of thresholds retains the same two-point structure (Appendix~\ref{sec:bilevel}).

%% file: section/discussion.tex
\section{Discussion}

AI-assisted interventions in service delivery systems share a common operational structure: an algorithm produces a score, an alert nudges flagged individuals or providers, and a downstream service of limited capacity determines who is ultimately served. Standard deployment practice focuses on the first stage of this pipeline, building a more accurate predictor, and assumes that this will automatically lead to better outcomes. Our results show otherwise: both threshold selection and algorithm choice must account for the operational funnel from prediction to service delivery.

A practical strength of our framework is that it is immediately actionable. Practitioners can improve system efficacy using existing AI tools simply by adjusting thresholds based on capacity and behavioral parameters, without retraining. When choosing among candidate algorithms, evaluating with \operationalmetric{} rather than AUC requires only additional knowledge of the operating capacity distribution and behavioral parameters, quantities typically estimable from operational data.

Several extensions are natural. One is to move from the static formulation to a dynamic setting in which capacity or demand varies over time.
Another is to relax the hard service-capacity constraint and consider how the system manages queues when demand exceeds capacity. In such settings, cannibalization shifts from crowding out high-value requests with low-value ones to manifesting through waiting times and overtime costs.

More broadly, our results suggest a reframing of how AI tools should be evaluated in service-delivery contexts. Predictive accuracy alone is not sufficient. As AI-assisted intervention systems become more deeply embedded in healthcare, education, and public services, the relevant question is shifting from how well we can predict to how well AI models adapt to operational regimes to maximize real-world impact.

%% file: section/appendix_theory.tex
\section{Proofs}\label{app:proofs}

\subsection{Accuracy of the Fluid Approximation}\label{proof:prop:fluid_limit_approximation}

\begin{proof}{Proof.}
The fluid model makes two approximations: (i) it replaces the random request count by $\tilde N(\tau)=\min\{Np_1(\tau),M\}$, and (ii) the empirical quantile $\hat q_A(\tau)$ by the population quantile $q_A(\tau)$. Let $S_N=\sum_i d_i$ and $T_N=\sum_i d_i r_i$, with $p_1=p_1(\tau)=p_0+\Delta_P(1-\tau)>0$, so $\mathbb{E}[S_N]=Np_1$, and let $\rho_N=M/N\to\rho$.
Under random allocation a served requester has expected value $T_N/S_N$ and $\min\{S_N,M\}$ requesters are served. Therefore, the system efficacy can be written as
\begin{equation}\label{eq:obj_as_product}
\obj(\tau)=\mathbb{E}[N_s R_s],\qquad N_s:=\min\{S_N,M\},\quad R_s:=\begin{cases} T_N/S_N & \text{on } \{S_N\ge1\},\\ 0 & \text{on } \{S_N=0\},\end{cases}
\end{equation}
where $N_s$ is the served count and $R_s\in[0,1]$ is the average served value; on $\{S_N=0\}$ no one is served and $N_s R_s=0$, so the convention $R_s=0$ is innocuous. The fluid objective $\objfluid(\tau)=\tilde N(\tau)\tilde R(\tau)$ replaces both by deterministic counterparts.

\textbf{Per-unit convergence.} We first prove the convergence claim \eqref{eq:fluid-limit}.
Adding and subtracting $\tilde R(\tau)\,\mathbb{E}[N_s]$, the approximation error decomposes as
\[
\objfluid(\tau)-\obj(\tau) = \tilde N(\tau)\tilde R(\tau)-\mathbb{E}[N_s R_s] = \tilde R(\tau)\bigl(\tilde N(\tau)-\mathbb{E}[N_s]\bigr) + \mathbb{E}\bigl[N_s\bigl(\tilde R(\tau)-R_s\bigr)\bigr].
\]
Taking absolute values and applying the triangle inequality,
\begin{equation}\label{eq:error_decomposition}
\tfrac1N|\objfluid(\tau)-\obj(\tau)|\le \underbrace{\tilde R(\tau)\,\tfrac1N\bigl|\tilde N(\tau)-\mathbb{E}[N_s]\bigr|}_{\text{demand error}}+\underbrace{\tfrac1N\mathbb{E}\bigl[N_s\bigl|\tilde R(\tau)-R_s\bigr|\bigr]}_{\text{efficacy error}}.
\end{equation}
The first term is the error from replacing the random served count by its fluid counterpart, and the second is the error from replacing the average served value by the fluid per-request efficacy. We show that each term vanishes as $N\to\infty$.
For the demand error, we bound
\begin{align}
\tilde R(\tau)\,\tfrac1N\bigl|\tilde N(\tau)-\mathbb{E}[N_s]\bigr|
&\le \tfrac1N\bigl|\min\{\mathbb{E}[S_N],M\}-\mathbb{E}[\min\{S_N,M\}]\bigr| \notag\\
&\le \tfrac1N\,\mathbb{E}\bigl|\min\{\mathbb{E}[S_N],M\}-\min\{S_N,M\}\bigr| \notag\\
&\le \tfrac1N\,\mathbb{E}\bigl|S_N-\mathbb{E}[S_N]\bigr|
\;\le\; \tfrac1N\sqrt{\mathrm{Var}(S_N)}
\;\le\; \frac{1}{2\sqrt N}
\xrightarrow[N\to\infty]{}0,
\label{eq:demand_error_chain}
\end{align}
uniformly in $\tau$. The first inequality uses $0\le\tilde R(\tau)\le 1$ together with $\tilde N(\tau)=\min\{\mathbb{E}[S_N],M\}$ and $N_s=\min\{S_N,M\}$; the second is Jensen's inequality; the third uses that $x\mapsto\min\{x,M\}$ is $1$-Lipschitz; the fourth is the Cauchy--Schwarz inequality; the last uses $\mathrm{Var}(S_N)=\sum_i p_i(1-p_i)\le N/4$ for a sum of independent Bernoulli variables.
For the efficacy error, the Glivenko--Cantelli theorem gives $\hat q_A(\tau)\to q_A(\tau)$ a.s. Combining the law of large numbers with \Cref{assump:conditional_smoothness}, we obtain $\tfrac1N S_N\to p_1$ and $\tfrac1N T_N\to p_1\tilde R(\tau)$ a.s., and therefore $R_s\to\tilde R(\tau)$ a.s.\ (using $p_1>0$). Since $N_s\le M = \rho_N N$, bounded convergence gives
\begin{equation}\label{eq:efficacy_error}
\tfrac1N\mathbb{E}\bigl[N_s\bigl|\tilde R(\tau)-R_s\bigr|\bigr]\le \rho_N\,\mathbb{E}\bigl|\tilde R(\tau)-R_s\bigr|\xrightarrow[N\to\infty]{}0.
\end{equation}
The bounds \eqref{eq:demand_error_chain} and \eqref{eq:efficacy_error} applied to \eqref{eq:error_decomposition} prove \eqref{eq:fluid-limit}.

\textbf{Asymptotic upper bound.} We now prove \eqref{eq:fluid-upper} by showing $\obj(\tau) \le \objfluid(\tau) + o(N)$.
First note that $\min\{1, M/S_N\} \le 1$ and $r_i \ge 0$, and so we have
\[
  N_s R_s = \min\{S_N, M\}\,\frac{T_N}{S_N} = \min\Bigl\{1, \frac{M}{S_N}\Bigr\}\,T_N \le T_N
\]
on $\{S_N\ge1\}$; on $\{S_N=0\}$ both sides are zero. Then, plugging into the definition of $\obj(\tau)$, we get
$\obj(\tau) = \mathbb{E}[N_s R_s] \leq \mathbb{E}[T_N]$.
Moreover, $T_N/N \in [0,1]$ and $T_N/N \to p_1\tilde R(\tau)$ a.s.\ by the argument preceding \eqref{eq:efficacy_error}, so bounded convergence gives $\mathbb{E}[T_N] = N p_1(\tau)\,\tilde R(\tau) + o(N)$; the $o(N)$ term is the expected error from flagging by the empirical rather than the population quantile.
\[
  \obj(\tau) = \mathbb{E}[N_s R_s] \le N p_1(\tau)\,\tilde R(\tau) + o(N).
\]
We now show that the right-hand side is bounded by $\objfluid(\tau) + o(N)$ both when capacity is slack and when capacity constraints bind. When capacity is slack at $\tau$ (i.e., $N p_1(\tau) \le M$), we have $\tilde N(\tau) = N p_1(\tau)$. The right-hand side then equals $\objfluid(\tau) + o(N)$, which proves the claim in this regime. When capacity binds ($N p_1(\tau) > M$), we have $\tilde N(\tau) = M$. Since $x \mapsto \min\{x, M\}$ is concave, Jensen's inequality gives
\begin{equation}\label{eq:jensen_count}
  \mathbb{E}[N_s] = \mathbb{E}[\min\{S_N, M\}] \le \min\{\mathbb{E}[S_N], M\} = \min\{N p_1(\tau), M\} = \tilde N(\tau).
\end{equation}
Adding and subtracting $\tilde R(\tau)$ inside the expectation, we bound
\begin{align}
  \obj(\tau) = \mathbb{E}[N_s R_s]
  &= \mathbb{E}[N_s]\,\tilde R(\tau) + \mathbb{E}\bigl[N_s\bigl(R_s - \tilde R(\tau)\bigr)\bigr] \notag\\
  &\le \mathbb{E}[N_s]\,\tilde R(\tau) + \mathbb{E}\bigl[N_s\bigl|R_s - \tilde R(\tau)\bigr|\bigr]
  \;\le\; \tilde N(\tau)\,\tilde R(\tau) + o(N) \;=\; \objfluid(\tau) + o(N). \notag
\end{align}
The first inequality bounds the second term by its absolute value; the second uses the Jensen bound \eqref{eq:jensen_count} for the first term and the efficacy-error bound \eqref{eq:efficacy_error} for the second.
Dividing by $N$ and letting $N \to \infty$, we conclude
\[
  \limsup_{N\to\infty} \tfrac{1}{N}\bigl(\obj(\tau) - \objfluid(\tau)\bigr) \le 0.
\]
\end{proof}

\subsection{Optimal Threshold}\label{sec:proofs}
\subsubsection{\texorpdfstring{Proof of \Cref{lem:score_optimal_threshold_uniqueness}}{Proof of Lemma (Uniqueness of Score Optimal)}}\label{proof:lem:score_optimal_threshold_uniqueness}
\begin{lemma}
\label{lem:score_optimal_threshold_uniqueness}
Under \Cref{assump:joint_density,assump:marginal_positive,assump:conditional_smoothness,assump:monotone_calibration} with $\Delta_P > 0$, the score-optimal threshold $\tau^{*A}_{\mathrm{score}} (p_0, \Delta_P)$ is uniquely defined. When $p_0 > 0$, it is the unique interior point in $(0,1)$ satisfying the first-order condition $\frac{d}{d\tau}\tilde{R}(\tau^{*A}_{\mathrm{score}}) = 0$; when $p_0 = 0$, it is the boundary point $\tau^{*A}_{\mathrm{score}} = 1$.
\end{lemma}
\begin{proof}{Proof.}
Write $e(\tau) = \mathbb{E}(r \mid \hat r^A = q_A(\tau))$ for the marginal value at the threshold and $T(\tau) = \mathbb{E}(r \mid \hat r^A > q_A(\tau))$ for the mean value among flagged individuals. By \Cref{assump:joint_density,assump:conditional_smoothness}, both exist and are continuously differentiable in $\tau$. Since the support of $\hat r^A$ is an open interval by \Cref{assump:marginal_positive}, the boundary quantities at $\tau \in \{0,1\}$ are read as the corresponding one-sided limits. The chain rule gives
\begin{equation}\label{eq:r_derivative}
\frac{d}{d\tau} \tilde{R}\left(\tau \mid A, p_0, \Delta_P\right) = \frac{\Delta_P}{\bigl[p_0 + \Delta_P(1-\tau)\bigr]^2} \Bigl[ p_0 \bigl(\mathbb{E}(r) - e(\tau)\bigr) + \Delta_P (1-\tau) \bigl(T(\tau) - e(\tau)\bigr) \Bigr].
\end{equation}
Define
$$H(\tau) = (1-\tau) \bigl(T(\tau) - e(\tau)\bigr) - \frac{p_0}{\Delta_P} \bigl(e(\tau) - \mathbb{E}(r)\bigr),$$
so that $\mathrm{sgn}\bigl(\tfrac{d}{d\tau}\tilde{R}(\tau \mid A, p_0, \Delta_P)\bigr) = \mathrm{sgn}(H(\tau))$.
We first treat $p_0 > 0$. Differentiating $H$ and using the identity $\frac{d}{d\tau}\bigl[(1-\tau)T(\tau)\bigr] = -e(\tau)$, we obtain
$$H'(\tau) = -\Bigl[(1-\tau) + \frac{p_0}{\Delta_P}\Bigr]\, e'(\tau) \le 0.$$
Moreover, $e$ is strictly increasing in $\tau$: higher predicted scores carry higher expected true scores by \Cref{assump:monotone_calibration}, and $q_A(\tau)$ is strictly increasing by \Cref{assump:marginal_positive}. Integrating $H'$, for $\tau_1 < \tau_2$ we obtain $H(\tau_2) - H(\tau_1) \le -\frac{p_0}{\Delta_P}\bigl(e(\tau_2) - e(\tau_1)\bigr) < 0$, so $H$ is strictly decreasing on $[0,1]$. At the left endpoint, $T(0) = \mathbb{E}(r)$, and therefore
$$H(0) = \bigl(\mathbb{E}(r) - e(0)\bigr)\Bigl(1 + \frac{p_0}{\Delta_P}\Bigr) > 0,$$
since $e(0) < \mathbb{E}(r)$ by \Cref{assump:monotone_calibration}. At the right endpoint,
$$H(1) = -\frac{p_0}{\Delta_P}\bigl(e(1) - \mathbb{E}(r)\bigr) < 0,$$
since $e(1) > \mathbb{E}(r)$ and $p_0 > 0$. By \Cref{assump:conditional_smoothness}, $H$ is continuous. Being continuous and strictly decreasing on $[0,1]$ with $H(0) > 0$ and $H(1) < 0$, $H$ has a unique zero $\tau^{*A}_{\mathrm{score}} \in (0,1)$ by the intermediate value theorem, which is the stated first-order condition.
When $p_0 = 0$, $H(\tau) = (1-\tau)\bigl(T(\tau) - e(\tau)\bigr) > 0$ for all $\tau < 1$ by \Cref{assump:monotone_calibration}. Therefore $\tilde{R}$ is strictly increasing on $[0,1)$, and the score-optimal threshold is the boundary point $\tau^{*A}_{\mathrm{score}} = 1$.
\end{proof}
\subsubsection{\texorpdfstring{Proof of \Cref{thm:two_point_optimality}}{Proof of Theorem (Two-Point Optimality)}}\label{proof:thm:two_point_optimality}
\begin{proof}{Proof.}
Assume $\Delta_P > 0$. The proof analyzes three regimes of the capacity ratio $\frac{M}{N}$: in each regime we sign the derivative of the fluid objective and read off the maximizer. For notational convenience, we drop the conditioning on $A$, $p_0$, $\Delta_P$, $M$ and $N$.

\textbf{Case 1: Large capacity ($\frac{M}{N} \geq p_0 + \Delta_P$).}
In this setting, capacity is underutilized even when all individuals are nudged and the capacity matching threshold is $\tau_c = 0$ by \Cref{eq:capacity_matching_threshold}. Therefore, capacity is never binding, and the expected number of served requests is
\[
\nservedfluid(\tau) = N \left( p_0 + \Delta_P (1-\tau)\right),
\]
which is linear in $\tau$ with derivative $-N \Delta_P$.

The derivative of the objective $\objfluid(\tau | A, M, N, p_0, \Delta_P)$ is:
\begin{align*}
\frac{d}{d\tau} \objfluid(\tau)
&= \frac{d}{d\tau}\left[\nservedfluid(\tau) \cdot \tilde{R}(\tau)\right]\\
&= \nservedfluid'(\tau) \tilde{R}(\tau) + \nservedfluid(\tau) \tilde{R}'(\tau).
\end{align*}

Rearranging the derivative \eqref{eq:r_derivative}, we have
\begin{align}\label{eq:R_R_deriv_connection}
    \tilde{R}'(\tau) & = \frac{\Delta_P }{p_0 + \Delta_P(1-\tau)}\left[\frac{ p_0 \mathbb{E}(r) + \Delta_P (1-\tau) \mathbb{E}(r \mid \hat r^A > q_A(\tau))}{p_0 + \Delta_P(1-\tau)}-\mathbb{E}(r \mid \hat r^A = q_A(\tau))\right] \\
    & = \frac{\Delta_P }{p_0 + \Delta_P(1-\tau)}\left[\tilde{R}(\tau)-\mathbb{E}(r \mid \hat r^A = q_A(\tau))\right]. \notag
\end{align}

Substituting $\nservedfluid'(\tau) = -N \Delta_P$ and \eqref{eq:R_R_deriv_connection}, we obtain
\begin{align*}
\frac{d}{d\tau} \objfluid(\tau) =& -N \Delta_P \tilde{R}(\tau) + N\left( p_0 + \Delta_P (1-\tau)\right) \tilde{R}'(\tau) \\
=& - N \Delta_P \mathbb{E}(r \mid \hat r^A = q_A(\tau)) \leq 0.
\end{align*}

Therefore, $\objfluid(\tau)$ is decreasing in $\tau$, implying $\tau^{*A} = 0$ when $\tfrac{M}{N} \geq p_0+\Delta_P$. This corresponds to $\tau^{*A} = \min\{\tau^{*A}_{\mathrm{score}}, \tau_c\}$ since $\tau_c = 0$ in this regime.

\textbf{Case 2: Small capacity ($\frac{M}{N} \leq p_0$).} In this setting, capacity is fully utilized, even when no individual is nudged and $\tau_c = 1$ by \Cref{eq:capacity_matching_threshold}. The expected number of served requests is constant:
\[
\nservedfluid(\tau) = M \quad \forall \tau,
\]
with derivative $\nservedfluid'(\tau) = 0$.

Therefore, the derivative of the objective becomes:
\[
\frac{d}{d\tau} \objfluid(\tau) = \nservedfluid'(\tau) \tilde{R}(\tau) + \nservedfluid(\tau) \tilde{R}'(\tau) = M \tilde{R}'(\tau).
\]

Since the objective is proportional to $\tilde{R}'(\tau)$, the optimal threshold satisfies $\tilde{R}'(\tau^{*A}) = 0$, which by \Cref{lem:score_optimal_threshold_uniqueness} gives $\tau^{*A} = \tau^{*A}_{\mathrm{score}}(p_0, \Delta_P)$. This corresponds to $\tau^{*A} = \min\{\tau^{*A}_{\mathrm{score}}, \tau_c\}$ since $\tau_c = 1 > \tau^{*A}_{\mathrm{score}}$ in this regime when $p_0 > 0$; when $p_0 = 0$ this regime is empty, since $\frac{M}{N} > 0$.

\textbf{Case 3: Medium capacity ($p_0 < \frac{M}{N} < p_0+\Delta_P$).}
In this setting, there is excess capacity if no individual is flagged, but there is insufficient capacity to service all requests if all individuals are flagged. The capacity matching threshold is $0 < \tau_c < 1$. We write the expected number of served requests as a piecewise function:
\[
\nservedfluid(\tau) =
\begin{cases}
M & \text{if } \tau \leq \tau_c,\\
N \left[p_0 + \Delta_P(1-\tau)\right] & \text{if } \tau > \tau_c.
\end{cases}
\]
Note that $\nservedfluid(\tau)$ is continuous at $\tau_c$ by the definition of $\tau_c$ in \Cref{eq:capacity_matching_threshold}, but is not differentiable at $\tau_c$ since the left and right derivatives differ. For $\tau \neq \tau_c$, the derivative is:
\[
\nservedfluid'(\tau) =
\begin{cases}
0 & \text{if } \tau < \tau_c,\\
-N \Delta_P & \text{if } \tau > \tau_c.
\end{cases}
\]

Therefore, for $\tau \neq \tau_c$, the derivative of the objective is:
\[
\objfluid'(\tau) =
\begin{cases}
M \tilde{R}'(\tau) & \text{if } \tau < \tau_c,\\
- N \Delta_P \mathbb{E}(r \mid \hat r^A = q_A(\tau)) & \text{if } \tau > \tau_c.
\end{cases}
\]

By \Cref{lem:score_optimal_threshold_uniqueness}, $\operatorname{sgn}\tilde{R}'(\tau) = \operatorname{sgn}(H(\tau))$ with $H$ strictly decreasing, so $\tilde{R}'(\tau) > 0$ for $\tau < \tau^{*A}_{\mathrm{score}}(p_0, \Delta_P)$, $\tilde{R}'(\tau^{*A}_{\mathrm{score}}) = 0$, and $\tilde{R}'(\tau) < 0$ for $\tau > \tau^{*A}_{\mathrm{score}}$. We consider two subcases:
\begin{enumerate}
  \item If $\tau_c \leq \tau^{*A}_{\mathrm{score}}$, then $\tilde{R}'(\tau) > 0$ for all $\tau < \tau_c$.
  For $\tau < \tau_c$, we have $\objfluid'(\tau) = M \tilde{R}'(\tau) > 0$, so $\objfluid$ is strictly increasing on $[0, \tau_c)$. For $\tau > \tau_c$, we have $\objfluid'(\tau) < 0$, so $\objfluid$ is strictly decreasing on $(\tau_c, 1]$. Therefore, $\objfluid$ attains its maximum at $\tau^{*A} = \tau_c = \min\{\tau^{*A}_{\mathrm{score}}, \tau_c\}$.
  \item If $\tau_c > \tau^{*A}_{\mathrm{score}}$, then $\tilde{R}'(\tau) = 0$ at $\tau = \tau^{*A}_{\mathrm{score}} < \tau_c$.
  For $\tau < \tau^{*A}_{\mathrm{score}}$, we have $\objfluid'(\tau) = M \tilde{R}'(\tau) > 0$, and for $\tau^{*A}_{\mathrm{score}} < \tau < \tau_c$, we have $\objfluid'(\tau) = M \tilde{R}'(\tau) < 0$. For $\tau > \tau_c$, we also have $\objfluid'(\tau) = - N \Delta_P \mathbb{E}(r \mid \hat r^A = q_A(\tau)) <0.$
  Therefore, $\objfluid$ attains its maximum at $\tau^{*A} = \tau^{*A}_{\mathrm{score}} = \min\{\tau^{*A}_{\mathrm{score}}, \tau_c\}$.
\end{enumerate}
Combining all three cases, we conclude that for all capacity ratios $\frac{M}{N} > 0$:
\[
\tau^{*A}\left(\frac{M}{N}\mid  p_0,\, \Delta_P\right) = \min\left\{\,\tau^{*A}_{\mathrm{score}}(p_0,\, \Delta_P),\, \tau_c\left(\frac{M}{N}\mid  p_0,\, \Delta_P\right)\right\},
\]
which completes the proof.
\end{proof}

\subsubsection{\texorpdfstring{Proof of \Cref{cor:finite_two_point_optimality}}{Proof of Corollary (Two-point optimality in finite systems)}}\label{proof:cor:finite_two_point_optimality}
\begin{proof}{Proof.}
For simplicity, we omit the dependence on $(A, M, N, p_0, \Delta_P)$ in the notation; throughout, $M/N$ is held fixed as $N \to \infty$. The proof proceeds in two steps: we first upgrade the per-unit convergence of \Cref{prop:fluid_limit} to convergence uniform over $\tau \in [0,1]$, and then convert uniform convergence into asymptotic optimality of the fluid threshold. For the first step, the demand-error bound \eqref{eq:demand_error_chain} is at most $\tfrac{1}{2\sqrt{N}}$ for every $\tau$. The efficacy error converges uniformly by a uniform law of large numbers: the indicator family $\{\mathbf{1}\{\hat r^A \ge c\}\colon c\in\mathbb{R}\}$ is a VC class, so the empirical request count and value averages converge to their limits uniformly over thresholds. Moreover, $\hat q_A(\tau) \to q_A(\tau)$ uniformly in $\tau$ by the Glivenko--Cantelli theorem (the population quantile is unique by \Cref{assump:marginal_positive}), and the limit functions are continuous on the compact interval $[0,1]$ by \Cref{assump:conditional_smoothness}. Hence
\[
\sup_{\tau \in [0,1]} \tfrac{1}{N}\bigl|\objfluid(\tau) - \obj(\tau)\bigr| \xrightarrow[N\to\infty]{} 0.
\]
Let $\tilde{\tau}^{*A}\left(\tfrac{M}{N}\right) = \arg\max_{\tau \in [0,1]} \obj(\tau)$ be the optimal threshold in the finite system, so that $\max_{\tau} \obj(\tau) = \obj\bigl(\tilde{\tau}^{*A}(\tfrac{M}{N})\bigr)$. Decomposing,
\begin{align*}
  \obj\bigl(\tilde{\tau}^{*A}(\tfrac{M}{N})\bigr) - \obj\bigl(\tau^{*A}(\tfrac{M}{N})\bigr)
  = & \left[\obj\bigl(\tilde{\tau}^{*A}(\tfrac{M}{N})\bigr) - \objfluid\bigl(\tilde{\tau}^{*A}(\tfrac{M}{N})\bigr)\right] \\
  & + \left[\objfluid\bigl(\tilde{\tau}^{*A}(\tfrac{M}{N})\bigr) - \objfluid\bigl(\tau^{*A}(\tfrac{M}{N})\bigr)\right]
    + \left[\objfluid\bigl(\tau^{*A}(\tfrac{M}{N})\bigr) - \obj\bigl(\tau^{*A}(\tfrac{M}{N})\bigr)\right].
\end{align*}
The left-hand side is non-negative because $\tilde{\tau}^{*A}(\tfrac{M}{N})$ maximizes $\obj$, and the second bracket is non-positive because $\tau^{*A}(\tfrac{M}{N})$ maximizes $\objfluid$. Dropping the second bracket and normalizing by $N$,
\[
  0 \le \tfrac{1}{N}\Bigl[\obj\bigl(\tilde{\tau}^{*A}(\tfrac{M}{N})\bigr) - \obj\bigl(\tau^{*A}(\tfrac{M}{N})\bigr)\Bigr]
  \le \tfrac{1}{N}\bigl|\obj(\tilde{\tau}^{*A}) - \objfluid(\tilde{\tau}^{*A})\bigr|
  + \tfrac{1}{N}\bigl|\objfluid(\tau^{*A}) - \obj(\tau^{*A})\bigr|
  \le 2\sup_{\tau \in [0,1]} \tfrac{1}{N}\bigl|\objfluid(\tau) - \obj(\tau)\bigr|.
\]
Because $\tilde{\tau}^{*A}(\tfrac{M}{N})$ is data-dependent, the two threshold-specific errors are controlled by the uniform bound rather than by pointwise convergence at a fixed $\tau$. Letting $N \to \infty$ and applying the uniform convergence above concludes the proof.
\end{proof}
\subsection{Suboptimality of Prediction-Based Thresholds}\label{sec:proof_suboptimality_prediction_based_thresholds}

The next two lemmas characterize how the gap behaves under varying capacity and behavioral parameters. \Cref{thm:prediction_based_threshold_suboptimal} follows directly from the lemmas.

\subsubsection{\texorpdfstring{Proof of \Cref{lem:gap_prediction_based_threshold_capacity}}{Lemma (Gap Prediction-based Threshold Capacity)}}\label{proof:lem:gap_prediction_based_threshold_capacity}

\begin{lemma}[Prediction-based gap behavior under varying capacity]
\label{lem:gap_prediction_based_threshold_capacity}
Given fixed $p_0$ and $\Delta_P$, $\tau^{*A}_{\mathrm{score}}(p_0, \Delta_P)$ and $\tau_{A}$ remain constant. There are two regimes for the gap behavior with respect to the capacity ratio $\frac{M}{N}$:
\begin{enumerate}[(a)]
  \item If $\tau_{A} \geq \tau^{*A}_{\mathrm{score}}(p_0, \Delta_P)$ (as illustrated by policy $\tau_{A1}$ in \Cref{fig:policies_comparison}), the gap increases monotonically with the capacity ratio $\frac{M}{N}$.
  \item If $\tau_{A} < \tau^{*A}_{\mathrm{score}}(p_0,\Delta_P)$, the gap initially increases with $\frac{M}{N}$ (as illustrated by policy $\tau_{A2}$ in \Cref{fig:policies_comparison}), then decreases once
  \[
    \mathbb{E}[r \mid \hat r^A = q_A(\tau_c\!\left(\tfrac{M}{N}\mid  p_0, \Delta_P\right))]
    = \tilde{R}(\tau_{A}|A, p_0, \Delta_P)
    \quad \text{and} \quad
    \tau_c\!\left(\tfrac{M}{N}\mid  p_0, \Delta_P\right) > \tau_{A}.
  \]
  The gap is equal to zero at the value of $M/N$ such that $\tau_A = \tau_c\!\left(\tfrac{M}{N}\mid  p_0, \Delta_P\right)$
  and increases again for larger $\frac{M}{N}$ such that $\tau_c < \tau_{A}$.
\end{enumerate}
\end{lemma}
\begin{proof}{Proof.}
Fix an algorithm $A$ and parameters $(p_0,\Delta_P)$. The proof proceeds in three steps: we first express the gap in closed form under each ordering of $\tau_A$, $\tau^{*A}_{\mathrm{score}}$, and $\tau_c(\rho)$; we then sign the derivative of the gap in $\rho$ within each of four cases; finally, we assemble the cases along the path of $\tau_c(\rho)$ as $\rho$ increases to establish parts (a) and (b).
For notational convenience write $\rho = M/N$ and define the gap as a function of
capacity ratio,
\[
  G_A\!\left(\tau_A, \rho, p_0\right)
  \;:=\;
  \objfluid\!\left(\tau^{*A}(\rho \mid p_0,\Delta_P) \mid A,M,N,p_0,\Delta_P\right)
  -
  \objfluid\!\left(\tau_A \mid A,M,N,p_0,\Delta_P\right).
\]

We first decompose $\objfluid(\tau^{*A}(\rho))$ into three cases based on the ordering of
$\tau_A$, $\tau^{*A}_{\mathrm{score}}$, and $\tau_c(\rho)$:
\begin{enumerate}[(i)]
\item $\tau_c(\rho) \le \tau^{*A}_{\mathrm{score}}$:
\[
\objfluid\!\left(\tau^{*A}(\rho)\right)
=
N\!\left[p_0\mathbb{E}[r]+\Delta_P\!\left(1-\tau_c(\rho)\right)\mathbb{E}\!\big[r\,\big|\,\hat r^A\ge q_A(\tau_c(\rho))\big]\right].
\]
\item $\tau^{*A}_{\mathrm{score}} \le \tau_c(\rho)$:
\[
\objfluid\!\left(\tau^{*A}(\rho)\right)
=
M\,
\frac{
  p_0\mathbb{E}[r]\;+\;\Delta_P\!\left(1-\tau^{*A}_{\mathrm{score}}\right)\mathbb{E}\!\big[r\,\big|\,\hat r^A\ge q_A(\tau^{*A}_{\mathrm{score}})\big]
}{
  p_0+\Delta_P\!\left(1-\tau^{*A}_{\mathrm{score}}\right)
}.
\]
\end{enumerate}

When $\tau_A \le \tau^{*A}_{\mathrm{score}}$, subtracting $\objfluid(\tau_A)$ from each case yields the gap:
\[
\begin{aligned}
& G_A\!\left(\tau_A, \rho, p_0\right) = \objfluid\!\left(\tau^{*A}(\rho)\right) - \objfluid(\tau_A)
\\[0.5em]
&=
\begin{cases}
\displaystyle
N\!\bigl[p_0\mathbb{E}[r] + \Delta_P(1-\tau_c(\rho))\mathbb{E}\!\bigl[r \mid \hat r^A \ge q_A(\tau_c(\rho))\bigr]\bigr]
\\[0.5em]
\quad
-\;
N\!\bigl[p_0\mathbb{E}[r] + \Delta_P(1-\tau_A)\mathbb{E}\!\bigl[r \mid \hat r^A \ge q_A(\tau_A)\bigr]\bigr],
& \text{if } \tau_c(\rho) \le \tau_A,
\\[1.5em]
\displaystyle
N\!\bigl[p_0\mathbb{E}[r] + \Delta_P(1-\tau_c(\rho))\mathbb{E}\!\bigl[r \mid \hat r^A \ge q_A(\tau_c(\rho))\bigr]\bigr]
\\[0.5em]
\quad
-\;
M\,
\frac{p_0\mathbb{E}[r] + \Delta_P(1-\tau_A)\mathbb{E}\!\bigl[r \mid \hat r^A \ge q_A(\tau_A)\bigr]}{p_0+\Delta_P(1-\tau_A)},
& \text{if } \tau_A < \tau_c(\rho)\le \tau^{*A}_{\mathrm{score}},
\\[1.5em]
\displaystyle
M\,\frac{
p_0\mathbb{E}[r]+\Delta_P(1-\tau^{*A}_{\mathrm{score}})\mathbb{E}\!\bigl[r \mid \hat r^A \ge q_A(\tau^{*A}_{\mathrm{score}})\bigr]
}{
p_0+\Delta_P(1-\tau^{*A}_{\mathrm{score}})
}
\\[0.5em]
\quad
-\;
M\,\frac{
p_0\mathbb{E}[r]+\Delta_P(1-\tau_A)\mathbb{E}\!\bigl[r \mid \hat r^A \ge q_A(\tau_A)\bigr]
}{
p_0+\Delta_P(1-\tau_A)
},
& \text{if } \tau^{*A}_{\mathrm{score}} \le \tau_c(\rho).
\end{cases}
\end{aligned}
\]

When $\tau_A > \tau^{*A}_{\mathrm{score}}$, subtracting $\objfluid(\tau_A)$ from each case again yields
\[
\begin{aligned}
& G_A\!\left(\tau_A, \rho, p_0\right) = \objfluid\!\left(\tau^{*A}(\rho)\right) - \objfluid(\tau_A)
\\[0.5em]
&=
\begin{cases}
\displaystyle
N\!\bigl[p_0\mathbb{E}[r] + \Delta_P(1-\tau_c(\rho))\,
\mathbb{E}\!\bigl[r \mid \hat r^A \ge q_A(\tau_c(\rho))\bigr]\bigr]
\\[0.5em]
\quad
-\;
N\!\bigl[p_0\mathbb{E}[r] + \Delta_P(1-\tau_A)\,
\mathbb{E}\!\bigl[r \mid \hat r^A \ge q_A(\tau_A)\bigr]\bigr],
& \text{if } \tau_c(\rho)\le\tau^{*A}_{\mathrm{score}},
\\[1.5em]
\displaystyle
M\,
\frac{
p_0\mathbb{E}[r] + \Delta_P(1-\tau^{*A}_{\mathrm{score}})\,
\mathbb{E}\!\bigl[r \mid \hat r^A \ge q_A(\tau^{*A}_{\mathrm{score}})\bigr]
}{
p_0 + \Delta_P(1-\tau^{*A}_{\mathrm{score}})
}
\\[0.5em]
\quad
-\;
N\!\bigl[p_0\mathbb{E}[r] + \Delta_P(1-\tau_A)\,
\mathbb{E}\!\bigl[r \mid \hat r^A \ge q_A(\tau_A)\bigr]\bigr],
& \text{if } \tau^{*A}_{\mathrm{score}} < \tau_c(\rho)\le\tau_A,
\\[1.5em]
\displaystyle
M\,
\frac{
p_0\mathbb{E}[r] + \Delta_P(1-\tau^{*A}_{\mathrm{score}})\,
\mathbb{E}\!\bigl[r \mid \hat r^A \ge q_A(\tau^{*A}_{\mathrm{score}})\bigr]
}{
p_0 + \Delta_P(1-\tau^{*A}_{\mathrm{score}})
}
\\[0.5em]
\quad
-\;
M\,
\frac{
p_0\mathbb{E}[r] + \Delta_P(1-\tau_A)\,
\mathbb{E}\!\bigl[r \mid \hat r^A \ge q_A(\tau_A)\bigr]
}{
p_0 + \Delta_P(1-\tau_A)
},
& \text{if } \tau_A \le \tau_c(\rho).
\end{cases}
\end{aligned}
\]

The six sub-cases above reduce to four cases based on which threshold(s) saturate capacity:

\textbf{Case 1: Only the optimal threshold saturates capacity, with $\tau^{*A}(\rho) = \tau_c(\rho)$} ($\tau_c(\rho) \le \min\{\tau_A, \tau^{*A}_{\mathrm{score}}\}$).
The gap simplifies to
\[
G_A\!\left(\tau_A, \rho, p_0\right) = N\Delta_P\Bigl[(1-\tau_c(\rho))\mathbb{E}[r \mid \hat r^A \ge q_A(\tau_c(\rho))] - (1-\tau_A)\mathbb{E}[r \mid \hat r^A \ge q_A(\tau_A)]\Bigr].
\]
Since the second term is constant, we compute
\[
\frac{\partial G_A}{\partial \rho} = N\Delta_P \cdot \frac{\partial}{\partial \rho}\Bigl[(1-\tau_c(\rho))\mathbb{E}[r \mid \hat r^A \ge q_A(\tau_c(\rho))]\Bigr].
\]
Using the chain rule with $\frac{\partial \tau_c}{\partial \rho} = -\frac{1}{\Delta_P}$ and
$\frac{d}{d\tau}\bigl[(1-\tau)\mathbb{E}[r \mid \hat r^A \ge q_A(\tau)]\bigr] = -\mathbb{E}[r \mid \hat r^A = q_A(\tau)]$,
we obtain
\[
\frac{\partial G_A}{\partial \rho} = N\Delta_P \cdot \left(-\frac{1}{\Delta_P}\right)\left(-\mathbb{E}[r \mid \hat r^A = q_A(\tau_c(\rho))]\right) = N\,\mathbb{E}[r \mid \hat r^A = q_A(\tau_c(\rho))] > 0.
\]
Thus the gap is strictly increasing in $\rho$ in this regime.

\textbf{Case 2: Both thresholds saturate capacity} ($\tau_A < \tau_c(\rho) \le \tau^{*A}_{\mathrm{score}}$; requires $\tau_A \le \tau^{*A}_{\mathrm{score}}$).
Here $\tau^{*A}(\rho) = \tau_c(\rho)$ and since $\tau_A < \tau_c(\rho)$, both thresholds are aggressive enough to saturate capacity. The gap is
\[
G_A\!\left(\tau_A, \rho, p_0\right) = N\bigl[p_0\mathbb{E}[r] + \Delta_P(1-\tau_c(\rho))\mathbb{E}[r \mid \hat r^A \ge q_A(\tau_c(\rho))]\bigr] - M \cdot \tilde{R}(\tau_A \mid A, p_0, \Delta_P).
\]
Differentiating with $M = \rho N$:
\[
\frac{\partial G_A}{\partial \rho} = N\left[\mathbb{E}[r \mid \hat r^A = q_A(\tau_c(\rho))] -  \tilde{R}(\tau_A \mid A, p_0, \Delta_P)\right].
\]
The sign depends on whether the marginal efficacy at the optimal threshold exceeds the average efficacy at $\tau_A$.
When $\tau_c(\rho) = \tau^{*A}_{\mathrm{score}}$,
\[
\frac{\partial G_A}{\partial \rho} = N\left[\tilde{R}(\tau^{*A}_{\mathrm{score}} \mid A, p_0, \Delta_P) -  \tilde{R}(\tau_A \mid A, p_0, \Delta_P)\right] \geq 0.
\]
When $\tau_c(\rho) \rightarrow \tau_A$,
\[
\frac{\partial G_A}{\partial \rho} \rightarrow N\left[\mathbb{E}[r \mid \hat r^A = q_A(\tau_A)] -  \tilde{R}(\tau_A \mid A, p_0, \Delta_P)\right] < 0.
\]
As $\rho$ increases, $\tau_c(\rho)$ decreases, and hence $\mathbb{E}[r \mid \hat r^A = q_A(\tau_c(\rho))]$ decreases by \Cref{assump:monotone_calibration}; the derivative $\frac{\partial G_A}{\partial \rho}$ is therefore decreasing in $\rho$ and changes sign exactly once. That is, the gap increases with $\rho$ and then decreases. The gap is 0 if $\tau_A = \tau_c(\rho)$.

\textbf{Case 3: Only the optimal threshold saturates capacity, with $\tau^{*A}(\rho) = \tau^{*A}_{\mathrm{score}}$} ($\tau^{*A}_{\mathrm{score}} < \tau_c(\rho) \le \tau_A$; requires $\tau_A > \tau^{*A}_{\mathrm{score}}$).
Here $\tau^{*A}(\rho) = \tau^{*A}_{\mathrm{score}}$ (which is more aggressive than $\tau_c$) saturates capacity, while $\tau_A \ge \tau_c$ does not. The gap is
\[
G_A\!\left(\tau_A, \rho, p_0\right) = M \cdot \tilde{R}(\tau^{*A}_{\mathrm{score}} \mid A, p_0, \Delta_P) - N\bigl[p_0\mathbb{E}[r] + \Delta_P(1-\tau_A)\mathbb{E}[r \mid \hat r^A \ge q_A(\tau_A)]\bigr].
\]
Since the second term is constant and $M = \rho N$:
\[
\frac{\partial G_A}{\partial \rho} = N \cdot \tilde{R}(\tau^{*A}_{\mathrm{score}} \mid A, p_0, \Delta_P) > 0.
\]
Thus the gap is strictly increasing in $\rho$ in this regime.

\textbf{Case 4: Both thresholds saturate capacity} ($\tau_c(\rho) \ge \max\{\tau_A, \tau^{*A}_{\mathrm{score}}\}$).
Since both $\tau_A$ and $\tau^{*A}_{\mathrm{score}}$ are at most $\tau_c$, both are aggressive enough to saturate capacity. The gap simplifies to
\[
G_A\!\left(\tau_A, \rho, p_0\right) = M\Bigl[\tilde{R}(\tau^{*A}_{\mathrm{score}} \mid A, p_0, \Delta_P) - \tilde{R}(\tau_A \mid A, p_0, \Delta_P)\Bigr].
\]
Since $M = \rho N$:
\[
\frac{\partial G_A}{\partial \rho} = N\Bigl[\tilde{R}(\tau^{*A}_{\mathrm{score}} \mid A, p_0, \Delta_P) - \tilde{R}(\tau_A \mid A, p_0, \Delta_P)\Bigr].
\]
This derivative is non-negative and vanishes only when $\tau^{*A}_{\mathrm{score}} = \tau_A$, so the gap is non-decreasing in $\rho$.

As $\rho$ increases from 0, the capacity threshold $\tau_c(\rho)$ decreases from 1 toward 0, traversing different cases depending on the relationship between $\tau_A$ and $\tau^{*A}_{\mathrm{score}}$. Throughout, $G_A$ is continuous in $\rho$, since the case expressions agree at the case boundaries.

When $\tau_A \ge \tau^{*A}_{\mathrm{score}}$, as $\rho$ increases, we move through Cases 4 $\to$ 3 $\to$ 1. In Case~4 ($\tau_c \ge \tau_A$), the gap increases since $\tilde{R}(\tau^{*A}_{\mathrm{score}}) \ge \tilde{R}(\tau_A)$ when $\tau_A \ge \tau^{*A}_{\mathrm{score}}$. In Case~3 ($\tau^{*A}_{\mathrm{score}} < \tau_c \le \tau_A$), the gap increases with $\partial G/\partial\rho = N \cdot \tilde{R}(\tau^{*A}_{\mathrm{score}}) > 0$. In Case~1 ($\tau_c \le \tau_A$), the gap continues to increase. Thus the gap is monotonically increasing in $\rho$, establishing part~(a).

When $\tau_A < \tau^{*A}_{\mathrm{score}}$, as $\rho$ increases, we traverse Cases 4 $\to$ 2 $\to$ 1. In Case~4 ($\tau_c \ge \tau^{*A}_{\mathrm{score}}$), the gap increases since $\tilde{R}(\tau^{*A}_{\mathrm{score}}) > \tilde{R}(\tau_A)$ when $\tau_A < \tau^{*A}_{\mathrm{score}}$. In Case~2 ($\tau_A < \tau_c \le \tau^{*A}_{\mathrm{score}}$), the gap initially increases but then decreases once $\mathbb{E}[r \mid \hat r^A = q_A(\tau_c)] < \tilde{R}(\tau_A)$, reaching zero when $\tau_c = \tau_A$. In Case~1 ($\tau_c \le \tau_A$), the gap increases. This yields the non-monotone pattern in part~(b): the gap first increases, then decreases to zero at $\tau_c = \tau_A$, and then increases again.
\end{proof}

In case (a), $\tau_{A} \geq \tau^{*A}_{\mathrm{score}}$, so the threshold flags too few individuals and most served requests are low-risk independent ones. The wasted capacity grows with $\frac{M}{N}$, so the gap increases monotonically in $\frac{M}{N}$.

In case (b), $\tau_{A} < \tau^{*A}_{\mathrm{score}}$, so $\tau_{A}$ is too aggressive to handle cannibalization, though it may be fine for filling capacity. When $\frac{M}{N}$ is small, the optimal threshold is $\tau^{*A}_{\mathrm{score}}$ and $\tau_{A}$ flags too many low-value individuals; the gap still grows with $\frac{M}{N}$ because capacity is easy to fill and the only cost is the ignored cannibalization. As $\frac{M}{N}$ grows, the optimal threshold drops from $\tau^{*A}_{\mathrm{score}}$ toward $\tau_c$. The gap reaches zero when $\tau_c$ falls to $\tau_{A}$, where $\tau_{A}$ exactly matches capacity. For larger $\frac{M}{N}$, $\tau_{A}$ is too conservative relative to $\tau_c$, and the wasted capacity opens the gap again.

\subsubsection{\texorpdfstring{Proof of \Cref{lem:gap_prediction_based_threshold_behavioral}}{Lemma (Gap Prediction-based Threshold Behavioral)}}\label{proof:lem:gap_prediction_based_threshold_behavioral}

\begin{lemma}[Prediction-based gap behavior under varying baseline demand]
\label{lem:gap_prediction_based_threshold_behavioral}
Given fixed $\frac{M}{N}$ and $\Delta_P$,
both $\tau^{*A}_{\mathrm{score}}(p_0, \Delta_P)$ and
$\tau_c\!\left(\frac{M}{N}\mid  p_0,\, \Delta_P\right)$ vary with $p_0$,
while $\tau_{A}$ remains constant.
Let $\bar{p}_0$ be the value of $p_0$ satisfying
\begin{equation}\label{eq:tilde_p0}
    \tau^{*A}_{\mathrm{score}}(\bar{p}_0, \Delta_P)
  = \tau_c\!\left(\tfrac{M}{N}\mid  \bar{p}_0,\, \Delta_P\right).
\end{equation}
The gap behavior with respect to $p_0$ follows three regimes:
\begin{enumerate}[(a)]
\item If $\tau_{A} \ge \tau^{*A}\!\left(\tfrac{M}{N}\mid  \bar{p}_0,\, \Delta_P\right)$ (as illustrated by policy $\tau_{A1}$ in \Cref{fig:policies_comparison}),
the gap decreases as $p_0$ increases until $p_0 = \tfrac{M}{N} - \Delta_P(1 - \tau_{A})$. For $p_0 > \tfrac{M}{N} - \Delta_P(1 - \tau_{A})$, define
\begin{equation}\label{eq:derivative_R_p0_h}
    h(\tau) = \frac{\Delta_P(1-\tau)\bigl(\mathbb{E}[r] - \mathbb{E}[r \mid \hat r^A \ge q_A(\tau)]\bigr)}{(p_0 + \Delta_P(1-\tau))^2}.
\end{equation}
The gap then increases when $h(\tau^{*A}_{\mathrm{score}}(p_0, \Delta_P)) - h(\tau_A)>0$ and decreases when $h(\tau^{*A}_{\mathrm{score}}(p_0, \Delta_P)) - h(\tau_A)<0$.
\item If $\tau^{*A}\!\left(\tfrac{M}{N}\mid  1-\Delta_P,\, \Delta_P\right) < \tau_{A} < \tau^{*A}\!\left(\tfrac{M}{N}\mid  \bar{p}_0,\, \Delta_P\right)$ (as illustrated by policy $\tau_{A2}$ in \Cref{fig:policies_comparison}),
the gap decreases with $p_0$ and reaches zero when
\[
  p_0 = \tfrac{M}{N} - \Delta_P(1 - \tau_{A}).
\]
It then increases for larger $p_0$ until the following equality holds:
\[
\frac{\tilde{R}(\tau_{A}|A,p_0,\Delta_P) - \mathbb{E}[r]}
{p_0 + \Delta_P(1-\tau_{A})}
=
\frac{\mathbb{E}[r \mid \hat r^A = q_A(\tau_c\!\left(\tfrac{M}{N}\mid p_0,\Delta_P\right))] - \mathbb{E}[r]}
{p_0 + \Delta_P(1-\tau_c\!\left(\tfrac{M}{N}\mid p_0,\Delta_P\right))}.
\]
Beyond this point, the gap decreases again and reaches 0 when
$\tau^{*A}_{\mathrm{score}}(p_0, \Delta_P) = \tau_{A}$. For larger $p_0$, using \Cref{eq:derivative_R_p0_h}, the gap increases when $h(\tau^{*A}_{\mathrm{score}}(p_0, \Delta_P)) - h(\tau_A)>0$ and decreases when $h(\tau^{*A}_{\mathrm{score}}(p_0, \Delta_P)) - h(\tau_A)<0$.

\item If $\tau_{A} \leq \tau^{*A}\!\left(\tfrac{M}{N}\mid  1-\Delta_P,\, \Delta_P\right)$ (as illustrated by policy $\tau_{A3}$ in \Cref{fig:policies_comparison}),
the gap decreases with $p_0$ and reaches zero when
\[
  p_0 = \tfrac{M}{N} - \Delta_P(1 - \tau_{A}).
\]
It then increases for larger $p_0$ until the following equality holds:
\[
\frac{\tilde{R}(\tau_{A}|A,p_0,\Delta_P) - \mathbb{E}[r]}
{p_0 + \Delta_P(1-\tau_{A})}
=
\frac{\mathbb{E}[r \mid \hat r^A = q_A(\tau_c\!\left(\tfrac{M}{N}\mid p_0,\Delta_P\right))] - \mathbb{E}[r]}
{p_0 + \Delta_P(1-\tau_c\!\left(\tfrac{M}{N}\mid p_0,\Delta_P\right))}.
\]
Beyond this point, the gap decreases again.

\end{enumerate}
\end{lemma}

\begin{proof}{Proof.}
Fix $\tfrac{M}{N}$ and $\Delta_P$, and let $p_0$ vary. The proof mirrors that of \Cref{lem:gap_prediction_based_threshold_capacity}: we express the gap in closed form by cases, sign the derivative of the gap in $p_0$ within each case, and assemble the cases along the path of $\tau_c(p_0)$.
The algorithm's threshold $\tau_A$ remains constant, while both
$\tau^{*A}_{\mathrm{score}}(p_0, \Delta_P)$ and $\tau_c\!\left(\tfrac{M}{N}\mid  p_0, \Delta_P\right)$
vary with $p_0$. Define the gap as a function of $p_0$:
\[
G_A\!\left(\tau_A, \tfrac{M}{N}, p_0\right) := \objfluid\!\left(\tau^{*A}\!\left(\tfrac{M}{N}\mid  p_0, \Delta_P\right)\Big| A, M, N, p_0, \Delta_P\right)
- \objfluid\!\left(\tau_A \Big| A, M, N, p_0, \Delta_P\right).
\]

Let $\bar{p}_0$ denote the value where $\tau^{*A}_{\mathrm{score}}(\bar{p}_0, \Delta_P) = \tau_c\!\left(\tfrac{M}{N}\mid  \bar{p}_0, \Delta_P\right)$.
For $p_0 \le \bar{p}_0$, we have $\tau_c \le \tau^{*A}_{\mathrm{score}}$, so $\tau^{*A} = \tau_c$.
For $p_0 > \bar{p}_0$, we have $\tau_c > \tau^{*A}_{\mathrm{score}}$, so $\tau^{*A} = \tau^{*A}_{\mathrm{score}}$.

When $\tau_A \ge \tau^{*A}\!\left(\tfrac{M}{N}\mid  \bar{p}_0, \Delta_P\right)$, the gap admits the following decomposition:
\[
\begin{aligned}
& G_A\!\left(\tau_A, \tfrac{M}{N}, p_0\right) = \objfluid\!\left(\tau^{*A}(p_0)\right) - \objfluid(\tau_A)
\\[0.5em]
&=
\begin{cases}
\displaystyle
N\Delta_P\Bigl[(1-\tau_c)\mathbb{E}[r \mid \hat r^A \ge q_A(\tau_c)] - (1-\tau_A)\mathbb{E}[r \mid \hat r^A \ge q_A(\tau_A)]\Bigr],
\\[0.5em]
\hfill \text{if } p_0 \le \bar{p}_0 \text{ (i.e., } \tau_c \le \tau^{*A}_{\mathrm{score}} \le \tau_A\text{)},
\\[1.5em]
\displaystyle
M \cdot \tilde{R}(\tau^{*A}_{\mathrm{score}}(p_0)) - N\bigl[p_0\mathbb{E}[r] + \Delta_P(1-\tau_A)\mathbb{E}[r \mid \hat r^A \ge q_A(\tau_A)]\bigr],
\\[0.5em]
\hfill \text{if } \bar{p}_0 < p_0 \le \tfrac{M}{N} - \Delta_P(1-\tau_A) \text{ (i.e., } \tau^{*A}_{\mathrm{score}} < \tau_c \le \tau_A\text{)},
\\[1.5em]
\displaystyle
M\Bigl[\tilde{R}(\tau^{*A}_{\mathrm{score}}(p_0)) - \tilde{R}(\tau_A)\Bigr],
\\[0.5em]
\hfill \text{if } p_0 > \tfrac{M}{N} - \Delta_P(1-\tau_A) \text{ (i.e., } \tau_c > \tau_A\text{)}.
\end{cases}
\end{aligned}
\]

When $\tau_A < \tau^{*A}\!\left(\tfrac{M}{N}\mid  \bar{p}_0, \Delta_P\right)$, the gap admits the following decomposition:
\[
\begin{aligned}
& G_A\!\left(\tau_A, \tfrac{M}{N}, p_0\right) = \objfluid\!\left(\tau^{*A}(p_0)\right) - \objfluid(\tau_A)
\\[0.5em]
&=
\begin{cases}
\displaystyle
N\Delta_P\Bigl[(1-\tau_c)\mathbb{E}[r \mid \hat r^A \ge q_A(\tau_c)] - (1-\tau_A)\mathbb{E}[r \mid \hat r^A \ge q_A(\tau_A)]\Bigr],
\\[0.5em]
\hfill \text{if } p_0 \le \tfrac{M}{N} - \Delta_P(1-\tau_A) \text{ (i.e., } \tau_c \le \tau_A\text{)},
\\[1.5em]
\displaystyle
M\Bigl[\tilde{R}(\tau_c(p_0)) - \tilde{R}(\tau_A)\Bigr],
\\[0.5em]
\hfill \text{if } \tfrac{M}{N} - \Delta_P(1-\tau_A) < p_0 \le \bar{p}_0 \text{ (i.e., } \tau_A < \tau_c \le \tau^{*A}_{\mathrm{score}}\text{)},
\\[1.5em]
\displaystyle
M\Bigl[\tilde{R}(\tau^{*A}_{\mathrm{score}}(p_0)) - \tilde{R}(\tau_A)\Bigr],
\\[0.5em]
\hfill \text{if } p_0 > \bar{p}_0 \text{ (i.e., } \tau^{*A}_{\mathrm{score}} < \tau_c\text{)}.
\end{cases}
\end{aligned}
\]

Across both regimes, six sub-cases arise, but they reduce to four cases based on which threshold(s) saturate capacity.

\textbf{Case 1: Only optimal saturates} ($\tau_c \le \min\{\tau_A, \tau^{*A}_{\mathrm{score}}\}$). Here $\tau^{*A} = \tau_c$ (optimal exactly saturates) while $\tau_A \ge \tau_c$ (algorithm does not saturate). The gap is:
\[
G_A\!\left(\tau_A, \tfrac{M}{N}, p_0\right) = N\Delta_P\Bigl[(1-\tau_c)\mathbb{E}[r \mid \hat r^A \ge q_A(\tau_c)] - (1-\tau_A)\mathbb{E}[r \mid \hat r^A \ge q_A(\tau_A)]\Bigr].
\]
Differentiating w.r.t.\ $p_0$ in the interior regime $p_0 < \tfrac{M}{N}$, where $\tau_c \in (0,1)$ and $\tfrac{\partial \tau_c}{\partial p_0} = \tfrac{1}{\Delta_P}$:

\[
\frac{\partial G_A}{\partial p_0} = N\Delta_P \cdot \frac{\partial \tau_c}{\partial p_0} \cdot \frac{d}{d\tau}\Bigl[(1-\tau)\mathbb{E}[r \mid \hat r^A \ge q_A(\tau)]\Bigr]\Big|_{\tau=\tau_c}
= N\Delta_P \cdot \frac{1}{\Delta_P} \cdot \bigl(-\mathbb{E}[r \mid \hat r^A = q_A(\tau_c)]\bigr) < 0.
\]
Thus the gap \emph{decreases} as $p_0$ increases in this region.

\textbf{Case 2: Neither saturates} ($\tau^{*A}_{\mathrm{score}} < \tau_c \le \tau_A$).
Here $\tau^{*A} = \tau^{*A}_{\mathrm{score}}$ (optimal does not saturate since $\tau^{*A}_{\mathrm{score}} < \tau_c$) and $\tau_A \ge \tau_c$ (algorithm also does not saturate). The gap is:
\[
G_A\!\left(\tau_A, \tfrac{M}{N}, p_0\right) = M \cdot \tilde{R}(\tau^{*A}_{\mathrm{score}}(p_0)) - N\bigl[p_0\mathbb{E}[r] + \Delta_P(1-\tau_A)\mathbb{E}[r \mid \hat r^A \ge q_A(\tau_A)]\bigr].
\]
Differentiating w.r.t.\ $p_0$:
\[
\frac{\partial G_A}{\partial p_0} = M \cdot \frac{\partial \tilde{R}(\tau^{*A}_{\mathrm{score}})}{\partial p_0} - N\mathbb{E}[r].
\]
By the envelope theorem, $\frac{\partial \tilde{R}}{\partial \tau^{*A}_{\mathrm{score}}} = 0$ at the optimum, so:
\[
\frac{\partial \tilde{R}(\tau^{*A}_{\mathrm{score}})}{\partial p_0} = \frac{\Delta_P(1-\tau^{*A}_{\mathrm{score}})\bigl(\mathbb{E}[r] - \mathbb{E}[r \mid \hat r^A \ge q_A(\tau^{*A}_{\mathrm{score}})]\bigr)}{(p_0 + \Delta_P(1-\tau^{*A}_{\mathrm{score}}))^2}.
\]
Substituting and using the first-order condition $\tilde{R}(\tau^{*A}_{\mathrm{score}}) = \mathbb{E}[r \mid \hat r^A = q_A(\tau^{*A}_{\mathrm{score}})]$ in the rearranged form $\Delta_P(1-\tau^{*A}_{\mathrm{score}})\bigl(\mathbb{E}[r \mid \hat r^A \ge q_A(\tau^{*A}_{\mathrm{score}})] - \mathbb{E}[r]\bigr) = \bigl(p_0+\Delta_P(1-\tau^{*A}_{\mathrm{score}})\bigr)\bigl(\mathbb{E}[r \mid \hat r^A = q_A(\tau^{*A}_{\mathrm{score}})] - \mathbb{E}[r]\bigr)$, we obtain
\[
\frac{\partial G_A}{\partial p_0} = \frac{N}{p_0 + \Delta_P(1-\tau^{*A}_{\mathrm{score}})} \left[ \tfrac{M}{N}\mathbb{E}[r] - \tfrac{M}{N}\mathbb{E}[r \mid \hat r^A = q_A(\tau^{*A}_{\mathrm{score}})] - \bigl(p_0 + \Delta_P(1-\tau^{*A}_{\mathrm{score}})\bigr)\mathbb{E}[r] \right].
\]
Since $\tau^{*A}_{\mathrm{score}} < \tau_c$ in this region, we have $p_0 + \Delta_P(1-\tau^{*A}_{\mathrm{score}}) > \tfrac{M}{N}$, which implies the bracketed term is negative. Thus $\frac{\partial G_A}{\partial p_0} < 0$, and the gap \emph{decreases} as $p_0$ increases.

\textbf{Case 3: Only algorithm saturates} ($\tau_c > \max\{\tau_A, \tau^{*A}_{\mathrm{score}}\}$).
Here $\tau^{*A} = \tau^{*A}_{\mathrm{score}}$ (optimal does not saturate since $\tau^{*A}_{\mathrm{score}} < \tau_c$) and $\tau_A < \tau_c$ (algorithm saturates). The gap is:
\[
G_A\!\left(\tau_A, \tfrac{M}{N}, p_0\right) = M\Bigl[\tilde{R}(\tau^{*A}_{\mathrm{score}}(p_0)) - \tilde{R}(\tau_A)\Bigr].
\]
Differentiating w.r.t.\ $p_0$:
\[
\frac{\partial G_A}{\partial p_0} = M \left( \frac{\partial \tilde{R}(\tau^{*A}_{\mathrm{score}})}{\partial p_0} - \frac{\partial \tilde{R}(\tau_A)}{\partial p_0} \right).
\]

For $\tilde{R}(\tau^{*A}_{\mathrm{score}})$, using the chain rule and $\tau^{*A}_{\mathrm{score}}$ maximizes $\tilde{R}$:
\[
\frac{\partial \tilde{R}(\tau^{*A}_{\mathrm{score}})}{\partial p_0} = \frac{\partial \tilde{R}}{\partial p_0}\Big|_{\tau^{*A}_{\mathrm{score}}} + \frac{\partial \tau^{*A}_{\mathrm{score}}}{\partial p_0} \cdot \frac{\partial \tilde{R}}{\partial \tau}\Big|_{\tau^{*A}_{\mathrm{score}}}
= \frac{\partial \tilde{R}}{\partial p_0}\Big|_{\tau^{*A}_{\mathrm{score}}}.
\]

Therefore, it suffices to compare the direct effect on $R$:
\[
\frac{\partial G_A}{\partial p_0} = M \left( \frac{\partial \tilde{R}}{\partial p_0}\Big|_{\tau^{*A}_{\mathrm{score}}} - \frac{\partial \tilde{R}}{\partial p_0}\Big|_{\tau_A} \right).
\]

Calculating $\frac{\partial \tilde{R}(\tau)}{\partial p_0}$ for fixed $\tau$:
\[
\frac{\partial \tilde{R}(\tau)}{\partial p_0} = \Delta_P(1-\tau) \frac{\mathbb{E}[r] - \mathbb{E}[r \mid \hat r^A \ge q_A(\tau)]}{(p_0 + \Delta_P(1-\tau))^2}.
\]
Define $h(\tau) := \Delta_P(1-\tau) \frac{\mathbb{E}[r] - \mathbb{E}[r \mid \hat r^A \ge q_A(\tau)]}{(p_0 + \Delta_P(1-\tau))^2}$. To determine the sign, we analyze how $h(\tau)$ varies with $\tau$.
Differentiating $h(\tau)$ with respect to $\tau$:
\[
\frac{\partial h}{\partial \tau}
= \frac{\partial}{\partial \tau}\left[\Delta_P(1-\tau) \frac{\mathbb{E}[r] - \mathbb{E}[r \mid \hat r^A \ge q_A(\tau)]}{(p_0 + \Delta_P(1-\tau))^2}\right].
\]
Using the product and quotient rules:
\begin{align*}
  \begin{aligned}
   \frac{\partial h}{\partial \tau}  & = \frac{\Delta_P}{(p_0 + \Delta_P(1-\tau))^3}\biggl[ p_0 (\mathbb{E}[r \mid \hat r^A = q_A(\tau)] -\mathbb{E}[r] ) \\
  &\quad + \Delta_P(1-\tau)(\mathbb{E}[r]+\mathbb{E}[r \mid \hat r^A = q_A(\tau)] - 2\mathbb{E}[r \mid \hat r^A \ge q_A(\tau)])\biggr]
  \end{aligned}
\end{align*}

\textbf{Sub-case 3a} For $\tau \leq \tau^{*A}_{\mathrm{score}}$, we substitute the derivative
\begin{equation*}
  \begin{aligned}
    \tilde{R}'(\tau) &= \frac{\Delta_P}{(p_0 + \Delta_P(1-\tau))^2} [ p_0 (\mathbb{E}[r] - \mathbb{E}[r \mid \hat r^A = q_A(\tau)]) \\
    &\quad + \Delta_P(1-\tau)(\mathbb{E}[r \mid \hat r^A > q_A(\tau)] - \mathbb{E}[r \mid \hat r^A = q_A(\tau)])],
  \end{aligned}
\end{equation*}
from \eqref{eq:r_derivative} into $\frac{\partial h}{\partial \tau}$ and obtain
\begin{align*}
  \frac{\partial h}{\partial \tau}
  \begin{aligned}
    & = \frac{1}{p_0 + \Delta_P(1-\tau)}\biggl[ -\tilde{R}'(\tau)  + \frac{\Delta_P^2 (1-\tau)}{(p_0 + \Delta_P(1-\tau))^2} (\mathbb{E}[r]-\mathbb{E}[r \mid \hat r^A \ge q_A(\tau)])\biggr]
  \end{aligned}
\end{align*}
For $\tau \leq \tau^{*A}_{\mathrm{score}}$, $\tilde{R}'(\tau) \geq 0$. The second term is non-positive since $\mathbb{E}[r] \leq \mathbb{E}[r \mid \hat r^A \ge q_A(\tau)]$. Therefore, $\frac{\partial h}{\partial \tau} \leq 0$ for $\tau \leq \tau^{*A}_{\mathrm{score}}$.

\textbf{Sub-case 3b} For $\tau > \tau^{*A}_{\mathrm{score}}$:
Evaluating the expression in Sub-case 3a at $\tau^{*A}_{\mathrm{score}}$, where $\tilde{R}' = 0$, gives $\frac{\partial h}{\partial \tau} \big|_{\tau^{*A}_{\mathrm{score}}} <0$ since $\mathbb{E}[r] < \mathbb{E}[r \mid \hat r^A \ge q_A(\tau)]$. By this and the continuity of $\frac{\partial h}{\partial \tau}$, there exists an interval $(\tau^{*A}_{\mathrm{score}}, \tau^{\dagger})$ such that $h(\tau^{*A}_{\mathrm{score}}) > h(\tau)$ for all $\tau \in (\tau^{*A}_{\mathrm{score}}, \tau^{\dagger})$, so that $h(\tau^{*A}_{\mathrm{score}}) - h(\tau_A) > h(\tau) - h(\tau_A)$ and the sign of $\frac{\partial G_A}{\partial p_0} = M\,(h(\tau^{*A}_{\mathrm{score}}) - h(\tau_A))$ is determined accordingly. Beyond $\tau^{\dagger}$, $h$ increases back, so $h(\tau^{*A}_{\mathrm{score}}) \le h(\tau)$ for $\tau \ge \tau^{\dagger}$; the sign of $h(\tau^{*A}_{\mathrm{score}}) - h(\tau_A)$ may then reverse, which is the source of the condition-dependent behavior stated in the lemma.

\textbf{Case 4: Both saturate} ($\tau_A < \tau_c \le \tau^{*A}_{\mathrm{score}}$).
This occurs when $\tfrac{M}{N} - \Delta_P(1-\tau_A) \le p_0 < \tfrac{M}{N}$.
Here $\tau^{*A} = \tau_c$ (optimal exactly saturates) and $\tau_A < \tau_c$ (algorithm also saturates). The gap is:
\[
G_A\!\left(\tau_A, \tfrac{M}{N}, p_0\right) = M\Bigl[\tilde{R}(\tau_c(p_0)) - \tilde{R}(\tau_A)\Bigr].
\]
Differentiating w.r.t.\ $p_0$:
\[
\frac{\partial G_A}{\partial p_0} = M \left( \frac{\partial \tilde{R}(\tau_c(p_0))}{\partial p_0} - \frac{\partial \tilde{R}(\tau_A)}{\partial p_0} \right).
\]
For $\tilde{R}(\tau_A)$, since $\tau_A$ is fixed:
\[
\frac{\partial \tilde{R}(\tau_A)}{\partial p_0} = \Delta_P(1-\tau_A) \frac{\mathbb{E}[r] - \mathbb{E}[r \mid \hat r^A \ge q_A(\tau_A)]}{(p_0 + \Delta_P(1-\tau_A))^2} < 0.
\]
For $\tilde{R}(\tau_c(p_0))$, using the chain rule and $p_0 + \Delta_P(1-\tau_c) = \tfrac{M}{N}$:
\[
\frac{\partial \tilde{R}(\tau_c(p_0))}{\partial p_0} = \frac{\partial \tilde{R}}{\partial p_0}\Big|_{\tau_c} + \frac{\partial \tau_c}{\partial p_0} \cdot \frac{\partial \tilde{R}}{\partial \tau}\Big|_{\tau_c}
= \frac{\mathbb{E}[r] - \mathbb{E}[r \mid \hat r^A = q_A(\tau_c)]}{p_0 + \Delta_P(1-\tau_c)} < 0.
\]
Taking the difference:
\[
\frac{\partial \tilde{R}(\tau_c)}{\partial p_0} - \frac{\partial \tilde{R}(\tau_A)}{\partial p_0} = \frac{\tilde{R}(\tau_A) - \mathbb{E}[r]}{p_0 + \Delta_P(1-\tau_A)} - \frac{\mathbb{E}[r \mid \hat r^A = q_A(\tau_c)] - \mathbb{E}[r]}{p_0 + \Delta_P(1-\tau_c)}.
\]
At $\tau_c = \tau_A$ (i.e., $p_0 = \tfrac{M}{N} - \Delta_P(1-\tau_A)$): $\frac{\partial G_A}{\partial p_0} = M \cdot \frac{\tilde{R}(\tau_A) - \mathbb{E}[r \mid \hat r^A = q_A(\tau_A)]}{p_0 + \Delta_P(1-\tau_A)} > 0$. At $\tau_c = \tau^{*A}_{\mathrm{score}}$ (i.e., $p_0 = \bar{p}_0$): the derivative becomes negative. Thus the gap \emph{opens} as $p_0$ increases from $\tfrac{M}{N} - \Delta_P(1-\tau_A)$ and \emph{closes} as $p_0$ approaches $\bar{p}_0$, attaining at least one interior maximum in between.

As $p_0$ increases from 0, the capacity threshold $\tau_c(p_0)$ increases from $1 - \frac{M}{N\Delta_P}$ toward 1, traversing different cases depending on the relationship between $\tau_A$ and $\tau^{*A}_{\mathrm{score}}$.

When $\tau_A \geq \tau^{*A}(\tfrac{M}{N}\mid  \bar{p}_0, \Delta_P)$, as $p_0$ increases, we move through Cases 1 $\to$ 2 $\to$ 3b. In Case~1 ($\tau_c \le \tau_A < \tau^{*A}_{\mathrm{score}}$), the gap decreases with $\partial G/\partial p_0 < 0$. In Case~2 ($\tau^{*A}_{\mathrm{score}}(p_0) < \tau_c \le \tau_A$), as $\tau^{*A}_{\mathrm{score}}$ drops below $\tau_A$ for larger $p_0$, the gap continues to decrease with $\partial G/\partial p_0 < 0$. In Case~3b ($\tau_c > \tau_A \ge \tau^{*A}_{\mathrm{score}}(p_0)$), the gap behavior depends on $h(\tau^{*A}_{\mathrm{score}})$ versus $h(\tau_A)$.

When $\tau_A < \tau^{*A}(\tfrac{M}{N}\mid  \bar{p}_0, \Delta_P)$, as $p_0$ increases, we traverse Cases 1 $\to$ 4 $\to$ 3a ($\to$ 3b). In Case~1 ($\tau_c \le \tau_A < \tau^{*A}_{\mathrm{score}}$), the gap decreases with $\partial G/\partial p_0 < 0$, reaching zero when $\tau_c = \tau_A$ (i.e., $p_0 = \tfrac{M}{N} - \Delta_P(1-\tau_A)$). In Case~4 ($\tau_A < \tau_c \le \tau^{*A}_{\mathrm{score}}$), the gap opens again (since $\partial G/\partial p_0 > 0$ at $\tau_c = \tau_A$) and closes as $\tau_c \to \tau^{*A}_{\mathrm{score}}$ (since $\partial G/\partial p_0 < 0$ there), attaining at least one interior maximum in between. Past $\bar{p}_0$ the traversal enters Case~3a ($\tau_c > \tau^{*A}_{\mathrm{score}} > \tau_A$), where the gap behavior is governed by the sign of $h(\tau^{*A}_{\mathrm{score}}(p_0,\Delta_P)) - h(\tau_A)$. If $\tau_A < \tau^{*A}_{\mathrm{score}}(1-\Delta_P, \Delta_P)$, the decreasing $\tau^{*A}_{\mathrm{score}}(p_0)$ never falls below $\tau_A$ and the traversal remains in Case~3a. If instead $\tau_A \ge \tau^{*A}_{\mathrm{score}}(1-\Delta_P, \Delta_P)$, the score-optimal threshold eventually crosses $\tau_A$ and the traversal enters Case~3b ($\tau_c > \tau_A \ge \tau^{*A}_{\mathrm{score}}$), where the gap behavior again depends on $h(\tau^{*A}_{\mathrm{score}})$ versus $h(\tau_A)$.

This completes the proof.
\end{proof}

In case (a), $\tau_A \geq \tau^{*A}(\frac{M}{N}\mid \bar{p}_0, \Delta_P)$ and the threshold is too conservative throughout. When $p_0$ is small, the optimal threshold is set by capacity; raising $p_0$ adds baseline demand, so capacity fills more easily, $\tau_c$ rises toward $\tau_A$, and the gap shrinks. Once $p_0$ passes $\bar{p}_0$, the optimal threshold is set by cannibalization: $\tau^{*A}_{\mathrm{score}}$ falls as competition for slots grows, moving away from the fixed $\tau_A$, and the gap grows. The gap is therefore U-shaped in $p_0$.

In case (b), $\tau_A < \tau^{*A}(\frac{M}{N}\mid \bar{p}_0, \Delta_P)$ and $\tau_A$ falls in the middle range, so $\tau_A$ equals the optimal threshold at more than one value of $p_0$ and the gap is non-monotone with several zero-crossings. When $p_0$ is small, the optimal threshold is set by capacity and $\tau_A$ may exceed $\tau_c$, wasting capacity; raising $p_0$ lowers $\tau_c$ toward $\tau_A$, and the gap falls to zero at $\tau_c = \tau_A$. As $p_0$ rises further, $\tau_c$ moves above $\tau_A$ and cannibalization takes over. As in case (a), the gap first opens again because $\tau_A$ ignores the growing cannibalization, then may narrow as $\tau^{*A}_{\mathrm{score}}$ falls and its average efficacy comes close to the average efficacy at $\tau_A$.
Once $p_0$ is large enough that $\tau^{*A}_{\mathrm{score}}(1-\Delta_P, \Delta_P) \leq \tau_A$, the optimal threshold drops below $\tau_A$, so $\tau_A$ is now too conservative for cannibalization and the gap opens again.

\subsubsection{\texorpdfstring{Proof of \Cref{thm:prediction_based_threshold_suboptimal}}{Proof of Suboptimality of prediction-based thresholds}}\label{proof:thm:prediction_based_threshold_suboptimal}
\begin{proof}{Proof.}
The gap $G_A(\tau_A, \cdot) \ge 0$ by optimality of $\tau^{*A}$, and $G_A(\tau_A, \cdot) = 0$ if and only if $\tau_A = \tau^{*A}$ at that operating point. The theorem is therefore equivalent to counting the number of operating points at which the fixed $\tau_A$ coincides with the optimal threshold $\tau^{*A}$. Throughout we consider $\tau_A \in (0,1)$; the boundary value $\tau_A = 0$ is excluded, since there $\tau_A = \tau_c = \tau^{*A}$ for all $M/N \ge p_0 + \Delta_P$ and the gap is zero on an interval rather than at isolated points.

First, we analyze how many values of $M/N$ yield zero gap for a fixed $p_0$. By \Cref{lem:gap_prediction_based_threshold_capacity}, we have the following cases. If $\tau_A \ge \tau^{*A}_{\mathrm{score}}$ (case a): the gap is strictly increasing in $M/N$ and is strictly positive for all $M/N > 0$ whenever $\tau_A > \tau^{*A}_{\mathrm{score}}$, so $G_A = 0$ for zero values. If $\tau_A = \tau^{*A}_{\mathrm{score}}$, then $\tau_A = \tau^{*A}$ throughout the flat region where $\tau^{*A} = \tau^{*A}_{\mathrm{score}}$, so the gap equals zero if and only if
\[
  \frac{M}{N} \le p_0 + \Delta_P(1-\tau^{*A}_{\mathrm{score}}).
\]

If $\tau_A < \tau^{*A}_{\mathrm{score}}$ (case b): the gap equals zero at the unique value of $M/N$ where $\tau_c(M/N \mid p_0, \Delta_P) = \tau_A$, i.e.,
\[
  \frac{M}{N} = p_0 + \Delta_P(1-\tau_A).
\]
For $M/N > p_0 + \Delta_P$ we have $\tau_c = 0 < \tau_A$, so $\tau^{*A} = 0 \ne \tau_A$ and the gap is strictly positive. This is a single isolated point, so $G_A = 0$ for exactly one value of $M/N$. In both cases, the gap is zero for at most one value when $\tau_A \ne \tau^{*A}_{\mathrm{score}}$.

Next, we analyze how many values of $p_0$ yield zero gap for a fixed $M/N$. By \Cref{lem:gap_prediction_based_threshold_behavioral}, the gap equals zero precisely when $\tau_A = \tau^{*A}(M/N \mid p_0, \Delta_P)$, which occurs either when $\tau_A = \tau_c(p_0)$ or when $\tau_A = \tau^{*A}_{\mathrm{score}}(p_0)$. Each of these equations has at most one solution in $p_0$: $\tau_c(p_0) = \tau_A$ gives $p_0 = M/N - \Delta_P(1-\tau_A)$, which is unique since $\tau_c$ is strictly increasing in $p_0$, and $\tau^{*A}_{\mathrm{score}}(p_0) = \tau_A$ gives at most one solution since $\tau^{*A}_{\mathrm{score}}$ is strictly decreasing in $p_0$. To show the latter, the first-order condition $H(\tau^{*A}_{\mathrm{score}}) = 0$ from \Cref{lem:score_optimal_threshold_uniqueness} and the implicit function theorem give
\[
  \frac{\partial \tau^{*A}_{\mathrm{score}}}{\partial p_0}
  = -\,\frac{\partial H/\partial p_0}{\partial H/\partial \tau},
\]
where $\partial H/\partial \tau < 0$ and $\partial H/\partial p_0 = -\frac{1}{\Delta_P}\bigl(\mathbb{E}[r \mid \hat r^A = q_A(\tau^{*A}_{\mathrm{score}})] - \mathbb{E}[r]\bigr) < 0$ since $\mathbb{E}[r \mid \hat r^A = q_A(\tau^{*A}_{\mathrm{score}})] = \tilde{R}(\tau^{*A}_{\mathrm{score}}) > \mathbb{E}[r]$; hence $\frac{\partial \tau^{*A}_{\mathrm{score}}}{\partial p_0} < 0$. Examining the three cases in \Cref{lem:gap_prediction_based_threshold_behavioral}: case (a) yields at most one zero, case (b) yields exactly two zeros (one from each equation), and case (c) yields exactly one zero. Hence the gap equals zero for at most two isolated values of $p_0$.
\end{proof}

\subsection{Suboptimality of Capacity-Matching Thresholds}\label{sec:proof_suboptimality_capacity_matching}

The next two lemmas characterize how the gap defined for the capacity-matching threshold behaves under varying capacity and behavioral parameters.

\subsubsection{\texorpdfstring{Proof of \Cref{lem:gap_capacity_matching_MN}}{Lemma (Gap Capacity-matching Threshold under varying M/N)}}\label{proof:lem:gap_capacity_matching_MN}

\begin{lemma}[Capacity-matching gap behavior under varying capacity]
\label{lem:gap_capacity_matching_MN}
Fix $(p_0,\Delta_P)$.
Then $\tau^{*A}_{\mathrm{score}}(p_0,\Delta_P)$ is constant in $\tfrac{M}{N}$, while
$\tau_c(\tfrac{M}{N}\mid p_0,\Delta_P)$ varies monotonically with $\tfrac{M}{N}$.
The suboptimality gap $G_A\!\left(\tau_c, \tfrac{M}{N}, p_0\right)$ increases with $\tfrac{M}{N}$ for $\tfrac{M}{N} \le p_0$, then decreases for $p_0 < \tfrac{M}{N} < p_0 + \Delta_P(1-\tau^{*A}_{\mathrm{score}})$.
Beyond this point, $G_A\!\left(\tau_c, \tfrac{M}{N}, p_0\right)=0$ for all
  \(
  \tfrac{M}{N} \ge p_0 + \Delta_P(1-\tau^{*A}_{\mathrm{score}}).
  \)
\end{lemma}

\begin{proof}{Proof.}
Fix $p_0$ and $\Delta_P$. The proof expresses the gap in closed form in the three capacity regimes and signs its derivative in each. The gap is
\[
G_A\!\left(\tau_c, \tfrac{M}{N}, p_0\right)
= \begin{cases}
M\left(\tilde{R}(\tau^{*A}_{\mathrm{score}}) - \tilde{R}(1)\right) & \text{if } \tfrac{M}{N} \le p_0 \\
M\left(\tilde{R}(\tau^{*A}_{\mathrm{score}}) - \tilde{R}(\tau_c)\right) & \text{if } p_0 < \tfrac{M}{N} < p_0 + \Delta_P(1 - \tau^{*A}_{\mathrm{score}}) \\
0 & \text{if } \tfrac{M}{N} \ge p_0 + \Delta_P(1 - \tau^{*A}_{\mathrm{score}})
\end{cases}
\]
where $\tau_c = 1$ when $\tfrac{M}{N} \le p_0$.

\textbf{Case 1:} When $\tfrac{M}{N} \le p_0$, $G_A\!\left(\tau_c, \tfrac{M}{N}, p_0\right)$ increases with $\tfrac{M}{N}$ since $\tau_c=1$ and $\tau^{*A}_{\mathrm{score}}$ both remain constant, while $M$ grows.

\textbf{Case 2:} When $p_0 < \tfrac{M}{N} < p_0 + \Delta_P(1 - \tau^{*A}_{\mathrm{score}})$, we have
\begin{align*}
\frac{\partial}{\partial \tfrac{M}{N}}G_A\!\left(\tau_c, \tfrac{M}{N}, p_0\right)
&= N\left(\tilde{R}(\tau^{*A}_{\mathrm{score}}) - \tilde{R}(\tau_c)\right) - M \frac{\partial}{\partial \tfrac{M}{N}} \tilde{R}(\tau_c) \\
&= N\left(\tilde{R}(\tau^{*A}_{\mathrm{score}}) - \tilde{R}(\tau_c)\right) - M \frac{\partial \tau_c}{\partial \tfrac{M}{N}} \bigg|_{\tau_c} \frac{\partial}{\partial \tau} \tilde{R}(\tau) \bigg|_{\tau_c}
\end{align*}
Since $\tau_c = 1 - \frac{1}{\Delta_P}\left(\tfrac{M}{N} - p_0\right)$, we have $\frac{\partial \tau_c}{\partial \tfrac{M}{N}} = -\frac{1}{\Delta_P}$. Thus,
\begin{align*}
\frac{\partial}{\partial \tfrac{M}{N}}G_A\!\left(\tau_c, \tfrac{M}{N}, p_0\right)
&= N\left[\tilde{R}(\tau^{*A}_{\mathrm{score}}) - \tilde{R}(\tau_c) + \frac{M}{N} \frac{1}{\Delta_P} \tilde{R}'(\tau_c)\right]
\end{align*}
Plugging in \eqref{eq:R_R_deriv_connection}, we arrive at
\[
\frac{\partial}{\partial \tfrac{M}{N}}G_A\!\left(\tau_c, \tfrac{M}{N}, p_0\right) = N\left[\tilde{R}(\tau^{*A}_{\mathrm{score}}) - \mathbb{E}[r \mid \hat r^A = q_A(\tau_c)]\right].
\]
In this case, since $\tau_c > \tau^{*A}_{\mathrm{score}}$ and $\mathbb{E}[r \mid \hat r^A = q_A(\cdot)]$ is increasing, $\mathbb{E}[r \mid \hat r^A = q_A(\tau_c)] > \mathbb{E}[r \mid \hat r^A = q_A(\tau^{*A}_{\mathrm{score}})] = \tilde{R}(\tau^{*A}_{\mathrm{score}})$ by the first-order condition, so $\tilde{R}(\tau^{*A}_{\mathrm{score}}) - \mathbb{E}[r \mid \hat r^A = q_A(\tau_c)] < 0$.
Therefore, the gap decreases when $p_0 < \tfrac{M}{N} < p_0 + \Delta_P(1 - \tau^{*A}_{\mathrm{score}})$.

\textbf{Case 3:} When $\tfrac{M}{N} \ge p_0 + \Delta_P(1 - \tau^{*A}_{\mathrm{score}})$, we have $\tau_c \le \tau^{*A}_{\mathrm{score}}$, so the capacity-matching threshold is itself optimal, $\tau^{*A} = \min\{\tau^{*A}_{\mathrm{score}}, \tau_c\} = \tau_c$, and the gap is zero.

\end{proof}

\subsubsection{\texorpdfstring{Proof of \Cref{lem:gap_capacity_matching_p0}}{Lemma (Gap Capacity-matching Threshold under varying p0)}}\label{proof:lem:gap_capacity_matching_p0}

\begin{lemma}[Capacity-matching gap behavior under varying baseline demand]
\label{lem:gap_capacity_matching_p0}
Fix $\tfrac{M}{N}$ and $\Delta_P$.
Let $\bar p_0$ be defined as in \Cref{def:bar_p0}, i.e.,
\[
\tau^{*A}_{\mathrm{score}}(\bar p_0,\Delta_P)
=
\tau_c\!\left(\tfrac{M}{N}\mid \bar p_0,\Delta_P\right).
\]
The suboptimality gap $G_A\!\left(\tau_c, \tfrac{M}{N}, p_0\right)=0$ for all $p_0 \le \bar p_0$. It is then strictly increasing for $p_0 \in (\bar p_0, \tfrac{M}{N})$. Once $\tau_c=1$ (i.e., $p_0 \ge \tfrac{M}{N}$), the gap is decreasing in $p_0$.
\end{lemma}
\begin{proof}{Proof.}
The proof treats three regimes of $p_0$: below $\bar p_0$ the gap vanishes, on $(\bar p_0, \tfrac{M}{N})$ it is strictly increasing, and beyond $\tfrac{M}{N}$ it is strictly decreasing.
If $\tau_c\!\left(\tfrac{M}{N}\mid p_0,\Delta_P\right)\le \tau^{*A}_{\mathrm{score}}(p_0,\Delta_P)$, then
$\tau^{*A}=\tau_c$ and the gap is zero.
Thus we focus on the region where
\[
\tau_c\!\left(\tfrac{M}{N}\mid p_0,\Delta_P\right)>\tau^{*A}_{\mathrm{score}}(p_0,\Delta_P),
\]
in which $\tau^{*A}=\tau^{*A}_{\mathrm{score}}$ and the gap is $G_A = M\bigl[\tilde{R}(\tau^{*A}_{\mathrm{score}}(p_0)) - \tilde{R}(\tau_c(p_0))\bigr]$. The crossing $\bar p_0$ exists and is unique, since $\tau_c$ is increasing and $\tau^{*A}_{\mathrm{score}}$ decreasing in $p_0$.

Differentiate w.r.t.\ $p_0$. Since $\tilde{R}(\tau\mid p_0)$ depends on $p_0$ both directly and through the threshold, the total derivative of each term is
\[
\frac{d}{dp_0}\tilde{R}(\tau(p_0), p_0)
= \frac{\partial \tilde{R}}{\partial p_0}\Big|_{\tau} + \frac{d\tau}{dp_0}\,\tilde{R}'(\tau),
\qquad
\frac{\partial \tilde{R}}{\partial p_0}\Big|_{\tau} = \frac{\mathbb{E}[r] - \tilde{R}(\tau)}{p_0 + \Delta_P(1-\tau)}.
\]
Since $\tilde{R}'(\tau^{*A}_{\mathrm{score}})=0$ by optimality of $\tau^{*A}_{\mathrm{score}}$, we have
\[
\frac{d}{dp_0}\tilde{R}(\tau^{*A}_{\mathrm{score}}(p_0))
= \frac{\mathbb{E}[r] - \tilde{R}(\tau^{*A}_{\mathrm{score}})}{p_0 + \Delta_P(1-\tau^{*A}_{\mathrm{score}})}.
\]

For $\tau_c$, using $\tfrac{d\tau_c}{dp_0}=\tfrac{1}{\Delta_P}$, the identity \eqref{eq:R_R_deriv_connection} for $\tilde{R}'(\tau_c)$, and $p_0+\Delta_P(1-\tau_c)=\tfrac{M}{N}$,
\[
\frac{d}{dp_0}\tilde{R}(\tau_c(p_0))
= \frac{\mathbb{E}[r] - \tilde{R}(\tau_c)}{M/N} + \frac{1}{\Delta_P}\tilde{R}'(\tau_c)
= \frac{\mathbb{E}[r] - \mathbb{E}[r \mid \hat r^A = q_A(\tau_c)]}{M/N}.
\]
Therefore, we arrive at
\[
\frac{dG_A}{dp_0}
= M\!\left[
\frac{\mathbb{E}[r] - \tilde{R}(\tau^{*A}_{\mathrm{score}})}{p_0 + \Delta_P(1-\tau^{*A}_{\mathrm{score}})}
+ \frac{\mathbb{E}[r \mid \hat r^A = q_A(\tau_c)] - \mathbb{E}[r]}{M/N}
\right].
\]
Note that $p_0+\Delta_P(1-\tau^{*A}_{\mathrm{score}}) > \tfrac{M}{N}$ (since $\tau^{*A}_{\mathrm{score}} < \tau_c$) and $\tilde{R}(\tau^{*A}_{\mathrm{score}}) > \mathbb{E}[r]$, so the first term exceeds $\tfrac{\mathbb{E}[r] - \tilde{R}(\tau^{*A}_{\mathrm{score}})}{M/N}$; hence
\[
\frac{dG_A}{dp_0}
> M\,\frac{\mathbb{E}[r \mid \hat r^A = q_A(\tau_c)] - \tilde{R}(\tau^{*A}_{\mathrm{score}})}{M/N}
= M\,\frac{\mathbb{E}[r \mid \hat r^A = q_A(\tau_c)] - \mathbb{E}[r \mid \hat r^A = q_A(\tau^{*A}_{\mathrm{score}})]}{M/N}
> 0,
\]
The equality uses the first-order condition $\tilde{R}(\tau^{*A}_{\mathrm{score}}) = \mathbb{E}[r \mid \hat r^A = q_A(\tau^{*A}_{\mathrm{score}})]$; the final inequality holds because $\mathbb{E}[r \mid \hat r^A = q_A(\cdot)]$ is increasing (\Cref{assump:monotone_calibration}) and $\tau_c > \tau^{*A}_{\mathrm{score}}$.
This concludes the gap is strictly increasing for $p_0 \in (\bar p_0, \tfrac{M}{N})$.

Finally, when $p_0 \ge \tfrac{M}{N}$ we have $\tau_c = 1$, so $\tilde{R}(\tau_c) = \tilde{R}(1) = \mathbb{E}[r]$ is constant in $p_0$ and the gap is $G_A = M\bigl[\tilde{R}(\tau^{*A}_{\mathrm{score}}) - \mathbb{E}[r]\bigr]$. Differentiating over $p_0$, we have
\[
\frac{dG_A}{dp_0}
= M\,\frac{\mathbb{E}[r] - \tilde{R}(\tau^{*A}_{\mathrm{score}})}{p_0 + \Delta_P(1-\tau^{*A}_{\mathrm{score}})}
< 0,
\]
since $\tilde{R}(\tau^{*A}_{\mathrm{score}}) > \mathbb{E}[r]$, so the gap is decreasing in $p_0$.
\end{proof}

\subsubsection{\texorpdfstring{Proof of \Cref{thm:capacity_threshold_suboptimal}}{Proof of Suboptimality of capacity-matching threshold}}\label{proof:thm:capacity_threshold_suboptimal}

\begin{proof}{Proof.}
The theorem follows directly from \Cref{lem:gap_capacity_matching_p0,lem:gap_capacity_matching_MN}, which characterize the gap behavior under varying $p_0$ and $\tfrac{M}{N}$ respectively. Specifically, for fixed $p_0$, \Cref{lem:gap_capacity_matching_MN} shows the gap is positive on $\tfrac{M}{N} < p_0 + \Delta_P(1-\tau^{*A}_{\mathrm{score}})$ (increasing for $\tfrac{M}{N} \le p_0$, decreasing thereafter) and vanishes once $\tfrac{M}{N} \ge p_0 + \Delta_P(1-\tau^{*A}_{\mathrm{score}})$; for fixed $\tfrac{M}{N}$, \Cref{lem:gap_capacity_matching_p0} shows the gap vanishes for $p_0 \le \bar p_0$ and is positive for $p_0 > \bar p_0$ (increasing on $(\bar p_0, \tfrac{M}{N})$, then decreasing). In either direction the capacity-matching threshold is optimal only on these boundary regions and strictly suboptimal elsewhere.

It remains to establish the maximum-relative-gap claim. Write $e(\tau) = \mathbb{E}[r \mid \hat r^A = q_A(\tau)]$ for the marginal value at the threshold and $D(\tau) = p_0 + \Delta_P(1-\tau)$ for the request rate. Wherever the gap is positive we have $\tau^{*A}_{\mathrm{score}} < \tau_c$, so demand weakly exceeds capacity at both thresholds and both policies serve exactly $M$ requests; the relative gap therefore reduces to a ratio of per-request efficacies,
\[
G^{\mathrm{rel}}_A = 1 - \frac{\tilde{R}(\tau_c)}{\tilde{R}(\tau^{*A}_{\mathrm{score}})},
\qquad \tilde{R}(\tau^{*A}_{\mathrm{score}}) = e(\tau^{*A}_{\mathrm{score}})
\]
by the first-order condition of \Cref{lem:score_optimal_threshold_uniqueness}.

\emph{Direction (i): $\tfrac{M}{N}$ varies, $p_0$ fixed.} Here $\tau^{*A}_{\mathrm{score}}$ does not depend on $\tfrac{M}{N}$, so the denominator is constant. As $\tfrac{M}{N}$ decreases on $(p_0,\, p_0 + \Delta_P(1-\tau^{*A}_{\mathrm{score}}))$, the capacity-matching threshold $\tau_c = 1 - \tfrac{1}{\Delta_P}(\tfrac{M}{N}-p_0)$ increases toward $1$; since $\tilde{R}'(\tau) = \tfrac{\Delta_P}{D(\tau)}[\tilde{R}(\tau) - e(\tau)] < 0$ for $\tau > \tau^{*A}_{\mathrm{score}}$ by \Cref{lem:score_optimal_threshold_uniqueness}, $\tilde{R}(\tau_c)$ decreases and $G^{\mathrm{rel}}_A$ increases as capacity shrinks. Once $\tfrac{M}{N} \le p_0$ we have $\tau_c = 1$ and $\tilde{R}(\tau_c) = \mathbb{E}[r]$, so the relative gap is constant at its maximum $\bigl(\tilde{R}(\tau^{*A}_{\mathrm{score}}) - \mathbb{E}[r]\bigr)/\tilde{R}(\tau^{*A}_{\mathrm{score}})$, attained exactly on $\tfrac{M}{N} \le p_0$.

\emph{Direction (ii): $p_0$ varies, $\tfrac{M}{N}$ fixed.} Write $e^*(p_0) = e(\tau^{*A}_{\mathrm{score}}(p_0,\Delta_P))$, $\tilde{R}_c(p_0) = \tilde{R}(\tau_c(p_0) \mid p_0)$, and $D^* = p_0 + \Delta_P(1-\tau^{*A}_{\mathrm{score}})$. Two derivative identities drive the argument. First, along the capacity-matching path the request rate is constant, $p_0 + \Delta_P(1-\tau_c(p_0)) = \tfrac{M}{N}$; combining this with $\tfrac{d}{d\tau}\bigl[(1-\tau)\,\mathbb{E}[r \mid \hat r^A \ge q_A(\tau)]\bigr] = -e(\tau)$ and $\tfrac{d\tau_c}{dp_0} = \tfrac{1}{\Delta_P}$ gives
\[
\frac{d\tilde{R}_c}{dp_0} = \frac{\mathbb{E}[r] - e(\tau_c(p_0))}{M/N}.
\]
Second, by the envelope theorem (the first-order condition kills the $\tfrac{\partial \tilde R}{\partial \tau}$ term),
\[
\frac{de^*}{dp_0} = \frac{\partial \tilde{R}}{\partial p_0}\Big|_{\tau^{*A}_{\mathrm{score}}} = \frac{\Delta_P(1-\tau^{*A}_{\mathrm{score}})\bigl(\mathbb{E}[r] - \mathbb{E}[r \mid \hat r^A \ge q_A(\tau^{*A}_{\mathrm{score}})]\bigr)}{(D^*)^2} = \frac{\mathbb{E}[r] - e^*}{D^*},
\]
where the last equality uses the first-order condition in the rearranged form $\Delta_P(1-\tau^{*A}_{\mathrm{score}})\bigl(\mathbb{E}[r \mid \hat r^A \ge q_A(\tau^{*A}_{\mathrm{score}})] - \mathbb{E}[r]\bigr) = D^*\,(e^* - \mathbb{E}[r])$. On $(\bar p_0, \tfrac{M}{N})$ we then have
\[
\frac{d}{dp_0}\,\frac{\tilde{R}_c}{e^*}
= \frac{1}{(e^*)^2}\left[\frac{\bigl(\mathbb{E}[r] - e(\tau_c)\bigr)e^*}{M/N} - \frac{\tilde{R}_c\bigl(\mathbb{E}[r] - e^*\bigr)}{D^*}\right] \le 0,
\]
because $e(\tau_c) - \mathbb{E}[r] \ge e^* - \mathbb{E}[r] \ge 0$ (as $\tau_c > \tau^{*A}_{\mathrm{score}}$ and $e$ is increasing), $e^* \ge \tilde{R}_c$ (as $\tau^{*A}_{\mathrm{score}}$ maximizes $\tilde{R}$), and $\tfrac{M}{N} \le D^*$ (demand exceeds capacity at $\tau^{*A}_{\mathrm{score}}$), so the first term dominates the second term by term. Hence $G^{\mathrm{rel}}_A$ is nondecreasing in $p_0$ on $(\bar p_0, \tfrac{M}{N})$. For $p_0 \ge \tfrac{M}{N}$, $\tau_c = 1$ and $\tilde{R}_c = \mathbb{E}[r]$ is constant, while $\tfrac{de^*}{dp_0} = \tfrac{\mathbb{E}[r] - e^*}{D^*} < 0$, so $G^{\mathrm{rel}}_A = 1 - \mathbb{E}[r]/e^*$ is decreasing. The maximum over $p_0$ is therefore attained at $p_0 = \tfrac{M}{N}$, where $\tau_c = 1$ and $G^{\mathrm{rel}}_A = \bigl(\tilde{R}(\tau^{*A}_{\mathrm{score}}) - \mathbb{E}[r]\bigr)/\tilde{R}(\tau^{*A}_{\mathrm{score}})$, as claimed; this point is feasible under the model constraint $p_0 + \Delta_P \le 1$ whenever $\tfrac{M}{N} \le 1 - \Delta_P$, and otherwise the supremum over the feasible range is attained at the endpoint $p_0 = 1 - \Delta_P$.
\end{proof}

\subsection{Extensions}\label{app:proof_extensions}

\subsubsection{\texorpdfstring{Proof of \Cref{thm:hte_optimality}}{Proof of Theorem (Two-Point Optimality under Heterogeneous Treatment Effects)}}\label{app:hte}
For a flagging threshold $\tau$, define the fraction of the population flagged from region $k$ as $$\gamma_k(\tau) = \max\bigl(\lambda_{k+1} - \max(\tau, \lambda_k),\, 0\bigr),$$
which satisfies $\sum_{k=0}^{K} \gamma_k(\tau) = 1 - \tau$, consistent with the baseline model, and the mean true score among flagged requesters in region $k$ as
$$\bar r_k(\tau) = \mathbb{E}[r \mid \hat r^A \in [q_A(\max(\tau,\lambda_k)), q_A(\lambda_{k+1}))].$$
The number of served requests becomes
\begin{equation}
  \begin{aligned}
    & \tilde{N}\!\left(\tau \;\middle|\; M, N, p_0, \{\lambda_k\}, \{\Delta_P^{(k)}\}\right)  = \min\!\Bigl\{ N\Bigl[p_0 + \sum_{k=0}^{K} \Delta_P^{(k)}\,\gamma_k(\tau)\Bigr],
      \;M \Bigr\},
  \end{aligned}
  \label{eq:fluid_served_requests_het_te}
\end{equation}
and the per-request efficacy becomes
\begin{equation}
  \begin{aligned}
    & \tilde{R}\!\left(\tau \;\middle|\; A, p_0, \{\Delta_P^{(k)}\}\right)  = \frac{p_0\,\mathbb{E}[r] + \sum_{k=0}^{K} \Delta_P^{(k)}\,\gamma_k(\tau)\,\bar r_k(\tau)}
                   {p_0 + \sum_{k=0}^{K} \Delta_P^{(k)}\,\gamma_k(\tau)}.
  \end{aligned}
  \label{eq:fluid_efficacy_het_te}
\end{equation}

For analytical tractability, we assume monotonicity of the lift across regions:
\begin{equation}
  \Delta_P^{(0)} \leq \Delta_P^{(1)} \leq \cdots \leq \Delta_P^{(K)},
  \label{eq:hte_monotonicity}
\end{equation}
under which higher-predicted-risk individuals are more responsive to intervention.
Under monotonicity, the two-point structure of \Cref{thm:two_point_optimality} extends by region.

\begin{theorem}[Two-Point Optimality under Heterogeneous Treatment Effects]
\label{thm:hte_optimality}
Under \Cref{assump:joint_density,assump:marginal_positive,assump:conditional_smoothness,%
assump:monotone_calibration} and the monotonicity condition \eqref{eq:hte_monotonicity}, there exists an optimal active region $k_0$ such that the optimal flagging threshold is
\begin{equation}
  \begin{aligned}
    & \tau^{*A}_{\mathrm{HTE}}(\tfrac{M}{N} \mid p_0, \{\lambda_k\}, \{\Delta_P^{(k)}\})  = \max\Bigl\{ \lambda_{k_0},\, \min\bigl\{ \tilde{\tau}_{c,k_0},\, \tilde{\tau}^*_{\mathrm{score},k_0},\, \lambda_{k_0+1} \bigr\} \Bigr\},
  \end{aligned}
  \label{eq:hte_optimal_threshold}
\end{equation}
where $\tilde{\tau}_{c,k_0}$ and $\tilde{\tau}^*_{\mathrm{score},k_0}$ are the capacity-matching and score-optimal thresholds within region $k_0$.
The projection onto $[\lambda_{k_0}, \lambda_{k_0+1}]$ binds only in one boundary regime: at $k_0 = 0$ when capacity remains unfilled even with the entire population flagged (then $\tau^{*A}_{\mathrm{HTE}} = 0$). For $p_0 > 0$ the upper projection never binds, since the within-region score target is interior. Whenever the maximizer is interior to region $k_0$, the threshold reduces to $\min\{\tilde{\tau}_{c,k_0}, \tilde{\tau}^*_{\mathrm{score},k_0}\}$, mirroring \Cref{thm:two_point_optimality}.
\end{theorem}

\begin{proof}{Proof.}
The proof proceeds in three steps: we first reduce the objective on each region to the baseline single-threshold form, then derive within-region two-point targets, and finally select the active region. Recall $\bar r_k(\tau) = \mathbb{E}[r \mid \hat r^A \in [q_A(\max(\tau,\lambda_k)), q_A(\lambda_{k+1}))]$ from \eqref{eq:fluid_efficacy_het_te}, and write $e(\tau) = \mathbb{E}[r \mid \hat r^A = q_A(\tau)]$ for the marginal value at the threshold.

\emph{Step 1 (within-region reduction).}
Fix a region $k \in \{0,\ldots,K\}$ and restrict the threshold to $\tau \in [\lambda_k, \lambda_{k+1}]$. For such $\tau$, every region $j > k$ lies entirely above the threshold and is fully flagged, with $\gamma_j(\tau) = \lambda_{j+1} - \lambda_j$ and $\bar r_j(\tau) = \bar r_j(\lambda_j)$; every region $j < k$ is unflagged, with $\gamma_j(\tau) = 0$; and the active region contributes $\gamma_k(\tau) = \lambda_{k+1} - \tau$ with mean $\bar r_k(\tau) = \mathbb{E}[r \mid \hat r^A \in [q_A(\tau), q_A(\lambda_{k+1}))]$. Define the fixed request rate and value committed by the higher regions,
\begin{align}
  \Pi_k  &= \sum_{j=k+1}^{K} \Delta_P^{(j)}\,(\lambda_{j+1} - \lambda_j), \notag\\
  \Phi_k &= \sum_{j=k+1}^{K} \Delta_P^{(j)}\,(\lambda_{j+1} - \lambda_j)\,\bar r_j(\lambda_j). \notag
\end{align}

For $\tau \in [\lambda_k,\lambda_{k+1}]$ the served requests and efficacy reduce to
\begin{align}
  \tilde{N}(\tau) &= \min\{N D_k(\tau),\, M\}, \qquad
    D_k(\tau) = p_0 + \Pi_k + \Delta_P^{(k)}(\lambda_{k+1}-\tau), \label{eq:hte_Nk}\\
  \tilde{R}(\tau) &= \frac{A_k(\tau)}{D_k(\tau)}, \qquad
    A_k(\tau) = p_0\,\mathbb{E}[r] + \Phi_k + \Delta_P^{(k)}(\lambda_{k+1}-\tau)\,\bar r_k(\tau). \label{eq:hte_Rk}
\end{align}

\emph{Step 2 (within-region two-point targets).}
Equations \eqref{eq:hte_Nk}--\eqref{eq:hte_Rk} have the same structure as the baseline objective in the proof of \Cref{thm:two_point_optimality}, with effective lift $\Delta_P^{(k)}$ and the fixed block $(\Pi_k,\Phi_k)$ entering only as constants. Indeed, $D_k(\tau)$ is strictly decreasing with $D_k'(\tau) = -\Delta_P^{(k)}$, and using
$\frac{d}{d\tau}\bigl[(\lambda_{k+1}-\tau)\,\bar r_k(\tau)\bigr] = -e(\tau)$,
the score-optimal first-order condition is
\[
  \tilde{R}'(\tau) = \frac{\Delta_P^{(k)}}{D_k(\tau)}\bigl[\tilde{R}(\tau) - e(\tau)\bigr].
\]
At any $\tau$ with $\tilde{R}(\tau) = e(\tau)$, the map $\tau \mapsto \tilde{R}(\tau) - e(\tau)$ has derivative $-e'(\tau) < 0$ by \Cref{assump:monotone_calibration}, so it crosses zero at most once, exactly as in \Cref{lem:score_optimal_threshold_uniqueness}. Define the within-region capacity-matching and score-optimal \emph{targets} as the solutions to these conditions on the unrestricted line,
\begin{align}
  \tilde{\tau}_{c,k} &= \lambda_{k+1} - \tfrac{1}{\Delta_P^{(k)}}\bigl(\tfrac{M}{N} - p_0 - \Pi_k\bigr),
      \label{def:capacity_matching_hte}\\
  \tilde{\tau}^*_{\mathrm{score},k} &:\quad \tilde{R}(\tilde{\tau}^*_{\mathrm{score},k}) = e(\tilde{\tau}^*_{\mathrm{score},k}),
      \label{def:score_opt_hte}
\end{align}
where $\tilde{\tau}_{c,k}$ is the threshold equating $N D_k(\tau)$ with $M$, and $\tilde{\tau}^*_{\mathrm{score},k}$ is the unique stationary point of $\tilde{R}$. For $\tau \notin [\lambda_k, \lambda_{k+1}]$ these expressions are understood through the integral form $(\lambda_{k+1}-\tau)\,\bar r_k(\tau) = \int_{\tau}^{\lambda_{k+1}} e(u)\,du$, which extends $A_k$ and the first-order condition to all of $[0,1]$.
By the case analysis in the proof of \Cref{thm:two_point_optimality}, the unconstrained maximizer of the within-region objective is
\begin{equation}
  \tau^{*A(k)} = \min\bigl\{\tilde{\tau}_{c,k},\, \tilde{\tau}^*_{\mathrm{score},k}\bigr\}.
  \label{eq:within_region_opt}
\end{equation}
This two-point target says how far down region $k$ one must flag to balance capacity utilization against cannibalization; it is meaningful as an interior optimum only when it lands inside the region, i.e. $\tau^{*A(k)} \in [\lambda_k, \lambda_{k+1})$.

\emph{Step 3 (selecting the active region).}
The objective $\objfluid(\tau) = \tilde{N}(\tau)\tilde{R}(\tau)$ is continuous on the compact interval $[0,1] = \bigcup_{k=0}^{K} [\lambda_k, \lambda_{k+1}]$, so a global maximizer $\tau^{*A}_{\mathrm{HTE}}$ exists; let $k_0$ be a region containing it. Within region $k_0$ the objective is unimodal with peak at the two-point target \eqref{eq:within_region_opt}, increasing for $\tau < \tau^{*A(k_0)}$ and decreasing for $\tau > \tau^{*A(k_0)}$, so the within-region maximizer is the target projected onto the region:
\[
  \tau^{*A}_{\mathrm{HTE}}\!\left(\tfrac{M}{N} \mid p_0, \{\lambda_k\}, \{\Delta_P^{(k)}\}\right)
    = \max\bigl\{\lambda_{k_0},\, \min\{\tau^{*A(k_0)},\, \lambda_{k_0+1}\}\bigr\}.
\]
At interior boundaries the projection can be removed by re-choosing the region: if $\tau^{*A(k_0)} > \lambda_{k_0+1}$ with $k_0 < K$, the within-region maximum sits at $\lambda_{k_0+1}$, which also belongs to region $k_0+1$, so $k_0$ may be re-chosen upward; if $\tau^{*A(k_0)} < \lambda_{k_0}$ with $k_0 > 0$, the maximum sits at $\lambda_{k_0}$, which also belongs to region $k_0-1$, so $k_0$ may be re-chosen downward. Iterating, there always exists an active region whose target is interior,
except when the maximizer is pinned at the lower endpoint of $[0,1]$: the lower projection binds only at $k_0 = 0$, where capacity is abundant enough that flagging the entire population does not fill it and $\tau^{*A}_{\mathrm{HTE}} = 0$. For $p_0 > 0$ the upper projection never binds, since the within-region score target is interior.
In all other cases $\tau^{*A(k_0)} \in [\lambda_{k_0},\lambda_{k_0+1})$ and the display reduces to $\tau^{*A}_{\mathrm{HTE}} = \min\{\tilde{\tau}_{c,k_0}, \tilde{\tau}^*_{\mathrm{score},k_0}\}$, as stated. The monotonicity condition \eqref{eq:hte_monotonicity} is not used in Steps 1--3; it enters only through the top-down search below.
This yields a simple top-down search for $k_0$. Under the monotonicity condition $\Delta_P^{(0)} \le \cdots \le \Delta_P^{(K)}$, fully flagging a higher-responsiveness region is never dominated by leaving it partially flagged, so the search begins at the top region $K$ and moves downward: at region $k$, whenever the target \eqref{eq:within_region_opt} falls below $\lambda_k$, region $k$ is flagged in full and the search drops to region $k-1$. If the target at region $K$ exceeds $1$, the search stops immediately with $\tau^{*A}_{\mathrm{HTE}} = 1$; if no target is interior, it terminates at region $0$ with $\tau^{*A}_{\mathrm{HTE}} = 0$; otherwise it stops at the first region whose two-point target is interior, namely $k_0$. In every case the search solves at most $K+1$ within-region problems.
\end{proof}

%% file: section/appendix_bilevel.tex
\section{Additional Results}\label{sec:additional_result}
\subsection{Bilevel Interventions}\label{sec:bilevel}

\subsubsection{Model and Optimal Thresholds}\label{app:bilevel_model}
The baseline model assumes a uniform lift $\Delta_P$ across all flagged individuals. In practice, providers often deploy interventions of differing intensity, and the more intensive tier typically faces its own, tighter capacity constraint because it draws on scarce staff time. In healthcare, high-risk patients may receive a personal phone call, bounded by nurse availability, while moderate-risk patients receive only a text reminder, which is essentially costless to scale; in education, advisors may dedicate direct attention to high-priority alerts but allow routine dashboard flags to accumulate passively. Our framework extends naturally to capture such bilevel interventions.

Consider a bilevel thresholding policy defined by thresholds $\tau_W \leq \tau_S$ that partition the population into three tiers. Individuals with predicted scores above $\tau_S$ receive a \emph{strong} intervention with lift $\Delta^S_P$, subject to a strong-tier capacity $M_S \leq M$ that caps the number of individuals who can receive the strong intervention\footnote{If strong-tier individuals receive both interventions concurrently, $\Delta^S_P$ represents the combined lift from the joint intervention.}; those with scores in $[\tau_W, \tau_S)$ receive a \emph{weak} intervention with lift $\Delta^W_P < \Delta^S_P$; and the remainder receive no intervention. The fluid expressions for served requests and per-request efficacy extend directly:
\begin{equation}
\begin{aligned}
    \tilde{N}\!\left(\tau_W, \tau_S \mid M, N, p_0, \Delta_P^W, \Delta_P^S\right)
    = \min\Bigl\{
        N\Bigl[p_0
            + \Delta_P^W(1-\tau_W) + (\Delta_P^S - \Delta_P^W)(1-\tau_S)\Bigr],
        \;M\Bigr\},
\end{aligned}
\end{equation}
\begin{equation}
\begin{aligned}
    \tilde{R}\!\left(\tau_W, \tau_S \mid A, p_0, \Delta_P^W, \Delta_P^S\right)
    &= \frac{p_0\,\mathbb{E}[r]}
        {p_0 + \Delta_P^W(1-\tau_W) + (\Delta_P^S-\Delta_P^W)(1-\tau_S)} \\
    &+ \frac{\Delta_P^W(1-\tau_W)\,
        \mathbb{E}\!\left[r \mid \hat{r}^A \ge q_A(\tau_W)\right]}
        {p_0 + \Delta_P^W(1-\tau_W) + (\Delta_P^S-\Delta_P^W)(1-\tau_S)} \\
    &+ \frac{(\Delta_P^S-\Delta_P^W)(1-\tau_S)\,
        \mathbb{E}\!\left[r \mid \hat{r}^A \ge q_A(\tau_S)\right]}
        {p_0 + \Delta_P^W(1-\tau_W) + (\Delta_P^S-\Delta_P^W)(1-\tau_S)}.
\end{aligned}
\end{equation}
The optimal bilevel policy solves
\begin{equation}
    \label{eq:bilevel_opt}
\begin{aligned}
    \max_{\tau_W,\,\tau_S}\ \;
        & \tilde{N}(\tau_W, \tau_S \mid M,N,p_0,\Delta_P^W,\Delta_P^S) \cdot \tilde{R}(\tau_W, \tau_S \mid A, p_0, \Delta_P^W, \Delta_P^S) \\
    \text{s.t.}\ \;
        & 0 \le \tau_W \le \tau_S \le 1, \qquad 1 - \tau_S \le M_S/N.
\end{aligned}
\end{equation}

The two-point structure of the baseline characterization carries over, now stratified into two regimes by whether the strong-tier capacity $M_S$ binds. When $M_S$ is slack, both thresholds may be chosen freely, and it is always optimal to collapse them onto a single threshold with effective lift $\Delta^S_P$: weaker interventions are strictly dominated by extending the strong one to the same individuals. When $M_S$ binds, the threshold for the strong intervention is pinned at $\tau^*_S = 1 - M_S/N$, and the remaining problem is essentially a single-threshold problem for the weak intervention, conditioned on the requests induced by the strong-tier intervention. In either regime, the optimal thresholds admit a two-point characterization, and inherit the same insights about the tradeoff between cannibalization and capacity underutilization. Theorem~\ref{thm:bilevel_optimality} formalizes this characterization.

Throughout the analysis, let $D(\tau_W,\tau_S) = p_0 + \Delta_P^W(1-\tau_W) + (\Delta_P^S - \Delta_P^W)(1-\tau_S)$ denote the total request rate and $H^A(\tau) = (1-\tau)\,\mathbb{E}[r \mid \hat{r}^A \ge q_A(\tau)]$, which is weakly decreasing in $\tau$.
\begin{theorem}[Two-Point Optimality for Bilevel Interventions]
\label{thm:bilevel_optimality}
Under \Cref{assump:joint_density,assump:marginal_positive,assump:conditional_smoothness,assump:monotone_calibration}, the optimal bilevel thresholds admit the following two-point characterization. When the strong-tier capacity constraint $1-\tau_S \le M_S/N$ is not binding, the optimal thresholds coincide:
\begin{equation}
    \tau_W^* = \tau_S^* =
        \min\left\{
            \tau_c\left(\tfrac{M}{N} \mid p_0,\, \Delta_P^S\right),
            \tau_{\mathrm{score}}^*(p_0,\, \Delta_P^S)
        \right\},
\end{equation}
reducing to the single-threshold characterization of \Cref{thm:two_point_optimality}
with effective lift $\Delta_P^S$. When $1 - \tau_S = M_S/N$ binds, the targeting threshold for the strong intervention is fixed at $\tau_S^* = 1 - M_S/N$, and the threshold for the weak intervention satisfies
\begin{equation}
    \tau_W^* = \min\!\left\{\tilde{\tau}^*_{\mathrm{score},W},\; \tilde{\tau}_{c,W}\right\},
\end{equation}
where $\tilde{\tau}_{c,W}$ and $\tilde{\tau}^*_{\mathrm{score},W}$ are the
capacity-matching and score-optimal thresholds with $\tau_S$ fixed at $1 - M_S/N$, defined in \eqref{def:capacity_matching_weak} and \eqref{def:score_opitmal_weak} in the proof below.
\end{theorem}

\begin{proof}{Proof.}
Let $\tilde \tau_{c,W}$ be the capacity-matching threshold with $\tau_S$ fixed at $1 - M_S/N$:
\begin{equation}
\label{def:capacity_matching_weak}
\begin{aligned}
\tilde \tau_{c,W} = \tau_c\!\left(
            \tfrac{M}{N} - (\Delta_P^S - \Delta_P^W)\tfrac{M_S}{N}\mid
            p_0, \Delta_P^W\right),
\end{aligned}
\end{equation}
and $\tilde{\tau}_{\mathrm{score},W}^*$ be the score-optimal threshold with $\tau_S$ fixed at $1 - M_S/N$:
\begin{equation}
\label{def:score_opitmal_weak}
\begin{aligned}
    & \tilde{\tau}_{\mathrm{score},W}^* = \operatorname{arg\,max}_{0\,\le\,\tau_W\,\le\,1-M_S/N} \\
        &\frac{p_0\,\mathbb{E}[r]}{p_0 + (\Delta_P^S-\Delta_P^W)\tfrac{M_S}{N} +\Delta_P^W(1-\tau_W)}\\
        &+ \frac{\Delta_P^W(1-\tau_W)\,
            \mathbb{E}\!\left[r \mid \hat{r}^A \ge q_A(\tau_W)\right]}{p_0 + (\Delta_P^S-\Delta_P^W)\tfrac{M_S}{N} +\Delta_P^W(1-\tau_W)} \\
        &+ \frac{(\Delta_P^S-\Delta_P^W)\tfrac{M_S}{N}\,
            \mathbb{E}\!\left[r \mid \hat{r}^A \ge q_A\!\left(
                1-\tfrac{M_S}{N}\right)\right]}{p_0 + (\Delta_P^S-\Delta_P^W)\tfrac{M_S}{N} +\Delta_P^W(1-\tau_W)}.
\end{aligned}
\end{equation}

Problem~\eqref{eq:bilevel_opt} is equivalent to maximizing over two regimes:
\begin{align}
    \max\Bigg\{
        &\max_{\substack{0 \le \tau_W \le \tau_S \le 1,\ 1-\tau_S \le M_S/N \\
                         D(\tau_W,\tau_S) \le M/N}}
            D(\tau_W,\tau_S)\,\tilde{R}(\tau_W,\tau_S), \notag \\
        &\max_{\substack{0 \le \tau_W \le \tau_S \le 1,\ 1-\tau_S \le M_S/N \\
                         D(\tau_W,\tau_S) > M/N}}
            \frac{M}{N}\,\tilde{R}(\tau_W,\tau_S)
    \Bigg\}.
    \label{eq:two_regime}
\end{align}

\paragraph{Under-capacity regime ($D(\tau_W,\tau_S) \le M/N$).}
The objective reduces to maximizing $\Delta_P^W H^A(\tau_W) + (\Delta_P^S - \Delta_P^W)H^A(\tau_S)$ (dropping the constant $p_0\,\mathbb{E}[r]$) subject to the capacity and ordering constraints. Since
$H$ is weakly decreasing, the objective is weakly decreasing in both arguments, so any maximizer lies on the boundary of the feasible set. The KKT
conditions with multipliers $\lambda_1 \ge 0$ (capacity), $\lambda_2 \ge 0$ (strong-tier, $1-\tau_S \le M_S/N$), and $\mu \ge 0$ (ordering, $\tau_W \le
\tau_S$) are
\begin{align}
    &-\Delta_P^W\,\mathbb{E}[r \mid \hat{r}^A = q_A(\tau_W)]
        + \lambda_1 \Delta_P^W - \mu
        \;=\; 0, \label{eq:foc_tau_w_supp} \\
    &-(\Delta_P^S - \Delta_P^W)\,
        \mathbb{E}[r \mid \hat{r}^A = q_A(\tau_S)] \notag \\
    &\quad + \lambda_1(\Delta_P^S - \Delta_P^W)
        + \lambda_2 + \mu
        \;=\; 0. \label{eq:foc_tau_s_supp}
\end{align}

\textbf{Case U1: $\lambda_1 > 0$, $\lambda_2 = 0$ (only capacity binds).}
If the ordering constraint also binds ($\mu > 0$), then $\tau_W = \tau_S = \tau_c$, where $p_0 + \Delta_P^S(1-\tau_c) = M/N$, giving
\begin{equation}
    \tau_c \;=\; 1 - \frac{1}{\Delta_P^S}\!\left(\frac{M}{N} - p_0\right).
\end{equation}
If instead $\mu = 0$, \eqref{eq:foc_tau_w_supp}--\eqref{eq:foc_tau_s_supp} give $\mathbb{E}[r \mid \hat{r}^A = q_A(\tau_W)] = \mathbb{E}[r \mid \hat{r}^A = q_A(\tau_S)] = \lambda_1$, so $\tau_W = \tau_S$ by strict monotonicity (\Cref{assump:monotone_calibration}), and the same conclusion follows.

\textbf{Case U2: $\lambda_1 > 0$, $\lambda_2 > 0$ (both bind).}
Then $1-\tau_S = M_S/N$ and
\begin{equation}
    \tau_W \;=\; 1 - \frac{1}{\Delta_P^W}\!
        \left(\frac{M}{N} - p_0
              - (\Delta_P^S - \Delta_P^W)\frac{M_S}{N}\right).
\end{equation}

\textbf{Case U3: $\lambda_2 > 0$, $\lambda_1 = 0$ (only the strong-tier constraint binds).}
Then $\tau_S = 1 - M_S/N$. With $\lambda_1 = 0$, \eqref{eq:foc_tau_w_supp} gives $\mu = -\Delta_P^W\,\mathbb{E}[r \mid \hat{r}^A = q_A(\tau_W)] < 0$ for any interior $\tau_W$ (the marginal value is strictly positive by \Cref{assump:monotone_calibration} and $r \ge 0$), contradicting $\mu \ge 0$. Hence no interior $\tau_W$ is stationary, and the maximizer sits at the boundary $\tau_W = 0$, where the multiplier of the box constraint $\tau_W \ge 0$ (omitted from the display) absorbs the slack. This corner lies in the under-capacity regime only when capacity remains slack with the entire population flagged, $p_0 + \Delta_P^W + (\Delta_P^S - \Delta_P^W)\tfrac{M_S}{N} \le \tfrac{M}{N}$; in that case $(\tau_W^*, \tau_S^*) = (0,\, 1 - M_S/N)$, consistent with the two-point characterization since the capacity-matching threshold clamps at $\tilde\tau_{c,W} = 0$.

\paragraph{Over-capacity regime ($D(\tau_W,\tau_S) > M/N$).}
Here $\tilde{N} = M$ and the objective is $\tfrac{M}{N}\tilde{R}$. Let $A$ and $B$ denote the numerator and denominator of $\tilde{R}$, with $H^{A\prime}(\tau) = -\mathbb{E}[r \mid \hat{r}^A = q_A(\tau)]$. In this regime the binding capacity direction is the lower bound $D(\tau_W,\tau_S) \ge M/N$, so the capacity constraint enters as $M/N - D \le 0$. The KKT conditions with multipliers $\lambda_1 \ge 0$ (strong-tier), $\lambda_2 \ge 0$ (capacity lower bound), and $\lambda_3 \ge 0$ (ordering) are
\begin{align}
    &\frac{\Delta_P^W(H^{A\prime}(\tau_W) B + A)}{B^2}
        - \lambda_2 \Delta_P^W - \lambda_3
        \;=\; 0, \label{eq:foc_tau1_supp} \\
    &\frac{(\Delta_P^S-\Delta_P^W)(H^{A\prime}(\tau_S) B + A)}{B^2}
        + \lambda_1 + \lambda_2(\Delta_P^W - \Delta_P^S) + \lambda_3
        \;=\; 0. \label{eq:foc_tau2_supp}
\end{align}
Combining \eqref{eq:foc_tau1_supp} and \eqref{eq:foc_tau2_supp} yields
\begin{align}
    \lambda_1 + \lambda_3\frac{\Delta_P^S}{\Delta_P^W}
        \;=\;& \frac{\Delta_P^S - \Delta_P^W}{B}
            \Bigl[\mathbb{E}[r \mid \hat{r}^A = q_A(\tau_S)] \notag \\
        &\quad - \mathbb{E}[r \mid \hat{r}^A = q_A(\tau_W)]\Bigr].
    \label{eq:combined_foc_supp}
\end{align}
Since $\tau_W \le \tau_S$, the right-hand side is nonnegative by \Cref{assump:monotone_calibration}, consistent with $\lambda_1,\lambda_3 \ge 0$.

\textbf{Case O1: $\lambda_1 > 0$ (strong-tier binds).}
$\tau_S^* = 1 - M_S/N$, and the remaining problem is a single-threshold problem in $\tau_W$ whose baseline is augmented by the fixed high-value request block $(\Delta_P^S-\Delta_P^W)\tfrac{M_S}{N}$ contributed by the strong tier, entering both the numerator and the denominator of $\tilde R$. If the capacity bound is slack ($\lambda_2 = 0$) and the ordering is slack ($\tau_W < \tau_S$, so $\lambda_3 = 0$), then \eqref{eq:foc_tau1_supp} reduces to $\tilde{R} = \mathbb{E}[r \mid \hat{r}^A = q_A(\tau_W)]$, the score-optimal condition for $\tau_W$ with $\tau_S$ fixed at $1 - M_S/N$; thus $\tau_W^*$ solves
\begin{equation}
\begin{aligned}
    & \max_{0 \,\le\, \tau_W \,\le\, 1 - M_S/N} \frac{
        p_0\,\mathbb{E}[r]}{
        p_0
        + (\Delta_P^S - \Delta_P^W)\tfrac{M_S}{N}
        + \Delta_P^W(1-\tau_W)
    }\\
    &\quad +\frac{
        \Delta_P^W(1-\tau_W)\,\mathbb{E}\!\left[r \mid \hat{r}^A > q_A(\tau_W)\right]
         }{
        p_0
        + (\Delta_P^S - \Delta_P^W)\tfrac{M_S}{N}
        + \Delta_P^W(1-\tau_W)
    }\\
    &\quad +\frac{
        (\Delta_P^S - \Delta_P^W)\tfrac{M_S}{N}\,
          \mathbb{E}\!\left[r \mid \hat{r}^A > q_A\!\left(1 - \tfrac{M_S}{N}\right)\right]
    }{
        p_0
        + (\Delta_P^S - \Delta_P^W)\tfrac{M_S}{N}
        + \Delta_P^W(1-\tau_W)
    },
\end{aligned}
\end{equation}
giving $\tau_W^* = \tilde{\tau}_{\mathrm{score},W}^*$. If instead the capacity bound binds ($\lambda_2 > 0$), $\tau_W$ is pinned at the capacity-matching threshold with $\tau_S$ fixed, and by the two-point structure of \Cref{thm:two_point_optimality} applied to this single-threshold problem,
\begin{equation}
    \tau_W^* \;=\; \min\!\left\{
        \tilde{\tau}_{\mathrm{score},W}^*,\;
        \tilde{\tau}_{c,W}
    \right\}.
\end{equation}

\textbf{Case O2: $\lambda_1 = 0$ (strong-tier slack).}
Suppose $\tau_W < \tau_S$. Then the ordering constraint is slack, so $\lambda_3 = 0$ by complementary slackness, and \eqref{eq:combined_foc_supp} forces
\[
    0 = \frac{\Delta_P^S - \Delta_P^W}{B}
        \Bigl[\mathbb{E}[r \mid \hat{r}^A = q_A(\tau_S)]
            - \mathbb{E}[r \mid \hat{r}^A = q_A(\tau_W)]\Bigr].
\]
Since $\Delta_P^S > \Delta_P^W$ and, by \Cref{assump:monotone_calibration}, $\mathbb{E}[r \mid \hat{r}^A = q_A(\tau_S)] > \mathbb{E}[r \mid \hat{r}^A = q_A(\tau_W)]$ whenever $\tau_W < \tau_S$, the right-hand side is strictly positive, a contradiction. Hence $\tau_W = \tau_S$, and the problem collapses to the single-threshold characterization of \Cref{thm:two_point_optimality} with effective lift $\Delta_P^S$: if the capacity constraint binds ($\lambda_2 > 0$),
\begin{equation}
    \tau_W^* = \tau_S^* = \tau_c\!\left(\tfrac{M}{N} \mid p_0,\, \Delta_P^S\right),
\end{equation}
and otherwise
\begin{equation}
    \tau_W^* = \tau_S^* = \tau_{\mathrm{score}}^*(p_0,\, \Delta_P^S).
\end{equation}

The two subcases of Case O2 establish the first part of the theorem, with Case U1 supplying the shared capacity-matching corner $\tau_c$; Case O1 establishes the second part, with Cases U2 and U3 supplying the capacity-matching corner $\tilde{\tau}_{c,W}$ (interior and clamped at $0$, respectively). Taking the maximum over the two regimes in \eqref{eq:two_regime} therefore yields the stated min-form in each regime.
\end{proof}

\Cref{prop:tilde_tau_geq_tau_score} reveals a monotonicity property of the bilevel solution: when a stronger intervention is available, the score-optimal threshold for the weak intervention $\tilde{\tau}^*_{\mathrm{score},W}$ is at least as large as the single-tier score-optimal threshold $\tau^*_{\mathrm{score}}(p_0, \Delta_P^W)$, up to the feasibility cap $1 - M_S/N$ imposed by the strong tier. Intuitively, the strong-tier intervention has a higher lift on request probability and therefore mitigates cannibalization more effectively. This allows the weak-tier intervention to be more selective and target a smaller population, pushing the score-optimal threshold upward.
\begin{proposition}\label{prop:tilde_tau_geq_tau_score}
The constrained score-optimal threshold $\tilde{\tau}_{\mathrm{score},W}^*$ defined in \Cref{def:score_opitmal_weak} satisfies
$$\tilde{\tau}_{\mathrm{score},W}^* \ge \min\left\{\tau_{\mathrm{score}}^*(p_0, \Delta_P^W),\; 1 - \tfrac{M_S}{N}\right\}.$$
In particular, whenever the single-tier score-optimal threshold is feasible for the weak tier, $\tau_{\mathrm{score}}^*(p_0, \Delta_P^W) \le 1 - M_S/N$, it holds that $\tilde{\tau}_{\mathrm{score},W}^* \ge \tau_{\mathrm{score}}^*(p_0, \Delta_P^W)$.
\end{proposition}

\begin{proof}{Proof.}
Let $e(\tau) = \mathbb{E}[r \mid \hat r^A = q_A(\tau)]$ denote the marginal value at the threshold. The per-request efficacy is $R(\tau) = A(\tau)/D_W(\tau)$ with
\[
  A(\tau) = p_0\,\mathbb{E}[r] + \Delta_P^W(1-\tau)\,\mathbb{E}[r \mid \hat r^A \ge q_A(\tau)],
  \qquad
  D_W(\tau) = p_0 + \Delta_P^W(1-\tau).
\]
Using $\tfrac{d}{d\tau}\bigl[(1-\tau)\mathbb{E}(r \mid \hat r^A \ge q_A(\tau))\bigr] = -e(\tau)$,
\[
  R'(\tau) = \frac{\Delta_P^W}{D_W(\tau)}\bigl(R(\tau) - e(\tau)\bigr),
\]
and the score-optimal threshold satisfies $R(\tau^*_{\mathrm{score}}) = e(\tau^*_{\mathrm{score}})$.

The weak-tier problem with $\tau_S$ fixed at $1 - M_S/N$ adds the fixed request-rate and value blocks nudged by the strong tier,
\[
  \Pi = (\Delta_P^S - \Delta_P^W)\tfrac{M_S}{N},
  \qquad
  \Phi = (\Delta_P^S - \Delta_P^W)\tfrac{M_S}{N}\,\bar r_S,
  \qquad
  \bar r_S := \mathbb{E}\!\left[r \mid \hat r^A \ge q_A\!\left(1 - \tfrac{M_S}{N}\right)\right].
\]
Therefore, the per-request efficacy can be written as
\[
  \tilde R(\tau) = \frac{A(\tau) + \Phi}{D_W(\tau) + \Pi}.
\]
Since $\Pi, \Phi$ are constant in $\tau_W$, $\tilde\tau^*_{\mathrm{score},W}$ satisfies $\tilde R = e$ when interior and equals $1 - M_S/N$ otherwise. As in \Cref{lem:score_optimal_threshold_uniqueness}, both $R - e$ and $\tilde R - e$ have strictly negative derivative at any zero. Therefore, each crosses zero at most once.

We can decompose $\tilde R$ into two parts:
\[
  \tilde R(\tau) - e(\tau)
  = \bigl(R(\tau) - e(\tau)\bigr)
  + \frac{\Pi}{D_W(\tau) + \Pi}\bigl(\bar r_S - R(\tau)\bigr).
\]
For any $\tau \le 1 - M_S/N$, monotone calibration (\Cref{assump:monotone_calibration}) gives $\bar r_S \ge \mathbb{E}[r \mid \hat r^A \ge q_A(\tau)] \ge R(\tau)$, where the last inequality holds because $R(\tau)$ is a weighted average of the tail mean and the smaller population mean $\mathbb{E}[r]$.
We have
\[
  \tilde R(\tau) - e(\tau) \;\ge\; R(\tau) - e(\tau),
  \qquad \tau \le 1 - M_S/N.
\]
The constrained weak-tier optimum is $\tilde\tau^*_{\mathrm{score},W} = \min\{\tau^\dagger,\, 1 - M_S/N\}$, where $\tau^\dagger$ is the unique zero of $\tilde R - e$ on $[0,1]$ (set $\tau^\dagger := 1 - M_S/N$ if no zero exists on $[0, 1-M_S/N]$; then $\tilde R - e > 0$ throughout the feasible range and the objective increases up to the endpoint). We distinguish two cases. First, suppose $\tau^*_{\mathrm{score}}(p_0,\Delta_P^W) \le 1 - M_S/N$. Taking the inequality above at $\tau = \tau^*_{\mathrm{score}}$ and using $R(\tau^*_{\mathrm{score}}) = e(\tau^*_{\mathrm{score}})$,
\[
\tilde R(\tau^*_{\mathrm{score}}) - e(\tau^*_{\mathrm{score}}) \;\ge\; R(\tau^*_{\mathrm{score}}) - e(\tau^*_{\mathrm{score}}) \;=\; 0.
\]
Since $\tilde R - e$ is positive before its unique crossing at $\tau^\dagger$ and negative after, $\tau^*_{\mathrm{score}}$ has not yet passed the crossing, i.e., $\tau^*_{\mathrm{score}} \le \tau^\dagger$; as $\tau^*_{\mathrm{score}} \le 1 - M_S/N$ as well, $\tau^*_{\mathrm{score}} \le \min\{\tau^\dagger, 1 - M_S/N\} = \tilde\tau^*_{\mathrm{score},W}$. Second, suppose $\tau^*_{\mathrm{score}}(p_0,\Delta_P^W) > 1 - M_S/N$. By the single crossing of $R - e$ at $\tau^*_{\mathrm{score}}$, we have $R - e > 0$ on $[0, 1 - M_S/N]$, hence $\tilde R - e > 0$ there as well, so $\tau^\dagger \ge 1 - M_S/N$ and $\tilde\tau^*_{\mathrm{score},W} = 1 - M_S/N$. In both cases $\tilde\tau^*_{\mathrm{score},W} \ge \min\{\tau^*_{\mathrm{score}}(p_0,\Delta_P^W),\, 1 - M_S/N\}$, which concludes the proof.

\end{proof}

\begin{figure}[h]
    \centering
    \includegraphics[width=0.85\linewidth]{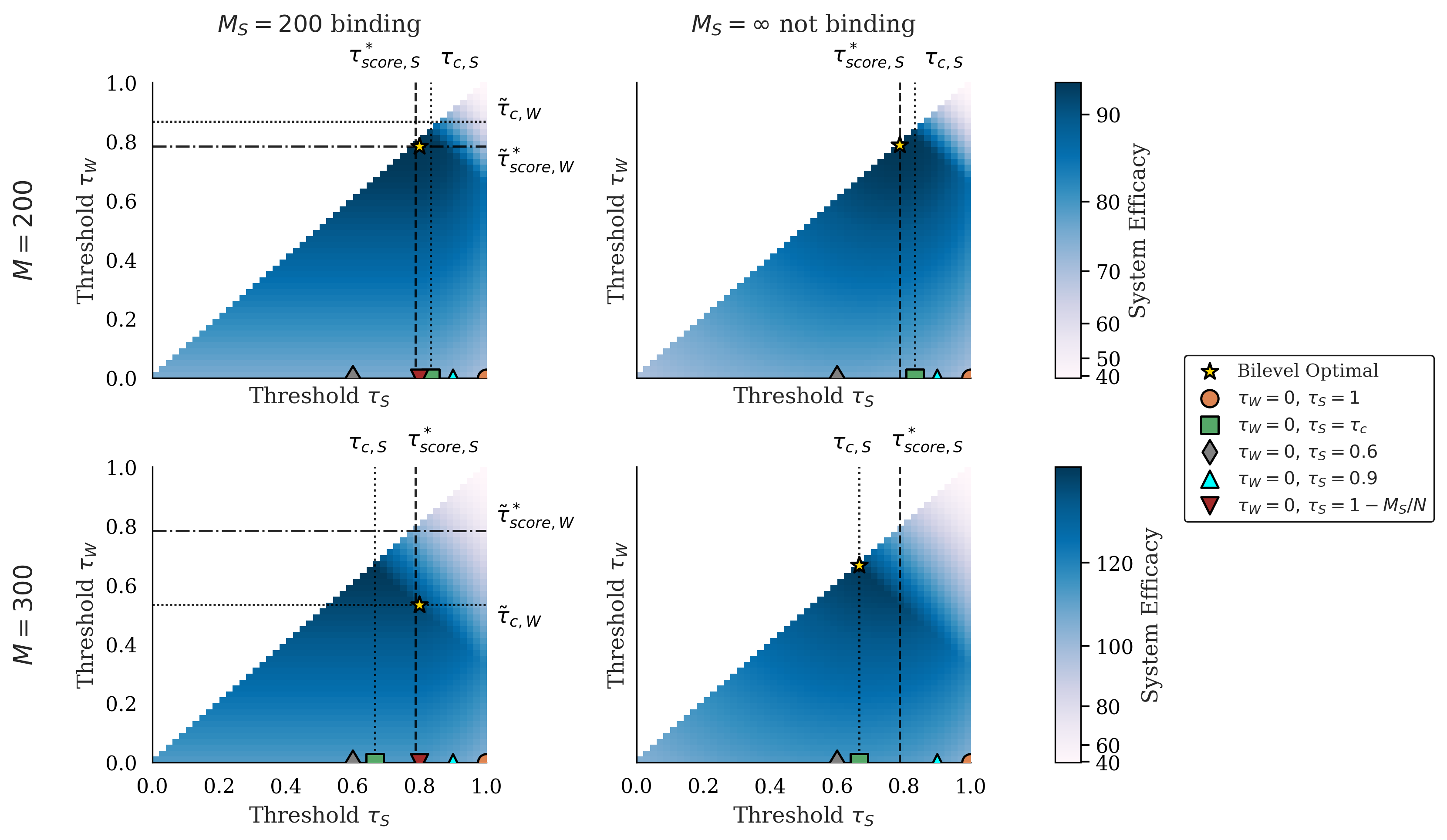}
    \caption{System efficacy as a function of the bilevel thresholds $(\tau_W, \tau_S)$ under varying total capacity $M \in \{200, 300\}$ and strong-tier
    capacity $M_S \in \{200, \infty\}$. The color gradient represents system efficacy, with darker shading indicating higher values. The star marks the bilevel optimal $(\tau_W^*, \tau_S^*)$, and the remaining markers show benchmark policies. Dashed vertical lines indicate the capacity-matching threshold $\tau_{c,S}$ and the score-optimal threshold $\tau^*_{\mathrm{score},S}$; in the binding regime (right panels), the dot-dashed horizontal line indicates $\tilde{\tau}^*_{\mathrm{score},W}$. Parameters: $N = 1000$, $\Delta_P^W = 0.3$,
    $\Delta_P^S = 0.6$, $r \sim 0.7\,\mathrm{Beta}(2,10) + 0.3\,\mathrm{Beta}(8,2)$, $\hat{r} = r + \mathcal{N}(0, 0.1^2)$.}
    \label{fig:bilevel_heatmap}
\end{figure}

\Cref{fig:bilevel_heatmap} illustrates the two-regime structure of the optimal bilevel policy across four configurations. In the left panels ($M_S = \infty$), the optimal solution lies on the diagonal $\tau_W = \tau_S$, confirming that the bilevel policy collapses to a single threshold
with effective lift $\Delta_P^S$.
As total capacity rises from $M = 200$ to $M = 300$, the optimum shifts leftward, reflecting the shift from a capacity-underutilization binding constraint to a cannibalization binding constraint. In the right panels ($M_S = 200$, binding), the optimum sits on the segment $\tau_S = 1 - M_S/N$ with $\tau_W^* = \min\{\tilde{\tau}^*_{\mathrm{score},W}, \tilde{\tau}_{c,W}\}$, matching the second part of \Cref{thm:bilevel_optimality}.

\subsubsection{Case Study Illustration}\label{app:case_study_bilevel}

In this subsection, we illustrate \Cref{thm:bilevel_optimality} on the NYP case study data. The weak intervention corresponds to a passive EHR dashboard flag visible on chart review, and the strong intervention to a direct page to the nurse, with stipulated lifts $\Delta_P^W = 0.1$ and $\Delta_P^S = 0.5$.

We compute the bilevel optimum $(\tau_W^*, \tau_S^*)$ characterized in \Cref{thm:bilevel_optimality} and compare it against two benchmark policies that fix $\tau_W = 0$ so that all patients receive at least the weak intervention. The two benchmarks differ in how the strong threshold is set: $\tau_S = \tau_{c,S}(M/N)$ (capacity-matching) and $\tau_S = 0.8$ (fixed prediction-based). We consider two settings for the strong-tier capacity, $M_S = 20$ (binding) and $M_S = \infty$ (non-binding), and report the relative gap to the bilevel optimum in \Cref{fig:sepsis_bilevel}.

The patterns match the theory in \Cref{thm:bilevel_optimality}. When the strong-tier capacity binds (panels (a) and (c)), the bilevel optimum substantially outperforms both universal-weak-coverage benchmarks, with losses up to roughly 60\% at low capacity under the fixed prediction-based threshold. When the strong-tier capacity does not bind (panels (b) and (d)), the bilevel optimum collapses to a single threshold with effective lift $\Delta_P^S$ (\Cref{thm:bilevel_optimality}), and the gap to the benchmark policies narrows. Across both settings and both algorithms, the capacity-matching $\tau_S$ consistently outperforms the fixed prediction-based $\tau_S = 0.8$.

\begin{figure}[h]
\centering
\includegraphics[width=\linewidth]{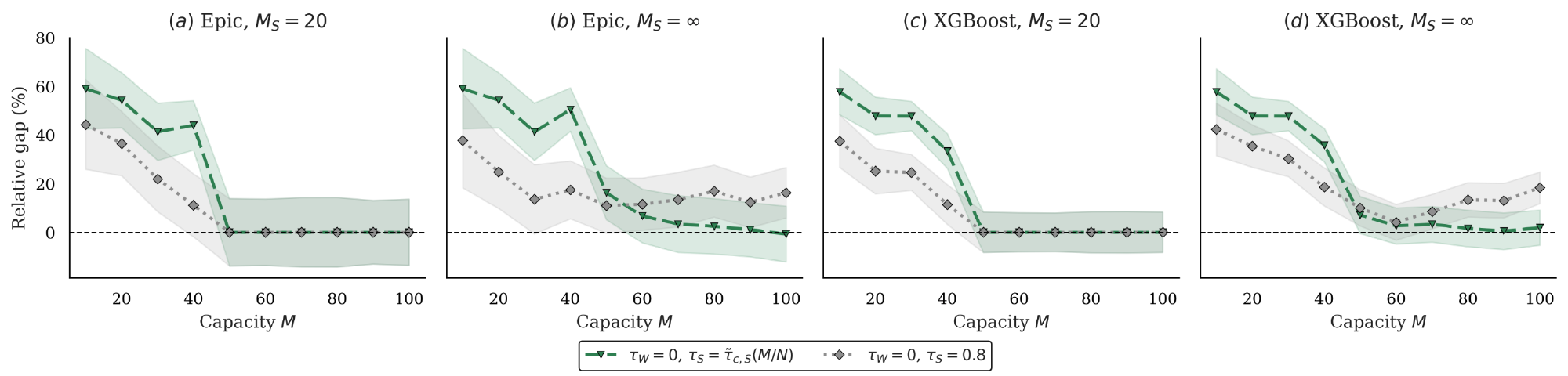}
\caption{Illustration of bilevel interventions in the NYP sepsis case study. Panels (a) and (b) report results for Epic under strong-tier capacity $M_S = 20$ (binding) and $M_S = \infty$ (non-binding); panels (c) and (d) report the analogous results for XGBoost. Both benchmark policies fix the weak threshold at $\tau_W = 0$ and differ in the strong threshold: $\tau_S = \tau_{c,S}(M/N)$ (capacity-matching) and $\tau_S = 0.8$ (fixed prediction-based). Relative gap is measured against the bilevel optimum $(\tau_W^*, \tau_S^*)$ characterized in \Cref{thm:bilevel_optimality}. Parameters: $N = 200$, $p_0 = 0.1$, $\Delta_P^W = 0.1$, $\Delta_P^S =0.5$. Shaded regions indicate $\pm$ one standard error across trials.}
\label{fig:sepsis_bilevel}
\end{figure}

%% file: section/appendix_result.tex
\subsection{Robustness of Optimal Thresholds to Estimated Parameters}\label{app:robustness-estimated-params}
The two-point threshold $\tau^{*A}$ in \Cref{thm:two_point_optimality} depends on the behavioral parameters $(p_0,\Delta_P)$ and on the risk distribution, none of which is known to the planner in practice. A natural question is how the performance of the optimal thresholds \Cref{thm:two_point_optimality} may degrade once these inputs are replaced by estimates from past requests. We show from a Monte Carlo simulations that the estimates of behavioral parameters converge to the ground truth, the plug-in threshold $\hat\tau^{*A}$ and the resulting system efficacy converge to optimum when the training sample size increases.

\paragraph{Setup.} We fix the true parameters at $p_0 = 0.1$ and $\Delta_P = 0.5$, and let the predicted score distribution be the mixture $0.7\,\mathrm{Beta}(2,10) + 0.3\,\mathrm{Beta}(8,2)$. Assume a historical thresholding policy with threshold $\tau_0 = 0.7$ generates $N$ request observations $\{(r_i,\, d_i,\, y_i)\}_{i=1}^{N}$.
We have the below estimators
\begin{equation*}
  \hat p_0 \;=\; \frac{1}{|\mathcal{U}|}\sum_{i \in \mathcal{U}} y_i,
  \qquad
  \hat\Delta_P \;=\; \frac{\bar y_{\mathcal{C}} - \hat p_0}{\bar r_{\mathcal{C}}},
\end{equation*}
where $\mathcal{C} = \{i : d_i = 1\}$, $\mathcal{U} = \{i : d_i = 0\}$, and $\bar y_{\mathcal{C}}$, $\bar r_{\mathcal{C}}$ denote sample averages over the flagged group. The plug-in threshold $\hat\tau^{*A}$ is then computed from the two-point characterization in \Cref{thm:two_point_optimality} using $(\hat p_0,\hat\Delta_P)$ and the empirical risk distribution, with target capacity ratio $M/N = 0.2$. We vary the log size $N \in \{100,\,500,\,2000,\,10000\}$ and average across $30$ replications per $N$.

\paragraph{Parameter convergence.} \Cref{fig:param-estimates-convergence} reports the mean of the moment estimators across replications, with shaded $\pm 1$ standard-deviation bands. Both estimators are essentially unbiased and concentrate at the parametric $1/\sqrt{N}$ rate. The across-replication standard deviation of $\hat p_0$ falls from $0.045$ at $N=100$ to $0.003$ at $N=10{,}000$; the corresponding standard deviation of $\hat\Delta_P$ falls from $0.135$ to $0.015$. The dispersion of $\hat\Delta_P$ is larger because it is a ratio estimator whose denominator depends on $\hat p_0$.

\begin{figure}[ht]
  \centering
  \includegraphics[width=0.6\linewidth]{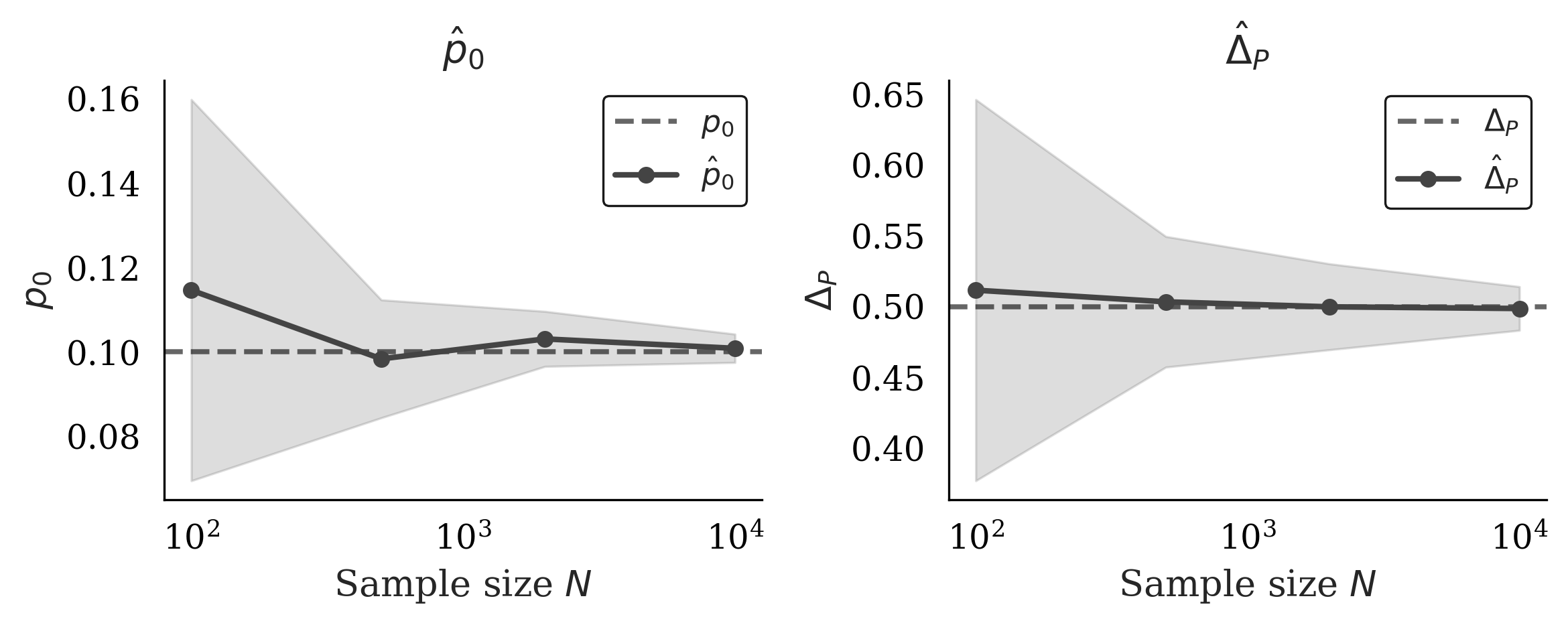}
  \caption{Convergence of the moment estimators $\hat p_0$ (left) and $\hat\Delta_P$ (right) as the size $N$ of the request log grows. Solid line: mean across $30$ replications; shaded band: $\pm 1$ standard deviation; dashed line: true value. Parameters: $p_0 = 0.1$, $\Delta_P = 0.5$, $r \sim 0.7\,\mathrm{Beta}(2,10) + 0.3\,\mathrm{Beta}(8,2)$.}
  \label{fig:param-estimates-convergence}
\end{figure}

\paragraph{Threshold and system efficacy.} \Cref{fig:convergence-estimated-params} translates parameter uncertainty into the downstream quantities of interest. The left panel shows that the plug-in threshold $\hat\tau^{*A}$ tracks the population optimum $\tau^{*A} = 0.729$: the bias falls below $3 \times 10^{-3}$ already at $N=500$, and the across-replication standard deviation falls from $0.047$ at $N=100$ to $0.002$ at $N=10{,}000$. The right panel reports the \emph{relative system efficacy} $\tilde W(\hat\tau^{*A})/\tilde W(\tau^{*A})$, i.e., the fluid objective at the plug-in threshold normalized by its value at the oracle threshold. The ratio reaches $0.992$ at $N=100$ and exceeds $0.9998$ for $N \ge 500$, indicating that the efficacy loss from plug-in estimation is at most around $1\%$ even at the smallest log size considered. The two panels together show that errors in $(\hat p_0,\hat\Delta_P)$ enter the objective only through a function that is locally flat near $\tau^{*A}$, so first-order errors in the parameters produce only second-order losses in system efficacy. The two-point characterization is therefore robust to plug-in estimation, and the prescriptive guidance in \Cref{thm:two_point_optimality} extends to settings where the planner must learn $(p_0,\Delta_P)$ from operational data.

\begin{figure}[ht]
  \centering
  \includegraphics[width=0.6\linewidth]{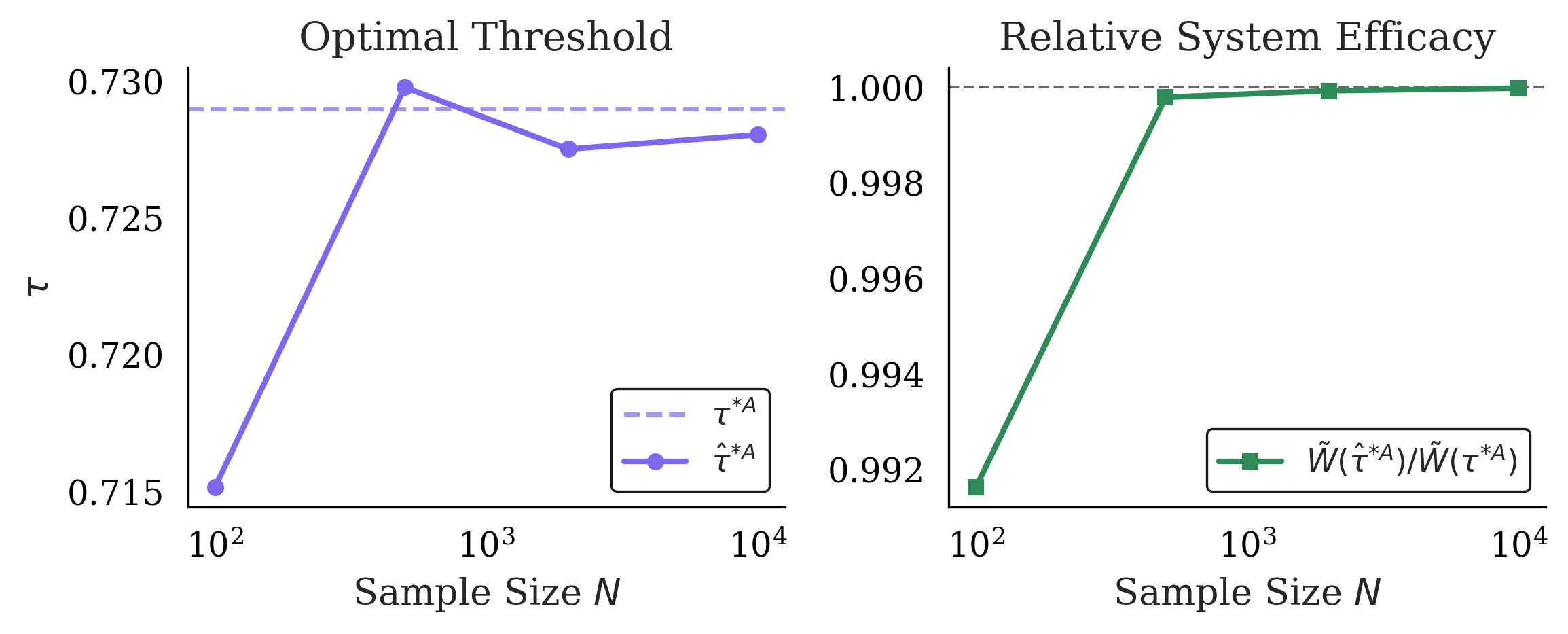}
  \caption{Convergence of the plug-in threshold (left) and the relative system efficacy (right) as $N$ grows. Left: mean of $\hat\tau^{*A}$ (solid) compared with the population optimum $\tau^{*A}$ (dashed). Right: ratio of the fluid objective at $\hat\tau^{*A}$ to its value at $\tau^{*A}$, which approaches $1$ from below. Parameters: $M/N = 0.2$; remaining parameters as in \Cref{fig:param-estimates-convergence}.}
  \label{fig:convergence-estimated-params}
\end{figure}

%% file: section/appendix_case_study.tex
\section{Case Study Details}\label{app:case_study}
\subsection{Additional Empirical Results}\label{app:case_study_extra}
The main-text case study compares Epic and an XGBoost model, where the higher-AUC model (Epic) is outperformed by the lower-AUC model over a range of low-capacity regimes. To extend this analysis, we train a second XGBoost model, XGBoost(2), on a different covariate set; we hereafter refer to the main-text XGBoost model as XGBoost(1). \Cref{fig:model_auc_3} summarizes the predictive performance of all three models: XGBoost(2) achieves the highest AUC of $0.854$, compared to $0.826$ for Epic and $0.806$ for XGBoost(1), with optimal score threshold $\tau^*_{\mathrm{score}} = 0.957$.

\begin{figure}[t]
  \centering
  \includegraphics[width=\textwidth]{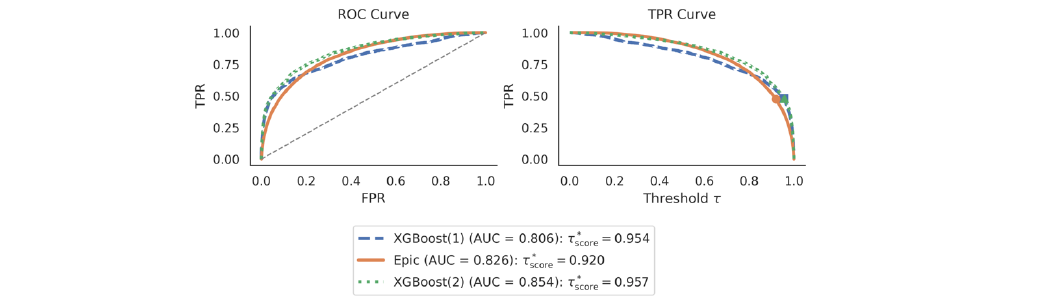}
  \caption{ROC curves (left) and TPR curves (right) for three prediction models: XGBoost(1) (AUC$=0.806$), Epic (AUC$=0.826$), and XGBoost(2) (AUC$=0.854$). Dots on the TPR curves mark the optimal score threshold $\tau^*_{\mathrm{score}}$ for each model (0.954, 0.920, and 0.957, respectively).}
  \label{fig:model_auc_3}
\end{figure}

\paragraph{Threshold Selection.}
Panels (a)--(c) of \Cref{fig:model_comparison_3} show that the main-text conclusion on threshold selection extends to all three models. Fixed-threshold policies are substantially suboptimal regardless of the underlying prediction model: the relative gap reaches up to 40\% at low capacity for $\tau = 0.6$ and up to 20\% at medium capacity for $\tau = 0.8$, while the capacity-matching policy maintains a near-zero gap across all capacity levels.

\paragraph{Algorithm Selection.}
Panel (d) of \Cref{fig:model_comparison_3} reports realized system efficacy as a function of capacity $M$ under optimal thresholding. Despite having the lowest AUC, XGBoost(1) outperforms Epic at low capacity ($M \in [10, 30]$), with a maximum improvement of 22.8\% and an average improvement of 13.1\%. XGBoost(2) outperforms Epic over a slightly wider range ($M \in [10, 40]$), with a maximum improvement of 20.7\% and an average of 12.8\%. At higher capacity levels ($M \gtrsim 40$), all three curves nearly overlap. Notably, XGBoost(1) and XGBoost(2) are statistically indistinguishable across all capacity levels: the higher AUC of XGBoost(2) does not translate into meaningfully better operational performance once both models are deployed at their optimal thresholds.

\begin{figure}
  \centering
  \includegraphics[width=\textwidth]{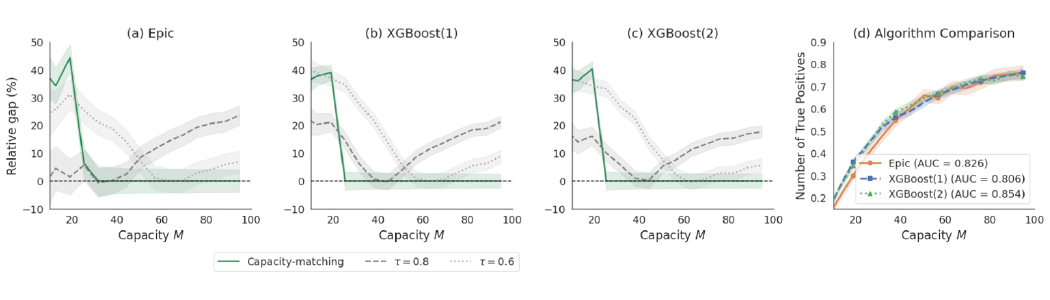}
  \caption{Threshold and algorithm comparison under varying capacity, extending \Cref{fig:sepsis} to include XGBoost(2). Panels (a)--(c) report relative gaps in system efficacy under different thresholding policies, compared to the efficacy-optimal threshold, for Epic, XGBoost(1), and XGBoost(2) respectively; the capacity-matching policy achieves near-zero gap across all capacities, while fixed thresholds incur up to 40\% relative loss at low capacity. Panel (d) reports the absolute number of true positives under optimal thresholding for all three models; despite having the highest overall AUC, XGBoost(2) does not outperform XGBoost(1) at any capacity level. Parameters: $N=200$, $p_0=0.1$, $\Delta_P=0.5$. Shaded regions indicate $\pm$ one standard error across trials.}
  \label{fig:model_comparison_3}
\end{figure}

\paragraph{Operational AUC.}
\Cref{fig:opauc_comparison_3models} compares the Operational AUC of the three models. Fixing the lower capacity limit at $M^L \in \{10, 40, 75\}$, we plot $\mathrm{OpAUC}(\mathrm{Uniform}(M^L/N, M^U/N))$ as a function of the upper bound $M^U$. Both XGBoost models achieve higher OpAUC than Epic across all values of $M^U$, while remaining nearly indistinguishable from each other. This is consistent with the pattern in Panel (d) of \Cref{fig:model_comparison_3}: a higher overall AUC does not guarantee superior operational performance, and the two XGBoost models are effectively interchangeable once deployed at their optimal thresholds.

\begin{figure}[t]
  \centering
  \includegraphics[width=\textwidth]{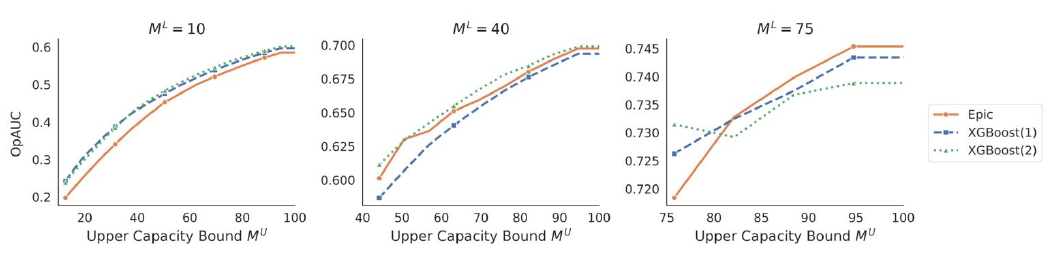}
  \caption{Operational AUC curves for the three models given $\mu = \mathrm{Uniform}(M^L/N, M^U/N)$. Despite XGBoost(2) having the highest overall AUC, it does not achieve the highest operational AUC. Parameters: $N=200$, $p_0=0.1$, $\Delta_P=0.5$.}
  \label{fig:opauc_comparison_3models}
\end{figure}

\subsection{XGBoost Model Training}\label{app:model}
\paragraph{Data Preparation}
Our patient pool consists of all patients with Epic sepsis scores available. We define sepsis using the CDC Adult Sepsis Event surveillance definition of sepsis.  Sepsis onset is defined as the first occurrence of CDC organ dysfunction AND suspicion for infection (defined as antibiotic or blood culture order), similar to Sepsis-3.

We construct covariates from both dynamic and static data sources. Dynamic covariates include vital signs, medication administration records, and laboratory analytes. We use the vital signs record time as the primary observation time; for each observation, we retain the most recent medication and analyte values. We augment these with time-elapsed features indicating hours since the last recorded value for each variable. If no prior medication or laboratory record exists, the corresponding features are left as missing (approximately 39\% of observations have prior medication records; all have prior analyte records). Static covariates include patient demographics (race, sex), encounter characteristics (age at encounter, admission source, admission type), and comorbidities. Categorical variables are converted to integer encodings. See \Cref{tab:covariates} for a summary of covariates used in the training of XGBoost model.

\begin{table}[t]
\centering
\small
\caption{Summary of covariates used in the XGBoost model.}
\label{tab:covariates}
\begin{tabular}{p{3cm} p{9cm}}
\toprule
\textbf{Variable Type} & \textbf{Variables} \\
\midrule
Demographics & Age, Gender, Race \\
\addlinespace
Encounter Information & Admission Source, Admission Type \\
\addlinespace
Comorbidities & Congestive Heart Failure, Valvular, Pulmonary circulation disorders, Peripheral vascular disease, Hypertension, Paralysis, Other neurological disorders, Chronic pulmonary disease, Diabetes without chronic complications, Diabetes with chronic complications, Hypothyroidism, Renal failure, Liver disease, Chronic peptic ulcer disease, HIV, Lymphoma, Metastatic cancer, Tumor, Rheumatoid arthritis, Coagulation deficiency, Obesity, Weight loss, Blood loss, Anemia, Alcohol abuse, Drug abuse, Psychoses, Depression, Leukemia \\
\addlinespace
Laboratory Values & Albumin, ALT, Ammonia, AST, Bandemia, Bicarbonate, Bilirubin, BUN, CK-MB, Creatine Kinase, CRP, D-Dimer, ESR, Fibrinogen, Glucose, Hematocrit, INR, Lactate, LDH, Magnesium, PCO2, pH, Platelets, PO2, Potassium, Serum Creatinine, Sodium, Troponin, WBC \\
\addlinespace
Vital Signs & BP, Pulse, Pulse Oximetry, Respiratory Rate, Temperature, Weight \\
\addlinespace
Medications & Antibiotics, Benzodiazepines, Chemotherapy, Heparins, Immunosuppressants, Insulins, IV Fluids, Opioids, Steroids, Vasopressors \\
\addlinespace
Time-Elapsed Features & Time-elapsed features indicating hours since the last recorded value for each variable. \\
\bottomrule
\end{tabular}
\end{table}

After data preparation, our data comprises 1,810,242 observations corresponding to 85,160 encounters and 70,067 patients.

\paragraph{Target Definition}
We define the prediction target as sepsis onset within a specified horizon. For each observation, we compute \texttt{hours\_to\_sepsis} as the difference between sepsis onset and the observation time. We then define \texttt{IS\_Sepsis\_8H} $= 1$ if sepsis onset occurs within the next 8 hours (i.e., $0 < \texttt{hours\_to\_sepsis} < 8$), and 0 otherwise. The 8-hour horizon yields a positive class prevalence of approximately 0.7\%.

\paragraph{Model Specification}
We split the data into training and test sets by patient ID to prevent data leakage across observations from the same patient, using an 80/20 train-test split. We train an XGBoost classifier to predict \texttt{IS\_Sepsis\_8H} using the covariates described above. The model is configured with the following hyperparameters: 300 estimators, maximum tree depth of 10, learning rate of 0.05, row subsampling rate of 0.8, and column subsampling rate of 0.8. We use the binary logistic objective and AUC as the evaluation metric.